\documentclass[aip,jmp,preprint,onecolumn,nofootinbib,numerical,longbibliography,citeautoscript]{revtex4-1}
\usepackage{graphicx,amsmath,amssymb,amsthm}
\usepackage{subfigure}
\usepackage{float}
\usepackage{bm,mathptmx}
\usepackage[bookmarks=true]{hyperref} 
\hypersetup{
    bookmarksopen=true,         
    bookmarksnumbered=true,     
    bookmarksopenlevel={2},       
    unicode=true,          
    pdftoolbar=true,        
    pdfmenubar=true,        
    pdffitwindow=false,     
    pdfstartview={FitH},    
    pdftitle={Random Matrix Models, Double-Time Painlevé Equations and Wireless Relaying},    
    pdfauthor={Yang Chen, Nazmus S. Haq and Matthew R. McKay},     
    pdfdisplaydoctitle=true, 
    pdfsubject={Mathematics},   
    pdfcreator={Nazmus S. Haq},   
    pdfproducer={BaKoma Tex}, 
    pdfkeywords={Orthogonal polynomials} {Random Matrix Theory} {Painlevé Equations} {MIMO Systems}, 
    colorlinks=true,       
    linkcolor=blue,          
    citecolor=blue,        
    filecolor=magenta,      
    urlcolor=blue,          
    final=true,
}
\usepackage[sort&compress]{natbib}
\newtheorem{theorem}{Theorem}

\newtheorem{lemma}{Lemma}

\newtheorem{corollary}{Corollary}

\newtheorem{remark}{Remark}

\newcommand{\al}{\alpha}

\newcommand{\be}{\begin{equation}}
\newcommand{\ee}{\end{equation}}
\newcommand{\bq}{\begin{eqnarray}}
\newcommand{\eq}{\end{eqnarray}}

\newcommand{\ba}{\begin{array}}
\newcommand{\ea}{\end{array}}

\newcommand{\bt}{\beta}

\newcommand{\ga}{\gamma}

\renewcommand{\v}{{\mathsf{v}}}

\newcommand{\sig}{\sigma}
\def\beq#1#2\eeq{
        \begin{equation}
        \label{#1}
            #2
        \end{equation}}
\newcommand{\bea}{\begin{eqnarray}}
\newcommand{\eea}{\end{eqnarray}}
\renewcommand{\P}{\textsf{p}}
\newcommand{\ti}[1]{\tilde{#1}}
\newcommand{\mb}[1]{\mathbf{#1}}
\newcommand{\bg}{\bar{\gamma}}
\newcommand{\N}{N_s}

\begin{document}

\preprint{\scriptsize \href{http://dx.doi.org/10.1063/1.4808081 }{J. Math. Phys.} {\bf 54}, 063506 (2013)}

\title{Random Matrix Models, Double-Time Painlev\'e Equations, and Wireless Relaying}
\author{\firstname{Yang} \surname{Chen}}
\email{yangbrookchen@yahoo.co.uk, yayangchen@umac.mo}
\affiliation{\footnotesize Faculty of Science and Technology, Department of Mathematics, University of Macau, Av. Padre Tom\'as Pereira, Taipa Macau, China}
\author{\firstname{Nazmus} S. \surname{Haq}}
\email{nazmus.haq04@imperial.ac.uk}
\thanks{Author to whom correspondence should be addressed.}
\affiliation{\footnotesize Department of Mathematics, Imperial College London,
180 Queen's Gate, London SW7 2BZ, UK}
\author{\firstname{Matthew} R. \surname{McKay}}
\email{eemckay@ust.hk}
\affiliation{\footnotesize Department of Electronic and Computer Engineering, Hong Kong
University of Science and Technology (HKUST), Clear Water Bay, Kowloon, Hong Kong}
\date{22 December 2012; Revised: 26 May 2013}
\begin{abstract}
This paper gives an in-depth study of a multiple-antenna wireless
communication scenario in which a weak signal received at an
intermediate relay station is amplified and then forwarded to the
final destination. The key quantity determining system performance
is the statistical properties of the signal-to-noise ratio (SNR)
$\gamma$ at the destination. Under certain assumptions on the encoding
structure, recent work has characterized the SNR distribution
through its moment generating function, in terms of a certain Hankel
determinant generated via a deformed Laguerre weight. Here, we
employ two different methods to describe the Hankel determinant.
First, we make use of ladder operators satisfied by orthogonal
polynomials to give an exact characterization in terms of a ``double-time'' Painlev\'e differential
equation, which reduces to Painlev\'e V under certain limits.
Second, we employ Dyson's Coulomb Fluid method to derive a closed form
approximation for the Hankel determinant. The two
characterizations are used to derive closed-form expressions for the
cumulants of $\ga$, and to compute performance quantities of
engineering interest.
\end{abstract}
\maketitle






\section{Introduction}

Over the past decade, multiple-input multiple-output (MIMO) antenna
systems employing space-time coding have revolutionized the wireless
industry. Such systems are well-known to offer substantial benefits
in terms of channel capacity, as well as improved diversity and link
reliability, and as such these techniques are being incorporated
into a range of emerging industry standards. More recently, relaying
strategies have been also proposed as a means of further improving
the performance of space-time coded MIMO networks. Such methods are
particularly effective for improving the data transmission quality
of users who are located near the periphery of a communication cell.

Various MIMO relaying strategies have been proposed in the
literature; see, e.g.,
Refs.~\onlinecite{Sendonaris2003_1,Sendonaris2003_2,Laneman2003,Laneman2004,Nabar2004,Fan2007},
offering different trade-offs in various factors such as
performance, complexity, channel estimation requirements, and
feedback requirements.  In Ref.~\onlinecite{DharMckayMallik2010}, a
communication technique was considered which employed a form of
orthogonal space-time block coding (OSTBC), along with non-coherent
amplify-and-forward (AF) processing at the multi-antenna relay. This
method has the advantage of operating with only low complexity,
requiring only linear processing at all terminals, achieving high
diversity, and not requiring any short-term channel information at
either the relays or the transmitter. For this system, an important
performance measure---the so-called \emph{outage
probability}---which gives a fundamental probabilistic performance
measure over random communication channels, was considered in Refs.~\onlinecite{DharMckayMallik2010} and \onlinecite{SongShin2009}, based on
deriving an expression for the moment generating function of the
received signal to noise ratio (SNR). The outage probability could
then be computed by performing a subsequent numerical Laplace
Transform inversion. The exact moment generating function results
presented in Refs.~\onlinecite{DharMckayMallik2010} and \onlinecite{SongShin2009},
however, are quite complicated functions involving a particular
\emph{Hankel} determinant, and they yield little insight. Moreover,
for all but small system configurations, the expressions are
somewhat difficult to compute, particularly when implementing the
Laplace inversion.

In this paper, we significantly expand and elaborate upon the
results of Refs.~\onlinecite{DharMckayMallik2010} and \onlinecite{SongShin2009} by
employing analytical tools from random matrix theory and statistical
physics. Technically, the challenge is to appropriately characterize
the Hankel determinant which arises in the moment generating
function computation. In general, the Hankel determinant takes the
form,
\bq\label{def:HankelDn} D_n = \det\left( \mu_{i+j}
\right)_{i,j=0}^{n-1},
\eq
generated from the moments of a certain
weight function $w(x)$, $0\leq x<\infty,$
\bea
\mu_{k} := \int\limits_0^\infty x^{k}
w(x) dx \; , \hspace*{1cm} k = 0, 1, 2, \dots.
\nonumber
\eea
Note that the above is merely one of several possible representations for
$D_n$, with equivalent forms regularly found in the fields of
statistical physics and random matrix theory.

For the problem at hand, we are faced with a Hankel determinant that
is generated from a weight function of the form,
\begin{align} \label{eq:wAF}
w_{{\rm AF}}(x,T,t) =  x^\al e^{-x}\left(\frac{t+x}{T+x}\right)^{\N}, \qquad
0\leq x<\infty,
\end{align}
which is a ``two-time'' deformation of the classical Laguerre weight, $x^\al e^{-x}$.
The parameters in this weight, satisfy the following conditions:
\bea
\al>-1, \; \; T := \frac{t}{1 + c s} , \; \; t
> 0,\; \; c>0, \; \; {\N}>0,\; \; 0\leq s <\infty.
\eea For our problem both $\al$ and $\N$ are integers, where
$\al:=|N_R-N_D|$ is the absolute difference between the number of
relay and destination antennas, $N_R$ and $N_D$
respectively.
The parameters $\alpha$ and $N_s$ are determined by the model, but in fact can be extended to take non-integer values such that $\alpha >-1$ and $N_s>0$.\cite{Szego1939}

We will apply two different methods to characterize this Hankel
determinant. For the first method, we will derive new exact
expressions for $D_n[w_{\rm AF}]$ by employing the theory of
orthogonal polynomials associated with the weight $w_{{\rm
AF}}(x,T,t)$ and their corresponding ladder operators. For such
methods, extensive literature exists (see, e.g.,
Refs.~\onlinecite{TracyWidom1999,Magnus1995,ChenPruessner2005,ChenMckay2010,ChenIts2009,ChenIsmail1997,ChenIsmail2005,ChenFeigin2006,BasorChenZhang,BasorChenEhrhardt,BasorChenMekar2012}
for their use in applications involving unitary matrix ensembles),
though in the context of information theory and communications the
techniques have only very recently been introduced by the first and
third authors in Ref.~\onlinecite{ChenMckay2010}. For the second method, we derive an approximation for the Hankel determinant using the
general linear statistics theorem obtained in
Ref.~\onlinecite{ChenLawrence1997}, based on Dyson's Coulomb Fluid models.
\cite{ChenIsmail1997,Dyson1962I,Dyson1962II,Dyson1962III,ChenManning1994}
These results are essentially the Hankel analogue of asymptotic
results for Toeplitz determinants. \cite{Szego1919} Once again,
these techniques have been used extensively, particularly in
statistical physics,
 though only very recently
have they been applied to address problems in communications and
information theory.\cite{Kazakopoulos2010,ChenMckay2010}

It is important to note that in addition to the two methodologies
advocated above, there exists other integrable systems approaches
which can be used for characterizing Hankel determinants. For
example, $D_n$ may be written in an equivalent matrix integral
formulation, from which a `deform-and-study' or isomonodromic
deformation approach may be adopted. Essentially, this idea involves
embedding $D_n$ into a more general theory of the $\tau$-function,
\cite{JimboMiwaUeno1981vI,JimboMiwa1981vII} using bilinear
identities and linear Virasoro constraints. For details, see
Refs.~\onlinecite{AdlervanMoerbeke1995,AdlervanMoerbeke2001,AdlerShiotavanMoerbeke1995}.
Yet another approach is to characterize $D_n$ through the use of
Fredholm determinants as employed by Tracy and Widom.
\cite{TracyWidom1999} Both these exact, non-perturbative integrable
systems methods have their own specific advantages and
disadvantages, and in most cases lead to non-linear ordinary differential or partial differential equations (ODE/PDEs) satisfied by the Hankel determinant.
However, these equations are usually of higher order (equations of Chazy type usually appear), from which first integrals have to be found
to reduce to a second order ODE. The advantage of
the ladder operator method is that closed-form second order
equations are directly obtained for a quantity related to the Hankel
determinant, bypassing the need to find a first integral, or
evaluate any multiple integral.

The rest of this paper is organized as follows. In Sections \ref{SubSec:AnF_Model} and \ref{SubSec:AnF_Perf_Mesures}, we present a brief discussion of the model
which underpins the MIMO-AF wireless communication system of interest, and some basic measures which are employed to quantify
system performance. Then, in Sections \ref{SubSec:Stat_Char_ga} and \ref{SubSec:Alt_Char_ga}, we go on to introduce the moment generating function and cumulants of
interest, and pose the key mathematical problems to be dealt with in
the remainder of the paper.

In Section \ref{Sec:P_Char}, we establish an exact finite $n$ characterization of the Hankel
determinant $D_n$, employing the theory of orthogonal polynomials and their ladder
operators. 

In Section \ref{Sec:CF}, we introduce the Coulomb Fluid method, where we first give a
brief overview of its key elements,
following
Refs.~\onlinecite{ChenManning1994},~\onlinecite{ChenIsmail1997T},~\onlinecite{ChenLawrence1997,ChenMckay2010}.
In Section \ref{SubSec:CF_AnF_Calc}, we compute an asymptotic (large $n$) approximation
for the Hankel determinant $D_n$, and thus a corresponding
characterization for the moment generating function of interest. In
Section \ref{SubSec:SER_Analysis}, the Coulomb Fluid representation for the moment
generating function of the received SNR is shown to yield extremely
accurate approximations for the error performance of OSTBC MIMO systems with AF relaying, even when the
system dimensions are particularly small.

We also employ our
analytical results to compute closed-form expressions for the
cumulants of the received SNR, first via the Coulomb Fluid method in Section \ref{eq:SubSec:Cumulants},
and then presenting a refined analysis based on Painlev\'e
equations in Sections \ref{Sec:Cumulant_Painleve_Analysis} and \ref{Sec:Large_n_Corr}.

Subsequently, we give an asymptotic characterization of
the moment generating function, valid for scenarios for which the
average received SNR is high, deriving key quantities of interest to
communication engineers, including the so-called diversity order and
array gain.  These results are, once again, established via the
Coulomb Fluid approximation in Section \ref{Sec:CF_Large_s}, and subsequently validated with the
help of a Painlev\'e characterization in Section \ref{Sec:PV_Large_s}.

\subsection{Amplify and Forward Wireless Relay Model}\label{SubSec:AnF_Model}
Here we briefly recall the background for the model, having been
developed in Ref.~\onlinecite{DharMckayMallik2010}. The dual hop MIMO
communication system features a source, relay and destination
terminal, having $\N$, $N_R$ and $N_D$ antennas respectively. In the
process of transmitting a signal, each transmission period is
divided into two time slots. In the first time slot, the source
transmits to the relay. The relay then amplifies its received signal
subject to an average power constraint, prior to transmitting the
amplified signal to the destination terminal during the second
time-slot. We assume that the source and destination terminals are
sufficiently separated such that the direct link between them is
negligible (i.e., all communication is done via the relay).

Let $\mathbf{H_1}\in\mathbb{C}^{N_R\times {\N}}$ and
$\mathbf{H_2}\in\mathbb{C}^{N_D\times N_R}$ represent the channel
matrices between the source and relay, and the relay and destination
terminals respectively. Each channel matrix is assumed to have
uncorrelated elements distributed as
$\mathcal{CN}(0,1)$.\footnote{The notation $\protect \mathcal  {CN}(\mu ,\sigma ^2)$ represents a complex Gaussian distribution with mean $\mu $ and variance $\sigma ^2$.} The destination is assumed to have perfect
knowledge of $\mathbf{H_2}$ and either $\mathbf{H_1}$ or the
cascaded channel $\mathbf{H_2 H_1}$, while the relay and source
terminals have no knowledge of these.

A situation is considered where the source terminal transmits data
using a method called OSTBC encoding.~\cite{Jafarkhani2005} In this
situation, groups of independent and identically distributed complex
Gaussian random variables (referred to as information ``symbols'')
$s_i$, $i=1,\dots,N$ are assigned via a special codeword mapping to
a row orthogonal matrix $\mathbf{X} =
(\mathbf{x}_1,\dots,\mathbf{x}_{N_P}) \in\mathbb{C}^{\N\times N_P}$
satisfying the power constraint $ E \left[ \| \mathbf{x}_k \|^2
\right] = \bar{\gamma}$, where $N_P$ is the number of symbol periods
used to send each codeword.

Since it takes $N_P$ symbol periods to transmit $N$ symbols, the
coding rate is then defined as
\begin{align}
R=\frac{N}{N_P}.
\end{align}
The received signal matrix at the relay terminal at the end of the
first time slot,  $\mathbf{Y} =
(\mathbf{y}_1,\dots,\mathbf{y}_{N_P}) \in\mathbb{C}^{N_R\times
N_P}$, is given by
\begin{align}
\mathbf{Y} = \mathbf{H}_1 \mathbf{X} + \mathbf{N} \, ,
\end{align}
where $\mathbf{N} \in\mathbb{C}^{N_R\times N_P}$ has uncorrelated
entries distributed as $\mathcal{CN}(0,1)$, representing normalized
noise samples at the relay terminal. The relay then amplifies the
signal it has received by a constant gain matrix
$\mathbf{G}=\ti{a}\mathbf{I}_{N_R}$, where \bea\label{defn:ti(a2)}
\ti{a}^2=\frac{\ti{b}}{(1+\bg)N_R}. \eea In the above, $\ti{b}$ is
the total power constraint imposed at the relay, i.e., $ E \left[ \|
\mathbf{Gy}_k \|^2 \right] \leq\ti{b}$, whilst $\bg$ represents the
average received SNR at the relay. The received signal $\mathbf{R}
\in\mathbb{C}^{N_D \times N_P}$ at the destination terminal at the
end of the second time slot is then given by
\begin{align}
\mathbf{R} \; = \; \ti{a}\mathbf{H}_2 \mathbf{Y} + \mathbf{W} \; =
\; \ti{a}\mb{H}_2\mb{H}_1 \mb{X} + \ti{a}\mb{H}_2 \mb{N} + \mb{W}
,\qquad \label{eq:recModel}
\end{align}
where $\mathbf{W} \in\mathbb{C}^{N_D \times N_P}$ has uncorrelated
entries distributed as $\mathcal{CN}(0,1)$, representing normalized
noise samples at the destination terminal.

Next, the receiver applies the linear (noise whitening) operation,
\begin{align}
\ti{\mb{R}}= \mb{K}^{-1/2} \mb{R} , \qquad \mb{K} :=
\ti{a}^2\mb{H}_2\mb{H}_2^\dag+\mb{I}_{N_D},
\end{align}
to yield the equivalent input-output model
\begin{align}
\label{eq:equivModel} \ti{\mb{R}}=\ti{\mb{H}}\mb{X}+\ti{\mb{N}},
\end{align}
where $\ti{\mb{H}}=\ti{a}\mb{K}^{-1/2}\mb{H}_2\mb{H}_1$ and
$\ti{\mb{N}} \in \mathbb{C}^{N_D \times N_P}$ has uncorrelated
entries distributed as $\mathcal{CN}(0,1)$.




%

Based on the above relationship, standard linear OSTBC decoding can
be applied (see, e.g., Ref.~\onlinecite{Jafarkhani2005}). This results in
decomposing the matrix model (\ref{eq:equivModel}) into a set of
parallel non-interacting single-input single-output relationships
given by
$$
\ti{s}_l=\vert\vert\ti{\mb{H}}\vert\vert_Fs_l+\eta_l,\qquad
l=1,\dots,N,
$$
where $\eta_l$ is distributed as $\mathcal{CN}\left(0, 1 \right)$,
with $\vert\vert\ti{\mb{H}}\vert\vert_F$ representing the Frobenius
norm (or matrix norm) of $\ti{\mb{H}}$.

The quantity that we are interested in, the instantaneous SNR for
the $l$th symbol $\gamma_l$, can then be written as
\begin{align} \label{eq:SNR_AF_Def}
\gamma_l \, = \, \vert\vert\ti{\mb{H}}\vert\vert^2_F E \left[ \vert
s_l\vert^2 \right] \, = \,
\frac{\bg\ti{b}}{R{\N}(1+\bg)N_R}\operatorname{Tr}\left(\mb{H}_1^\dag\mb{H}_2^\dag\mb{K}^{-1}\mb{H}_2\mb{H}_1\right).
\end{align}
Since the right-hand side is independent of $l$, we may drop the $l$
subscript and denote the instantaneous SNR as $\gamma$ without loss
of generality.

\subsection{Wireless Communication Performance Measures}\label{SubSec:AnF_Perf_Mesures}

Here we recall some basic measures which are employed to quantify
the performance of wireless communication systems.  One of the most
common measures is the so-called symbol error rate (SER), which
quantifies the rate in which the transmitted symbols (or signals)
are detected incorrectly at the receiver.

For signals which are designed using standard {\em $M$-ary phase
shift keying (MPSK)} digital modulation formats, the phase of a
transmitted signal is varied to convey information, where $M \, \in
\{ 2, 4, 8, 16, \ldots \}$ represents the number of possible signal
phases. The SER can be expressed as follows \cite{Simon2005}
\begin{align} \label{eq:SERExact}
P_{\rm MPSK} = \frac{1}{\pi} \int\limits_0^\Theta {\cal M}_\gamma \left(
\frac{g_{\rm MPSK}}{\sin^2 \theta } \right) d \theta,
\end{align}
where ${\cal M}_\gamma ( \cdot )$ is the moment generating function
of the instantaneous SNR $\gamma$ at the receiver, and $\Theta = \pi
(M-1) / M$ and $g_{\rm MPSK} = \sin^2(\pi/M)$ are
modulation-specific constants. As also presented in
Ref.~\onlinecite{MckayZanella2009}, the SER may be approximated in terms of
${\cal M}_\gamma ( \cdot )$ but without the integral, via the
following expression:
\begin{align} \label{eq:SERApprox}
P_{\rm MPSK} \approx \left( \frac{\Theta}{2 \pi} - \frac{1}{6}
\right) {\cal M}_\gamma (g_{\rm MPSK}) + \frac{1}{4} {\cal M}_\gamma
\left( \frac{ 4 g_{\rm MPSK}}{3} \right) + \left( \frac{\Theta}{2
\pi} - \frac{1}{4} \right) {\cal M}_\gamma \left(\frac{g_{\rm
MPSK}}{\sin^2 \Theta}\right)  \; .
\end{align}

In addition to the SER, another useful quantity is the so-called
amount of fading (AoF), which serves to quantify the degree of
fading (i.e., the level of randomness) in the wireless channel, and
is expressed directly in terms of the \emph{cumulants} of $\gamma$.
Specifically, the AoF is defined as follows:
\begin{align}
{\rm AoF} = \kappa_2 / \kappa^2_1,
\end{align}
where $\kappa_1$ and $\kappa_2$ are the mean and variance of
$\gamma$ respectively.  As shown in Ref.~\onlinecite{Karagiannidis2004}, for
example, this quantity can be used to describe the achievable
capacity of the wireless link when the average SNR is low.


For our AF relaying system under consideration, it is clear that a
major challenge is to characterize the moment generating function
and also the cumulants of the instantaneous SNR given in
(\ref{eq:SNR_AF_Def}). This is the key focus of the paper.

\subsection{Statistical Characterization of the SNR $\ga$}\label{SubSec:Stat_Char_ga}
Here we introduce the moment generating function and cumulants of
interest, and pose the key mathematical problems to be dealt with in
the remainder of the paper.

Defining the positive integers,
\bea
m:={\rm max}(N_R,N_D),&\qquad&n={\rm min}(N_R,N_D),
\eea where $n\neq 0$ due to physical considerations, it was shown in Ref.~\onlinecite{DharMckayMallik2010}
that the moment generating function of the instantaneous SNR (\ref{eq:SNR_AF_Def}) may be
expressed as the multiple integral,
\begin{align}
\label{eq:MgfDharMckay}
\mathcal{M}_\gamma(s)&=
\frac{\frac{1}{n!}\int\limits_{[0,\infty)^n}\prod\limits_{1\leq
i<j\leq
n}(x_j-x_i)^2\prod\limits_{k=1}^{n}x_k^{m-n}e^{-x_k}\left(\frac{1+\ti{a}^2x_k}{1+\ti{a}^2(1+\frac{\bg
s}{R{\N} })x_k}\right)^{\N}dx_k}
{\frac{1}{n!}\int\limits_{[0,\infty)^n}\prod\limits_{1\leq
i<j\leq
n}(x_j-x_i)^2\prod\limits_{k=1}^{n}x_k^{m-n}e^{-x_k}dx_k}.
\end{align}
From the above, we may then compute the $l^{\rm th}$ cumulant of
$\ga$ via
\begin{align}
\label{def:kappa_l} \kappa_l=\:(-1)^l\frac{d^l}{ds^l}\:\log{\cal
M}_\ga(s)\Bigg|_{s=0}.
\end{align}
Let
\begin{align}
\al &:= m-n,\\
\label{defn:t}{t} &:= \frac{1}{\ti{a}^2},\\
\label{defn:c}c&:= \frac{\bg}{R{\N}},\\
\label{tf:T(s)}T=T(s)&:= \frac{{t}}{\left(1+cs\right)}.
\end{align}
We now consider the moment generating function
(\ref{eq:MgfDharMckay}) as a function of two variables, $(T,t)$, or
$(s,t)$, since $T=t/(1+cs)$. We may then write the moment generating
function (\ref{eq:MgfDharMckay}) as
\begin{align}
\label{eq:mgfMultipleIntegral}
\mathcal{M}_\gamma(T,{t})=\left(\frac{T}{{t}}\right)^{n{\N}}
\frac{\frac{1}{n!}\int\limits_{[0,\infty)^n}\prod\limits_{1\leq
i<j\leq n}(x_j-x_i)^2\prod\limits_{k=1}^n x_k^\al
e^{-x_k}\left(\frac{{t}+x_k}{T+x_k}\right)^{\N} dx_k}
{\frac{1}{n!}\int\limits_{[0,\infty)^n}\prod\limits_{1\leq i<j\leq
n}(x_j-x_i)^2\prod\limits_{k=1}^n x_k^\al e^{-x_k} dx_k}.
\end{align}
We see that this involves the weight with the parameters $T$ and $t$
\begin{align}
\label{eq:w(x,T,t)} w_{{\rm AF}}(x,T,{t}) &= x^\al
e^{-x}\left(\frac{{t}+x}{T+x}\right)^{\N},
\end{align}
which is a deformation of the classical generalized Laguerre weight
\begin{align}
w_{{\rm Lag}}^{(\al)}(x)&= x^\al e^{-x},\qquad 0\leq x<\infty.
\end{align}
Note that $w_{{\rm AF}}(x,t,t)=w_{{\rm Lag}}^{(\al)}(x).$

The multiple integral representation (\ref{eq:mgfMultipleIntegral})
is expressed as a ratio of Hankel determinants:
\begin{align}
\label{eq:Mgf(Dn)}
\mathcal{M}_\ga(T,{t}) &=
\left(\frac{T}{{t}}\right)^{n{\N}}\frac{D_n[w_{{\rm AF}}(\cdot,T,{t})]}{D_n[w_{{\rm Lag}}^{(\al)}(\cdot)]},\\
&=
\left(\frac{T}{{t}}\right)^{n{\N}}\frac{\det\Big(\mu_{i+j}(T,t)\Big)_{i,j=0}^{n-1}}{\det\Big(\mu_{i+j}(t,t)\Big)_{i,j=0}^{n-1}},
\; \label{eq:Mgf(detmu)}
\end{align}
where $\mu_{j}(T,t)$ is the $j$th moment of the weight
\begin{align}
\mu_j(T,t)=&\int\limits^\infty_0 \! x^jw_{{\rm AF}}(x,T,t) \ dx, \qquad j=0,1,2,\dots\nonumber
\end{align}
For $\N\in \mathbb{N}$, the moments of the weight $w_{{\rm AF}}$ are
expressed in terms of the Kummer function of the second kind,
\begin{align}
\label{eq:Moments(Kummer)} \mu_j(T,{t}) &=
t^{\al+j+1}\Gamma(\al+j+1)\sum\limits_{k=0}^{\N}\binom{{\N}}{k}\left(\frac{{t}-T}{T}\right)^kU(\al+j+1,\al+j+2-k,T).
\end{align}
The Hankel determinant for $T=t$, namely,
$$
D_n[w_{{\rm AF}}(\cdot,t,t)]=D_n[w_{{\rm Lag}}^{(\al)}(\cdot)],
$$
can be regarded as a normalization constant, so that
$\mathcal{M}_\ga(t,t)=1$, and its closed form expression is well-known (see Ref.~\onlinecite{Mehta2004}).
\begin{align}
D_n[w_{{\rm AF}}(\cdot,t,{t})]
&=\prod\limits_{i=0}^{n-1}\Gamma(\al+i+1)\Gamma(i+1), \nonumber\\
\label{eq:NormConst_n_al}&=\frac{G(n+1)G(n+\al+1)}{G(\al+1)},
\end{align}
where $G(z)$ is the Barnes G-function defined by
$G(z+1)=\Gamma(z)G(z)$ with $G(1)=1$.

From a simple change of variables, the cumulants can be calculated by transforming (\ref{def:kappa_l}) using (\ref{tf:T(s)}) as
\begin{align}
\label{eq:kappa_l(T)} \kappa_l =
c^l\left(\frac{T^2}{{t}}\frac{d}{dT}\right)^l\log
\mathcal{M}_\ga(T,{t})\bigg\vert_{T={t}}.
\end{align}

Note that an equivalent expression to the result given in (\ref{eq:Moments(Kummer)}) has been
previously reported in Ref.~\onlinecite{DharMckayMallik2010}.
This result,
whilst having numerical computation advantages when the system
dimensions are small, has a number of shortcomings. Most
importantly, its complexity does not lead to useful insight into the
characteristics of the probability distribution of the SNR $\gamma$,
and it does not clearly reveal the dependence on the system
parameters.  Moreover, in this form, the expression is not amenable
to asymptotic analysis, e.g., as the number of antennas grow
sufficiently large, and in such cases its numerical computation
becomes unwieldy.

In the following, we will seek alternative
simplified representations with the aim of overcoming these
shortcomings.  The key challenge is to characterize the Hankel
determinant in (\ref{eq:Mgf(Dn)}).
\begin{remark}
It is of interest to mention here that for the special case $\alpha
= 0$, the Hankel determinant (\ref{eq:Mgf(Dn)}) is directly related
to the shot-noise moment generating function of a disordered
quantum conductor. This was investigated in Ref.~\onlinecite{MuttalibChen2008} using the Coulomb Fluid method.
Specifically, the two moment generating
functions are related by
\bea
\mathcal{M}_\ga(T,t)|_{\alpha=0}&=&\frac{\mathcal{M}_{\text{Shot-Noise}}(T,z)}{z^{n\N}},
\eea
where $z:=t/T$ and $\mathcal{M}_{\text{Shot-Noise}}(T,z)$ is the
Hankel determinant generated via
\bea
w(x,T,z)&=&e^{-Tx}\left(\frac{z+x}{1+x}\right)^{\N},\qquad
0\leq x<\infty.
\eea
Of course, $T$, $z$ and $\N$ have different
interpretations in this situation.

 We mention here
that the operator theory approach of Ref.~\onlinecite{BasorChenWidomJFA} justifies the Coulomb Fluid results obtained
in Ref.~\onlinecite{BasorChenWidomMSRI} on the shot-noise problem.
\end{remark}

\subsection{Alternative Characterization of the SNR $\ga$}\label{SubSec:Alt_Char_ga}
An alternative characterization for the moment generating function that will also prove to be
useful is derived as follows. From the definitions of  $m$ and $n$,
there are two possible choices for the parameter $t$,
\bea\label{eq:CF:t(n)Dep}
t=\frac{1}{\ti{a}^2}=\frac{(1+\bg)N_R}{\ti{b}}&=
\begin{cases}
\frac{(1+\bg)n}{\ti{b}} & N_R\leq N_D,\\
\frac{(1+\bg)nr}{\ti{b}},\qquad \text{where}\qquad r:=m/n, \qquad &N_R>N_D.
\end{cases}
\nonumber\\
\eea

We first consider the sub-case $N_R\leq N_D$. To this end, starting
with (\ref{eq:mgfMultipleIntegral}) and with the change of
variables, $x_i\to n x_i $, $i=1,\dots,n$, we obtain\footnote{For
the sake of brevity, instead of defining a new function, we write
$\mathcal{M}_\ga(T^\prime,{t}^\prime)$ in place of
$\mathcal{M}_\ga\left(nT^\prime,nt^\prime\right)$. }
\begin{scriptsize}
\bea\label{eq:CF:MGF_Scaled}
\mathcal{M}_\ga(T^\prime,t^\prime)&=&
\left(\frac{T^\prime}{{t^\prime}}\right)^{n{\N}}
\frac{\frac{1}{n!}\int\limits_{[0,\infty)^n}\prod\limits_{l=1}^ndx_l\exp\Bigg[2\sum\limits_{1\leq j<k\leq n}\log\vert x_j-x_k\vert-n\sum\limits_{j=1}^n\left(x_j-\beta\log x_j\right)-\sum\limits_{j=1}^n{\N}\log\left(\frac{T^\prime+x_j}{{t}^\prime+x_j}\right)\Bigg]}
{\frac{1}{n!}\int\limits_{[0,\infty)^n}\prod\limits_{l=1}^ndx_l\exp\Bigg[2\sum\limits_{1\leq j<k\leq n}\log\vert x_j-x_k\vert-n\sum\limits_{j=1}^n\left(x_j-\beta\log x_j\right)\Bigg]},\nonumber\\
\eea
\end{scriptsize}
where
\bea
\label{def:tprime}t^\prime&:=&\frac{t}{n},\\
\label{def:Tprime}T^\prime&:=&\frac{T}{n},\\
\label{def:CF:beta}\beta&:=&\frac{m}{n}-1.
\eea
Note that in terms of $s$, $T^\prime$ can be written as
\bea
T^\prime&=&\frac{t^\prime}{1+cs}.
\eea
Equivalently, we can write (\ref{eq:CF:MGF_Scaled}) as
\begin{equation}\label{def:CF:MGF_Zn/Z0}
\mathcal{M}_\ga(T^\prime,t^\prime)=\left(\frac{T^\prime}{{t^\prime}}\right)^{n{\N}}\frac{Z_n(T^\prime,t^\prime)}{Z_n({t}^\prime,{t}^\prime)},
\end{equation}
where
\begin{small}
\bea\label{def:CF:ZnCFM} 
Z_n (T^\prime,t^\prime) &:=& D_n\left(nT^\prime,nt^\prime\right)
\nonumber\\
&=&\frac{1}{n!}\int\limits_{[0,\infty)^n}
\exp \left[ - \Phi (x_1, \ldots, x_n ) - \sum_{j=1}^n {\N}\log\left(\frac{T^\prime+x_j}{{t}^\prime+x_j}\right)
\right] \prod_{l=1}^n d x_l,\;
\eea
\end{small}
and
\begin{align}
\label{eq:CF:PhiDefinition_Intro}
 \Phi(x_1, \ldots, x_n) := -2 \sum_{1
\leq j < k \leq n} \log | x_j - x_k | + n \sum_{j=1}^n \left(x_j-\beta\log x_j\right) \;.
\end{align}
\begin{remark}
We introduce the variables $T^\prime$ and $t^\prime$ in order to account for the $n$-dependence of the variables $T$ and $t$.
 This is important since the above representation will be useful for deriving a large $n$ approximation
for the moment generating function based on the
Coulomb Fluid linear statistics approach in Section \ref{Sec:CF}.
\end{remark}

For the sub-case $N_R>N_D$, we have from (\ref{eq:CF:t(n)Dep}) that
$t$ is instead given by \bea t&=&\frac{(1+\bg)(1+\bt)n}{\ti{b}},
\eea where $\bt=\frac{m}{n}-1.$ Hence equations
(\ref{def:CF:MGF_Zn/Z0})--(\ref{eq:CF:PhiDefinition_Intro}) and the
results of Sections \ref{Sec:CF}-\ref{Sec:Large_n_Corr} are valid
for the case $N_R>N_D$ upon transforming $t^\prime$ and $T^\prime$
to $t^\dag$ and $T^\dag$ via \bea\label{def:tdag}
t^\prime&=:&(1+\bt) t^\dag \eea and \bea\label{def:Tdag}
T^\prime&=:&(1+\bt) T^\dag \eea respectively. In this case, we may
write $T^\dag$ in terms of the variable $s$ as \bea T^\dag&=&
\frac{t^\dag}{1+cs}. \eea
\section{Painlev\'e Characterization via the Ladder Operator
Framework}\label{Sec:P_Char}

Using the ladder operator framework, and treating $T$ and $t$ as
independent variables, the result below is proved in
Appendix \ref{sec:Toda}. This gives a PDE satisfied by
$\log\mathcal{M}_\ga(T,t)$ in the variables $T$ and $t.$

\begin{theorem}\label{Thm:Hn(Tt0)}
Let the quantity $H_n(T,t)$ be defined through the Hankel
determinant $D_n(T,{t})$ as \bea \label{def:HnIntro} H_n(T,{t})&:=&
(T\partial_T+{t}\partial_{{t}})\log
D_n(T,{t}) \\
\label{eq:Hn=LAFlogMgf}&=&
(T\partial_T+{t}\partial_{{t}})\log\mathcal{M}_\ga(T,t),
\eea
where
the second equality follows from (\ref{eq:Mgf(Dn)}). Then
$H_n(T,t)$ satisfies the following PDE:\footnote{Dropping the $T$ and $t$ dependence notation for
the sake of brevity.}: 
\bea \label{eq:PDE(Hn)Intro}
H_n&=&-2(\partial_TH_n)(\partial_{{t}}H_n)+(2n-\N+\al+T)\partial_TH_n+(2n+\N+\al+{t})\partial_{{t}}H_n\nonumber\\
&&\pm A_1(H_n)\mp
A_2(H_n)-\frac{\Big[(T\partial_{TT}^2+{t}\partial_{T{t}}^2)H_n\Big]\Big[(T\partial_{T{t}}^2+{t}\partial_{{t}{t}}^2)H_n\Big]
+A_1(H_n)A_2(H_n)}{2\Big[T(\partial_TH_n)+{t}(\partial_{{t}}H_n)-H_n+n(n+\al)\Big]},\nonumber\\
\eea
where
\begin{small}
\bea\label{defn:A1(Hn)}
\{A_1(H_n)\}^2&=&\Big((T\partial_{TT}^2+{t}\partial_{T{t}}^2)H_n\Big)^2
\nonumber\\&&
+4\Big(T(\partial_TH_n)+{t}(\partial_{{t}}H_n)-H_n+n(n+\al)\Big)\Big(\partial_TH_n\Big)\Big(\partial_TH_n-\N\Big),
\eea
\end{small}
and
\begin{small}
\bea\label{defn:A2(Hn)}
\{A_2(H_n)\}^2&=&\Big((T\partial_{T{t}}^2+{t}\partial_{{t}{t}}^2)H_n\Big)^2
\nonumber\\&&
+4\Big(T(\partial_TH_n)+{t}(\partial_{{t}}H_n)-H_n+n(n+\al)\Big)\Big(\partial_{{t}}H_n\Big)\Big(\partial_{{t}}H_n+\N\Big).
\eea
\end{small}
\end{theorem}

\begin{remark}
If $H_n(T,t)$ is a function of $T$ only, i.e., $H_n(T,t)=Y_n(T)$,
then the PDE (\ref{eq:PDE(Hn)Intro})
reduces to the Jimbo-Miwa-Okamoto $\sigma$-form
\cite{JimboMiwa1981vII} associated with Painlev\'e V: \bea
\Big(T{Y_n}^{\prime\prime}\Big)^2&=&\Big({Y_n}-T{Y_n}^\prime+2({Y_n}^\prime)^2+(\nu_0+\nu_1+\nu_2+\nu_3){Y_n}^\prime\Big)^2\nonumber
\\
&&-4\big(\nu_0+{Y_n}^\prime\big)\big(\nu_1+{Y_n}^\prime\big)\big(\nu_2+{Y_n}^\prime\big)\big(\nu_3+{Y_n}^\prime\big),\nonumber
\eea where ${}^\prime$ denotes differentiation with $T$, and \bea
\nu_0=0,\qquad\nu_1=-n-\al,\qquad \nu_2=-n,\qquad
\nu_3=-N_s.\nonumber \eea
Such a reduction can be obtained from a
$t\to\infty$ limit in (\ref{eq:mgfMultipleIntegral}).
\end{remark}
\begin{remark}
A similar reduction can be found when $H_n(T,t)$ is a function of
$t$ only, i.e., $H_n(T,t)=Y_n(t)$, which is a $\sigma-$form of
Painlev\'e V with parameters \bea \nu_0=0,\qquad\nu_1=-n-\al,\qquad
\nu_2=-n,\qquad \nu_3=N_s.\nonumber \eea This can be obtained from a
$T\to\infty$ limit in (\ref{eq:mgfMultipleIntegral}).
\end{remark}

With a change of variables the PDE, equation
(\ref{eq:PDE(Hn)Intro}), can be converted into a form which will be
convenient for the later computation of cumulants. Specifically, let
\bea
T&=&\frac{v}{1+cs},\\
{t}&=&v.
\eea
Under this transformation, $H_n(T,t)$ becomes
$H_n\Big(\frac{v}{1+cs},v\Big)$, which we write as $H_n(s,v)$ for
the sake of brevity. Hence,
\bea
\label{eq:LAFs1logMgf}H_n(s,v)&=&v\partial_v\log
\mathcal{M}_\ga(s,v).
\eea
We have that $v={t}$ and $s=\frac{1}{c}(\frac{{t}}{T}-1)$ and consequently,
\bea
\frac{\partial}{\partial T} &=&
-\frac{(1+cs)^2}{vc}\frac{\partial}{\partial s},\\
\frac{\partial}{\partial t}&=&%
\frac{(1+cs)}{vc}\frac{\partial}{\partial s}+\frac{\partial}{\partial v}.
\eea
Under the change of variables, the original PDE (\ref{eq:PDE(Hn)Intro}) becomes
the following PDE in the variables $(s,v)$:
\begin{small}
\bea\label{eq:PDE(s,v)}
H_n&=&2\frac{(1+cs)^2}{vc}\Big(\partial_sH_n\Big)\Big(\partial_v H_n\Big)+\frac{2(1+cs)^3}{(vc)^2}\Big(\partial_sH_n\Big)^2+(2n+N_s+\al+v)\Big(\partial_vH_n\Big)\nonumber\\
&&+2N_s\frac{(1+cs)}{vc}\Big(\partial_sH_n\Big)-(2n-N_s+\al)\frac{s(1+cs)}{v}\Big(\partial_sH_n\Big)
\pm A_1^\ast(H_n)\mp A_2^\ast(H_n)
\nonumber\\
&& \hspace*{-0.5cm}
-\frac{A_1^\ast(H_n)A_2^\ast(H_n)(vc)^2-(1+cs)^2\Big[(v\partial^2_{vs}-\partial_s)H_n\Big]
\Big[(1+cs)\big(v\partial_{vs}^2-\partial_s\big)H_n+v^2c\big(\partial^2_{vv}H_n\big)\Big]}{2(vc)^2\Big[v(\partial_vH_n)-H_n+n(n+\al)\Big]},
\nonumber\\
\eea
\end{small}
where $A_1^\ast(H_n)$ and $A_2^\ast(H_n)$ obtained from a
corresponding transformation are
\begin{small}
\bea\label{eq:A1*2}
(vc)^2A_1^\ast(H_n)^2&=&(1+cs)^4\Big[\big(v\partial_{vs}^2-\partial_{s}\big)H_n\Big]^2\nonumber\\
&&
+4(1+cs)^2\bigg[v\big(\partial_vH_n\big)-H_n+n(n+\al)\bigg]\bigg[\partial_sH_n\bigg]\bigg[(1+cs)^2\big(\partial_sH_n\big)+N_svc\bigg],
\nonumber\\
\eea
\end{small}
and
\begin{small}
\bea\label{eq:A2*2}
(vc)^2A_2^\ast(H_n)^2&=&\Big[(1+cs)\big(v\partial_{vs}^2-\partial_s\big)H_n+v^2c\big(\partial^2_{vv}H_n\big)\Big]^2\nonumber\\
&&+4\bigg[v\big(\partial_vH_n\big)-H_n+n(n+\al)\bigg]\bigg[\Big((1+cs)\partial_s+vc\partial_v\Big)H_n\bigg]
\nonumber\\&&\;\;
\times\bigg[\Big((1+cs)\partial_s+vc\partial_v\Big)H_n+N_svc\bigg].
\eea
\end{small}
We note that in the PDE (\ref{eq:PDE(s,v)}), the only higher order
derivatives are the mixed derivatives with respect to $s$ and $v$,
and second order derivatives with respect to $v$.




It is of interest to note that the Hankel determinant in
(\ref{eq:mgfMultipleIntegral}) or (\ref{eq:Mgf(Dn)}) has the
following asymptotic forms:
\bea
\lim_{T\to\infty} \bigg(T^{n\N}D_n[w_{{\rm AF}}(\cdot,T,{t})]\bigg)&=&D_n[w_{\rm dLag}(\cdot,{t},\al,\N)],\\
\lim_{{t}\to\infty}
\bigg(t^{-n\N}D_n[w_{{\rm AF}}(\cdot,T,{t})]\bigg)&=&D_n[w_{\rm dLag}(\cdot,T,\al,-\N)],
\eea
where the generating weight is the
``first-time" deformation of the Laguerre weight
\bea \label{def:dLag_weight}
w_{\rm dLag}(x,T,\al,\N)=
 x^\al e^{-x}(x+T)^{\N}.
\eea In previous work by the authors,~\cite{ChenMckay2010} the
Hankel determinant $D_n[w_{{\rm dLag}}(\cdot,T,\al,\N)]$ was shown
to arise in the analysis of the moment generating function of the
single-user MIMO channel capacity, and was shown to involve a
Painlev\'e V differential equation.

\section{Coulomb Fluid Method For Large $n$ Analysis}\label{Sec:CF}

In this section, we make use of the Coulomb Fluid method, which is
particularly convenient when the size of the matrix $n$ is large, to
describe the MIMO-AF problem. The idea is to treat the
eigenvalues as identically charged particles, with logarithmic
repulsion, and held together by an external potential. When $n$, in
this context the number of particles, is large, this assembly is
regarded as a continuous fluid, first put forward by Dyson,
\cite{Dyson1962I,Dyson1962II,Dyson1962III} where the eigenvalues
were supported on the unit circle. For a detailed description of
cases where the charged particles are supported on the line, see
Refs.~\onlinecite{ChenIsmail1997T},~\onlinecite{ChenManning1994,ChenLawrence1997}.

An extension of the methodology to the study of linear statistics,
namely, the sum of functions of the eigenvalues of the form
$$
\sum\limits_{j=1}^n f(x_j),
$$
can be found in Ref.~\onlinecite{ChenLawrence1997} and will be used extensively
in this paper. The main benefit of this approach, based on singular
integral equations, is that it leads to relatively simple expressions
for characterizing our moment generating function.

In Section \ref{SubSec:Coulomb_Prelim}, for completeness, we give a
brief overview of the key elements of the Coulomb Fluid method,
following
Refs.~\onlinecite{ChenIsmail1997T},~\onlinecite{ChenLawrence1997},~\onlinecite{ChenManning1994,ChenMckay2010}.
In Section \ref{SubSec:CF_AnF_Calc}, we compute explicit solutions
for the key quantities of interest in the Coulomb Fluid framework.
Subsequently, in Section \ref{SubSec:SER_Analysis}, we combine our
Coulomb Fluid results with either (\ref{eq:SERExact}) or
(\ref{eq:SERApprox}) to directly yield analytical approximations for
the SER of the MIMO-AF scheme under consideration. Quite remarkably,
these are shown to be extremely accurate, even for very small
dimensions. Similar accuracy is demonstrated in Section
\ref{eq:SubSec:Cumulants} for SNR cumulant approximations obtained
from the Coulomb Fluid results.

\subsection{Preliminaries of the Coulomb Fluid Method}\label{SubSec:Coulomb_Prelim}
We start by considering the ratio in the moment generating function
expression (\ref{def:CF:MGF_Zn/Z0}), noting that it is
of the form
\begin{equation}\label{defn:CF:Dn/D0}
\frac{Z_n(T^\prime,t^\prime)}{Z_n({t}^\prime,{t}^\prime)}
=e^{-\big[F_n(T^\prime,{t}^\prime)-F_n({t}^\prime,{t}^\prime)\big]},
\end{equation}
where \bea\label{defn:CF:ZnCFM} Z_n (T^\prime,t^\prime) :=
\frac{1}{n!}\int\limits_{[0,\infty)^n} \exp \left[ - \Phi (x_1,
\ldots, x_n ) - \sum_{j=1}^n f(x_j,T^\prime,{t}^\prime) \right]
\prod_{l=1}^n d x_l  \eea
\begin{align}
\label{eq:CF:PhiDefinition} \Phi(x_1, \ldots, x_n) := -2 \sum_{1
\leq j < k \leq n} \log | x_j - x_k | + n \sum_{j=1}^n \v_0 (x_j) \;
\end{align}
and $F_n:=-\log Z_n$ is known as the
\emph{free energy}.

This expression embraces the moment generating function expression
(\ref{def:CF:MGF_Zn/Z0}) with appropriate selection of the functions
$\textsf{v}_0(x)$ and $f(x,T^\prime,t^\prime)$.
A key motivation for writing our problem in this form is that it
admits a very intuitive interpretation in terms of statistical
physics, as originally observed by Dyson (see
Refs.~\onlinecite{Dyson1962I,Dyson1962II,Dyson1962III}). Specifically, if the
eigenvalues $x_1,\dots,x_n$ are interpreted as the positions of $n$
identically charged particles, then $\Phi(x_1, \ldots, x_n)$ is
recognized as the total energy of the repelling charged particles,
which are confined by the external potential $n\textsf{v}_0(x)$. The
function $f(x,T^\prime,{t}^\prime)$ acts as a perturbation to the
system, resulting in a modification to the external potential. The
quantity $F_n(T^\prime,t^\prime)$ may be interpreted as the free
energy of the system under an external perturbation
$f(x,T^\prime,t^\prime)$, with $F_n(t^\prime,t^\prime)$ the free
energy of the unperturbed system.

For sufficiently large $n$, the system of particles, following
Dyson, may be approximated as a continuous fluid where techniques of
macroscopic physics and electrostatics can be applied. For large $n$
we expect the external potential $n\textsf{v}_0(x)$ to be strong
enough to overcome the logarithmic repulsion between the particles
(or eigenvalues), and hence the particles or fluid will be confined
within a finite interval to be determined through a minimization
process. For this continuous fluid, we introduce a macroscopic
density $\sigma(x)dx$, referred to as the equilibrium density. Since
$\textsf{v}_0(x)$ is convex for $x\in \mathbb{R}$, this density is
supported on a single interval denoted by $(a,b)$, to be determined
later (see Ref.~\onlinecite{ChenIsmail1997T} for a detailed explanation). The
equilibrium density is obtained by minimizing the free-energy
functional:
\begin{equation}\label{eq:CF:FM(sig)}
F_n(T^\prime,{t}^\prime):=\int\limits^{b}_{a}\sigma(x)\Big(n^2\textsf{v}_0(x)+nf(x,T^\prime,t^\prime)\Big)dx-n^2\int\limits^{b}_{a}\int\limits^{b}_{a}\sigma(x)\log\vert
x-y\vert\sigma(y)dx dy,
\end{equation}
subject to \bea\label{eq:CF:CFConstraint}
\int\limits^{b}_{a}\sigma(x)dx&=&1. \eea With Frostman's Lemma
\cite[p. 65]{Tsuji1959}, the minimizing $\sigma(x)dx$ can be
characterized through the
integral equation 
\bea
n^2\textsf{v}_0(x)+n{f}(x,T^\prime,t^\prime)-2n^2\int\limits^{b}_{a}\log\vert
x-y\vert\sigma(y)dy&=&A, \eea where $x\in[a,b]$ and $A$ is the
Lagrange multiplier for the normalization condition
(\ref{eq:CF:CFConstraint}), which can be interpreted as the chemical
potential of the fluid.~\cite{ChenIsmail1997T,ChenManning1994}
Noting that the integral equation above has a logarithmic kernel,
taking a derivative with respect to $x\in(a,b)$ converts it into a
singular integral equation of the form \bea\label{eq:CF:CFSIntEqn}
{\textsf{v}_0}^\prime(x)+\frac{{f}^\prime(x,T^\prime,t^\prime)}{n}&=&
2\mathcal{P}\int\limits^{b}_{a}\frac{\sigma(y)}{x-y}dy, \eea where
$\mathcal{P}$ denotes Cauchy principal value.

Noting the structure (in $n$) of the left-hand side of
(\ref{eq:CF:CFSIntEqn}), it is clear that $\sigma(\cdot)$ must take
the general form:
\begin{equation}\label{eq:CF:OptimalSig}
\sigma(x)=\sigma_0(x)+\frac{\sigma_c(x,T^\prime,t^\prime)}{n},
\end{equation}
where $\sig_0(x)dx$ is the density of the original system in the
absence of any perturbation, while $\sig_c(x,T^\prime,t^\prime)$
represents the deformation of $\sigma_0(x)$ caused by
$f(x,T^\prime,t^\prime)$.  Furthermore, to satisfy
(\ref{eq:CF:CFConstraint}), we have \bea
\int\limits_a^b\sig_0(x)dx=1, \qquad\qquad
\int\limits_a^b\sig_c(x,T^\prime,t^\prime)dx=0. \eea Substituting
(\ref{eq:CF:OptimalSig}) into (\ref{eq:CF:CFSIntEqn}), and comparing
orders of $n$, we see that $\sig_0(x)$ solves
\bea\label{eq:CF:voprimeSInt} {\textsf{v}_0}^\prime(x)&=&
2\mathcal{P}\int\limits^{b}_{a}\frac{\sigma_0(y)}{x-y}dy, \eea and
$\sig_c(x,T^\prime,t^\prime)$ solves \bea\label{eq:CF:fprimeSInt}
{f}^\prime(x,T^\prime,t^\prime)&=&
2\mathcal{P}\int\limits^{b}_{a}\frac{\sigma_c(y,T^\prime,t^\prime)}{x-y}dy.
\eea
Following Ref.~\onlinecite{ChenIsmail1997T}, where the choice for the
solution for $\sigma_0$  has been extensively discussed based on the
theory described in Ref.~\onlinecite{Gakhov1966}; the solution to
(\ref{eq:CF:voprimeSInt}) subject to the boundary condition
$\sig_0(a)=\sig_0(b)=0$ reads
\begin{equation}\label{eq:CF:sig0}
\sigma_0(x)=\frac{\sqrt{(b-x)(x-a)}}{2\pi^2}
\int\limits^{b}_{a}\frac{\textsf{v}_0^\prime(x)-\textsf{v}_0^\prime(y)}{(x-y)\sqrt{(b-y)(y-a)}}dy,
\end{equation}
together with a supplementary condition,
\begin{equation}
\label{eq:CF:BSupSigma1}\int\limits^{b}_{a}\frac{{\textsf{v}}_0^\prime(x)}{\sqrt{(b-x)(x-a)}}dx=0.
\end{equation}
The solution to (\ref{eq:CF:fprimeSInt}) subject to
$\int\limits_a^b\sig_c(x,T^\prime,t^\prime)dx=0$ is
\begin{equation}\label{eq:CF:sigc}
\sigma_c(x,T^\prime,t^\prime)=\frac{1}{2\pi^2\sqrt{(b-x)(x-a)}}
\mathcal{P}\int\limits^b_a\frac{\sqrt{(b-y)(y-a)}}{y-x}f^\prime(y,T^\prime,t^\prime)dy.
\end{equation}
Finally, the normalization condition (\ref{eq:CF:CFConstraint})
becomes
\begin{equation}
\label{eq:CF:BSupSigma2}\int\limits^{b}_{a}\frac{x{\textsf{v}}_0^\prime(x)}{\sqrt{(b-x)(x-a)}}dx=2\pi.
\end{equation}
The end points of the support of the density $\sig_0(x)$, $a$, and
$b$, are determined by (\ref{eq:CF:BSupSigma1}) and
(\ref{eq:CF:BSupSigma2}), and will depend on parameters associated
with $\textsf{v}_0.$ For a description see Ref.~\onlinecite{ChenIsmail1997T}.
(We mention here a related problem, namely, the probability that
there is a gap in the spectrum of the random matrix, has been
studied using the Coulomb Fluid approach in
Refs.~\onlinecite{ChenManning1994_Gap_short,ChenManning1996_Gap_long}.)

With the above results, for sufficiently large $n$, we may
approximate the ratio (\ref{defn:CF:Dn/D0}) as (see
Ref.~\onlinecite{ChenLawrence1997} for more details)
\bea\label{eq:CF:ZnCFApprox}
\frac{Z_n(T^\prime,t^\prime)}{Z_n(t^\prime,t^\prime)}&\approx&
\exp\Big(-S_2^{{\rm AF}}(T^\prime,{t}^\prime)-nS_1^{{\rm
AF}}(T^\prime,{t}^\prime)\Big), \eea where \bea
\label{def:CF:S1AF}S_1^{{\rm
AF}}(T^\prime,{t}^\prime)&=&\int\limits_a^b\!\sigma_0(x)
f(x,T^\prime,t^\prime)dx,\\
\label{def:CF:S2AF}S_2^{{\rm
AF}}(T^\prime,{t}^\prime)&=&\frac{1}{2}\int\limits_a^b
\!\sigma_c(x,T^\prime,t^\prime)f(x,T^\prime,t^\prime)dx. \eea In the
sequel, we will find explicit solutions for these quantities.

\subsection{Coulomb Fluid Calculations for the SNR Moment Generating Function}\label{SubSec:CF_AnF_Calc}

Using (\ref{eq:CF:ZnCFApprox}), the moment generating function
$\mathcal{M}_\ga(T^\prime,{t}^\prime)$ in (\ref{def:CF:MGF_Zn/Z0})
takes the form
\begin{eqnarray}
\label{eq:CF:Mgf(Tptp)}
\mathcal{M}_\ga(T^\prime,{t}^\prime)&=&\left(\frac{T^\prime}{{t}^\prime}\right)^{n\N}
\frac{Z_n(T^\prime,{t}^\prime)}{Z_n({t}^\prime,{t}^\prime)},\nonumber\\
&\approx&\exp\left(-S_2^{{\rm
AF}}(T^\prime,{t}^\prime)-n\left[S_1^{{\rm AF}}
(T^\prime,{t}^\prime)-{\N}\log\left(\frac{T^\prime}{{t}^\prime}\right)\right]\right)
,
\end{eqnarray}
for which, comparing (\ref{def:CF:ZnCFM}) and
(\ref{eq:CF:PhiDefinition_Intro})
 with (\ref{defn:CF:ZnCFM}) and (\ref{eq:CF:PhiDefinition}), the functions $\v_0(x)$ and
 $f(x,T',t')$ are identified as
\bea
\label{eq:CF:v0(x)}\textsf{v}_0(x)&:=& x - \beta\log x,\\
\label{eq:CF:f(xTt0)}f(x,T^\prime,{t}^\prime)&:=&{\N}\log\left(\frac{T^\prime+x}{{t}^\prime+x}\right),
\eea where $\textsf{v}_0(x)$ given by (\ref{eq:CF:v0(x)}) is convex.

From here, the $l$-th cumulant can be extracted from the formula
\bea\label{eq:CF:CumulantFormulaTptp}
\kappa_l&=&c^l\left(\frac{{T^\prime}^2}{{t^\prime}}\frac{d}{dT^\prime}\right)^l\log
\mathcal{M}_\ga(T^\prime,{t}^\prime)\bigg\vert_{T^\prime={t}^\prime}.
\eea The objective of the subsequent analysis is to evaluate the
quantities $S_1^{{\rm AF}}$ and $S_2^{{\rm AF}}$. As we shall see,
we will need to solve numerous integrals which are not readily
available. Thus, to aid the reader, we have compiled these integrals
in Appendix \ref{App:Int}.

We start by considering the end points of the support  $a$ and $b$.
These are determined by equations (\ref{eq:CF:BSupSigma1}) and
(\ref{eq:CF:BSupSigma2}). With the integral identities
(\ref{AInt:1/sqrt})-(\ref{AInt:x/sqrt}) in Appendix \ref{App:Int},
we obtain, after a few easy steps, 
\begin{equation}\begin{aligned}
ab&=\beta^2,\\
a+b&=2(2+\bt), 
\end{aligned}\end{equation} which leads to
\begin{equation}\label{eq:CF:(a)(b)[bt]}\begin{aligned}
a&=&{\!\!\!}2+\beta-2\sqrt{1+\beta},\\
b&=&{\!\!\!}2+\beta+2\sqrt{1+\beta}.
\end{aligned}\end{equation}
The limiting density $\sigma_0(x)$ in (\ref{eq:CF:sig0}) can be
computed using the integral identity (\ref{AInt:1/x+tsqrt}) as
\begin{equation}\label{eq:CF:Sig0(mu)}
\sigma_0(x)=\frac{\sqrt{(b-x)(x-a)}}{2\pi x}, \qquad x\in (a,b),
\end{equation}
which is the Mar\^cenko-Pastur law.
\cite{MarcenkoPastur1967,KatzavCastillo2010,ChenMckay2010}
Meanwhile, with the aid of the integral identity
(\ref{AInt:Psqrt/(x+t)}) in Appendix \ref{App:Int},
$\sigma_c(x,T^\prime,{t}^\prime)$ given by (\ref{eq:CF:sigc}) reads
\bea
\sigma_c(x,T^\prime,{t}^\prime)&=&\frac{{\N}}{2\pi
\sqrt{(b-x)(x-a)}}\left(\frac{\sqrt{(T^\prime+a)(T^\prime+b)}}{x+T^\prime}-\frac{\sqrt{({t}^\prime+a)({t}^\prime+b)}}{x+{t}^\prime}\right).
\nonumber\\
\eea

Using $\sig_0(x)$ and invoking the integral identities
(\ref{AInt:log(x+t)/sqrt})-(\ref{AInt:xlog(x+t)/sqrt}) in Appendix \ref{App:Int}, gives
\begin{small}
\bea\label{eq:CF:S1AF} S_1^{{\rm
AF}}(T^\prime,{t}^\prime)&=&\frac{{\N}}{2}\bigg({t}^\prime-T^\prime+\sqrt{(T^\prime+a)(T^\prime+b)}-
\sqrt{(t^\prime+a)(t^\prime+b)}\bigg)
\nonumber\\&&
+{\N}(a+b)\log\left(\frac{\sqrt{T^\prime+a}+\sqrt{T^\prime+b}}{\sqrt{{t}^\prime+a}+\sqrt{{t}^\prime+b}}\right)
\nonumber\\
&&+\frac{{\N}\sqrt{ab}}{2}\log\left(\frac{\Big(\sqrt{ab}+\sqrt{(t^\prime+a)(t^\prime+b)}\Big)^2-{{t}^\prime}^2}{\Big(\sqrt{ab}+\sqrt{(T^\prime+a)(T^\prime+b)}\Big)^2-{T^\prime}^2}\right)\nonumber\\
&&+\frac{{\N}(a+b)}{4}\log\left(\frac{\Big(\sqrt{(t^\prime+a)(t^\prime+b)}+{t}^\prime\Big)^2-ab}{\Big(\sqrt{(T^\prime+a)(T^\prime+b)}+T^\prime\Big)^2-ab}\right)
+\frac{{\N}(a+b)}{4}\log\left(\frac{T^\prime}{{t}^\prime}\right).\nonumber\\
\eea
\end{small}
Moreover, using $\sig_c(x,T^\prime,t^\prime)$ and invoking the
integral identities (\ref{AInt:1/x+tsqrt}),
(\ref{AInt:log(x+t)/(x+t)sqrt}) and
(\ref{AInt:log(x+t2)/(x+t1)sqrt}) in Appendix \ref{App:Int} gives
\bea \label{eq:CF:S2AF}S_2^{{\rm
AF}}(T^\prime,{t}^\prime)&=&\frac{{\N}^2}{2}\log\left(\frac{4\sqrt{(T^\prime+a)(T^\prime+b)}\sqrt{(t^\prime+a)(t^\prime+b)}}{\big(\sqrt{(T^\prime+a)(T^\prime+b)}+\sqrt{(t^\prime+a)(t^\prime+b)}\big)^2-(T^\prime-{t}^\prime)^2}\right).
\nonumber\\
\eea

\subsection{SER Performance Measure Analysis Based on Coulomb Fluid}\label{SubSec:SER_Analysis}
Combining (\ref{eq:CF:Mgf(Tptp)}) with (\ref{eq:CF:(a)(b)[bt]}),
(\ref{eq:CF:S1AF}), and (\ref{eq:CF:S2AF}) yields a closed-form
asymptotic expression for the moment generating function of the
instantaneous SNR in (\ref{eq:SNR_AF_Def}).  This, in turn, combined
with either (\ref{eq:SERExact}) or (\ref{eq:SERApprox}), directly
yields analytical approximations for the SER of the MIMO-AF scheme
under consideration. The accuracy of these approximations is
confirmed in Fig.\ \ref{fig:SERQPSK}, where they are compared with
simulation results generated by numerically computing the exact SER
via Monte Carlo methods. Different antenna configurations are shown,
as represented by the form $({\N}, N_R, N_D)$. The curves labeled
``Coulomb (Exact SER)'' were generated by substituting our Coulomb
Fluid approximation into the exact expression of the SER
(\ref{eq:SERExact}), whilst the curves labeled ``Coulomb (Approx
SER)'' are generated by substituting our Coulomb Fluid approximation
into the approximate expression for the SER (\ref{eq:SERApprox}).

From these curves, it is remarkable that the Coulomb Fluid based
approximations, derived under the assumption of large $n$, very
accurately predict the SER even for small values of $n$. This is
evident, for example, by examining the set of curves corresponding
to the configuration (2, 3, 2), for which $n=2$.  In fact, even for
the extreme cases with $n=1$, the Coulomb Fluid curves still yield
quite high accuracy.

We note that the results in Fig.\ \ref{fig:SERQPSK} are
representative of those presented previously in Ref.~\onlinecite[Fig.\
5]{DharMckayMallik2010}.  The key difference therein was that the
analytic curves were based on substituting into either
(\ref{eq:SERExact}) or (\ref{eq:SERApprox}) an equivalent
determinant form of the moment generating function to that given in
(\ref{eq:Mgf(detmu)}) and (\ref{eq:Moments(Kummer)}). That
representation, involving a determinant of a matrix with elements
comprising sums of Kummer functions, is far more complicated than
the Coulomb Fluid representation, which is a simple algebraic
equation involving only elementary functions. The fact that the
Coulomb Fluid representation yields accurate approximations for the
SER for small as well as large numbers of antennas makes it a useful
analytical tool for studying the performance of arbitrary MIMO-AF
systems.

\begin{figure}[!ht]
\includegraphics[width=.9\textwidth]{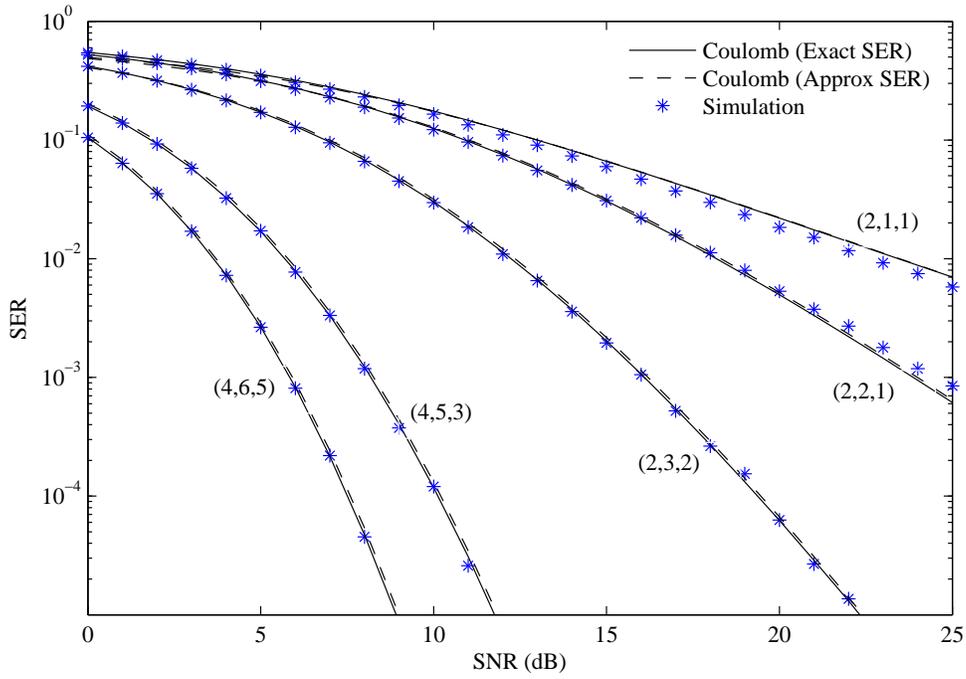}
\caption{Illustration of the SER versus average received SNR (at
relay) $\bar{\gamma}$; comparison of analysis and simulations. Each
set of curves represents a specific antenna configuration of the
form $({\N}, N_R, N_D)$. For configurations with $\N = 2$, the
full-rate Alamouti OSTBC is used (i.e., $R = 1$), whilst for $\N
=4$, a rate $1/2$ orthogonal design is employed (i.e., $R = 1/2$). In
all cases, QPSK digital modulation is assumed, such that $M = 4$.
The relay power $\tilde{b}$ is assumed to scale with $\bar{\gamma}$
by setting $\tilde{b} = \bar{\gamma}$.} \label{fig:SERQPSK}
\end{figure}


\subsection{Coulomb Fluid Analysis of Large $n$ Cumulants of SNR}
\label{eq:SubSec:Cumulants}

In this subsection we use the Coulomb Fluid results to study the
asymptotic cumulants of the SNR. We suppose $\bt$ is fixed, and
distinguish between two cases, $\bt=0$ and $\bt>0$.
\subsubsection{Case 1: $\bt=0$} Here $N_R=N_D$. Then from
(\ref{eq:CF:(a)(b)[bt]}), $a=0$ and $b=4$, leading to
\begin{small}
\bea\label{eq:CF:S1AFCase1}
S_1^{{\rm AF}}(T^\prime,{t}^\prime)-{\N}\log\left(\frac{T^\prime}{{t}^\prime}\right)&=&\frac{{\N}}{2}
\Big({t}^\prime-T^\prime+\sqrt{T^\prime(T^\prime+4)}-\sqrt{{t}^\prime({t}^\prime+4)}\Big)
\nonumber\\&&
+2{\N}\log\left(\frac{\sqrt{T^\prime {t}^\prime}+\sqrt{{t}^\prime(T^\prime+4)}}{\sqrt{T^\prime {t}^\prime}+\sqrt{T^\prime({t}^\prime+4)}}\right),
\eea
\end{small}
and
\bea\label{eq:CF:S2AFCase1}
S_2^{{\rm AF}}(T^\prime,{t}^\prime)&=&
\frac{{\N}^2}{2}\log\left(\frac{2\sqrt{T^\prime {t}^\prime}\sqrt{(T^\prime+4)({t}^\prime+4)}}{\sqrt{T^\prime {t}^\prime}\sqrt{(T^\prime+4)({t}^\prime+4)}+2T^\prime+2{t}^\prime+T^\prime {t}^\prime}\right).\:
\eea

From the closed-form approximation of the moment generating function (\ref{eq:CF:Mgf(Tptp)})
the cumulants can be extracted using (\ref{eq:CF:CumulantFormulaTptp}).

Alternatively, recall that the variable $T^\prime$ is related to $s$ via
$$
T^\prime=\frac{{t}^\prime}{1+cs},
$$
where $c=\bg/(R{\N})$. Therefore we expand $\log {\cal
M}_{\ga}(t'/(1+cs),t')$ obtained from the Coulomb Fluid formalism
(\ref{eq:CF:S1AFCase1}) and (\ref{eq:CF:S2AFCase1})
in a Taylor series about $s=0$. From
(\ref{eq:CF:Mgf(Tptp)}),
\bea\label{eq:Mgf(s)}
\log\mathcal{M}_\ga(t'/(1+cs),t')&=&\sum\limits_{j=1}^\infty(-1)^j\frac{s^j}{j!}\kappa_j(t^\prime),
\eea
where the first few cumulants are
\bea
\label{eq:CF:kappa1_bt_0}\frac{\kappa_1(t^\prime)}{c{\N}}&=&\frac{n}{2}\left(t^\prime+2-\sqrt{{t^\prime}^2+4{t^\prime}}\right),
\\
\label{eq:CF:kappa2_bt_0}\frac{\kappa_2(t^\prime)}{c^2{\N}}&=&\frac{{\N}}{(t^\prime+4)^2}+n(t^\prime+1)
-n\sqrt{{t^\prime}^2+4{t^\prime}}\left(1-\frac{1}{{t^\prime}+4}\right),\\
\label{eq:CF:kappa3_bt_0}\frac{\kappa_3({t^\prime})}{c^3{\N}}&=&\frac{12{\N}}{(t^\prime+4)^3}
+n(3{t^\prime+2})
-n\sqrt{{t^\prime}^2+4{t^\prime}}\left(3-\frac{4}{{t^\prime}+4}-\frac{2}{({t^\prime}+4)^2}\right),\\
\label{eq:CF:kappa4_bt_0}\frac{\kappa_4({t^\prime})}{c^4{\N}}&=&\frac{174{\N}}{(t^\prime+4)^4}
+6n(2{t^\prime+1})
-6n\sqrt{{t^\prime}^2+4{t^\prime}}\left(2-\frac{3}{{t^\prime}+4}-\frac{2}{({t^\prime}+4)^2}-\frac{2}{({t^\prime}+4)^3}\right),\nonumber\\ \\
\label{eq:CF:kappa5_bt_0}\frac{\kappa_5({t^\prime})}{c^5{\N}}&=&\frac{3120{\N}}{(t^\prime+4)^5}
+12n(5{t^\prime+2})\nonumber
\\&&-12n\sqrt{{t^\prime}^2+4{t^\prime}}\left(5-\frac{8}{{t^\prime}+4}-\frac{6}{({t^\prime}+4)^2}-\frac{8}{({t^\prime}+4)^3}-\frac{10}{({t^\prime}+4)^4}\right),
\eea
and
\bea
\label{eq:CF:kappa6_bt_0}\frac{\kappa_6({t^\prime})}{c^6{\N}}&=&\frac{67440{\N}}{(t^\prime+4)^6}
+120n(3{t^\prime+1})\nonumber
\\&&-120n\sqrt{{t^\prime}^2+4{t^\prime}}\left(3-\frac{5}{{t^\prime}+4}-\frac{4}{({t^\prime}+4)^2}-\frac{6}{({t^\prime}+4)^3}-\frac{10}{({t^\prime}+4)^4}
-\frac{14}{({t^\prime}+4)^5}\right).\nonumber\\
\eea

\subsubsection{Case 2: $\bt>0$} In this case, the end points of the
support, $a$ and $b,$ are given by (\ref{eq:CF:(a)(b)[bt]}).
%
Hence, going through the same process, the first five cumulants are
\begin{small}
\bea \label{eq:CF:kappa1_Gen_bt} \frac{\kappa_1^{\rm
CF}(t^\prime)}{c{\N}}&=&
\frac{n}{2}\left(t^\prime+2+\bt-\sqrt{(t^\prime+\bt)^2+4{t^\prime}}\right),\\
\label{eq:CF:kappa2_Gen_bt} \frac{\kappa_2^{\rm
CF}(t^\prime)}{c^{2}{\N}}&=&
\frac{{\N}(1+\bt){t^\prime}^2}{\big((t^\prime+\bt)^2+4t^\prime\big)^2}
+n\left({t^\prime}+1+\frac{\bt}{2}\right)-n\sqrt{(t^\prime+\bt)^2+4{t^\prime}}
\left(1-\frac{1}{2}\frac{\bt^2+(\bt+2)t^\prime}{\Big((t^\prime+\bt)^2+4{t^\prime}\Big)}\right),\nonumber\\
\\
\label{eq:CF:kappa3_Gen_bt} \frac{\kappa_3^{\rm
CF}(t^\prime)}{c^{3}{\N}}&=&
\frac{6{\N}(1+\bt)\big(\bt^2+(\bt+2)t^\prime\big){t^\prime}^2}{\big((t^\prime+\bt)^2+4t^\prime\big)^3}
+n\left(3{t^\prime}+\bt+2\right)\nonumber\\
&&-n\sqrt{(t^\prime+\bt)^2+4{t^\prime}}
\Bigg(3-\frac{2\big(1+\bt+\bt^2+(\bt+2){t^\prime}\big)}{(t^\prime+\bt)^2+4{t^\prime}}
+\frac{2\big(\bt^3+\bt^2+2(\bt+1)(\bt+2){t^\prime}\big)}{\big((t^\prime+\bt)^2+4{t^\prime}\big)^2}\Bigg), \nonumber\\
\\
\label{eq:CF:kappa4_Gen_bt} \frac{\kappa_4^{\rm
CF}(t^\prime)}{c^{4}{\N}}&=& \frac{36\N{{t^\prime}}^{2} \left(
1+\beta \right)  \Big(  \left( {\beta}^{2}+{\frac
{29}{6}}\,\beta+{\frac {29}{6}} \right)
{{t^\prime}}^{2}+2\,{\beta}^{2} \left( \beta+2\right)
{t^\prime}+{\beta}^{4} \Big)}
{\big((t^\prime+\bt)^2+4{t^\prime}\big)^4} +3\,n \left(
4t^\prime+\bt+2\right) 
\nonumber\\*&&
-3n\sqrt{(t^\prime+\bt)^2+4{t^\prime}} \Bigg( 1 +{\frac {3{t^\prime}
\left( {t^\prime}+2+\beta \right) }{({t^\prime}+\bt)^2+4{t^\prime}}}
+{\frac {2{{t^\prime}}^{2} \left( {t^\prime}-3 -3\,\beta \right) }{
\big((t^\prime+\bt)^2+4{t^\prime}\big)^2}} 
\nonumber\\*&&
\qquad-{\frac {2{{t^\prime}}^
{3} \left({t^\prime}^2 +(3+\bt){t^\prime}-2-3\,\beta\right)
}{\big((t^\prime+\bt)^2+4{t^\prime}\big)^3}}
\Bigg),
\eea
\end{small}
and
\begin{small}
\bea \label{eq:CF:kappa5_Gen_bt} \frac{\kappa_5^{\rm
CF}(t^\prime)}{c^{5}{\N}}&=& {\frac { 240\N{t^\prime}^2(1+\bt)\Big(
\left( { \beta}^{2}+\frac{13}{2}\beta+\frac{13}{2} \right)
{t^\prime}^{2}+2\,{\beta}^{2} \left( 2+\beta
 \right) t^\prime+{\beta}^{4} \Big)  \left(  \left( 2+\beta \right) t^\prime+{
\beta}^{2} \right)}
{\big((t^\prime+\bt)^2+4{t^\prime}\big)^5}}
\nonumber\\*&&
+12n\left(5t^\prime+\bt+ 2\right)
\nonumber\\*&&
-12n\sqrt{(t^\prime+\bt)^2+4{t^\prime}}
\Bigg( 1+{
\frac {4{t^\prime} \left(11t^\prime-4\bt- 8\right) }{({t^\prime}+\bt)^2+4{t^\prime}}}
-{\frac {4{t^\prime}^{2} \left(5{t^\prime}^2+2(5\bt+49){t^\prime}-37(\bt+1)\right) }{ \big((t^\prime+\bt)^2+4{t^\prime}\big)^2}}
\nonumber\\*&&
\qquad-{\frac {2{{t^\prime}}^{3} \left(10{t^\prime}^2+12(\bt+4){t^\prime}^2-3(125\bt+334){t^\prime}+348\bt+232\right) }{\big((t^\prime+\bt)^2+4{t^\prime}\big)^3}}
\nonumber\\*&&
\qquad+{\frac {10{t^\prime}^{4}
 \left(28{t^\prime}^3+3(12\bt+23){t^\prime}^2-2(135\bt+242)t^\prime+158\bt+79\right) }{ \big((t^\prime+\bt)^2+4{t^\prime}\big)^4}}
\nonumber\\*&&
\qquad-{\frac {280{u}^{5} \left(2{t^\prime}^3+(3\bt+4){t^\prime}({t^\prime}-4)+5\bt+2\right) }{ \big((t^\prime+\bt)^2+4{t^\prime}\big)^5}} \Bigg) .
\eea
\end{small}
The accuracy of these approximations is demonstrated in
Fig.\ref{Fig2}, where they are compared with simulation results
generated by numerically computing the exact cumulants via Monte
Carlo methods. As before, different antenna configurations are
shown, as represented by the form $({\N}, N_R, N_D)$. From these
curves, the accuracy of our Coulomb Fluid based approximations is
quite remarkable, even for small values of $n$.

We note that exact expressions for $\kappa_1$ and $\kappa_2$ (i.e.,
the mean and variance) were derived previously in Ref.~\onlinecite[Corollaries
1 and 2]{DharMckayMallik2010}, and these were expressed in terms of
summations of determinants (for the mean) as well as a rank-$3$
tensor (for the variance), each involving Kummer functions.  Such
results are obviously far more complicated than the Coulomb fluid
based cumulant expressions in
(\ref{eq:CF:kappa1_Gen_bt})--(\ref{eq:CF:kappa3_Gen_bt}), which
involve just very simple algebraic functions.

\begin{figure}[H]
\centering
{\includegraphics[width=0.45\columnwidth]{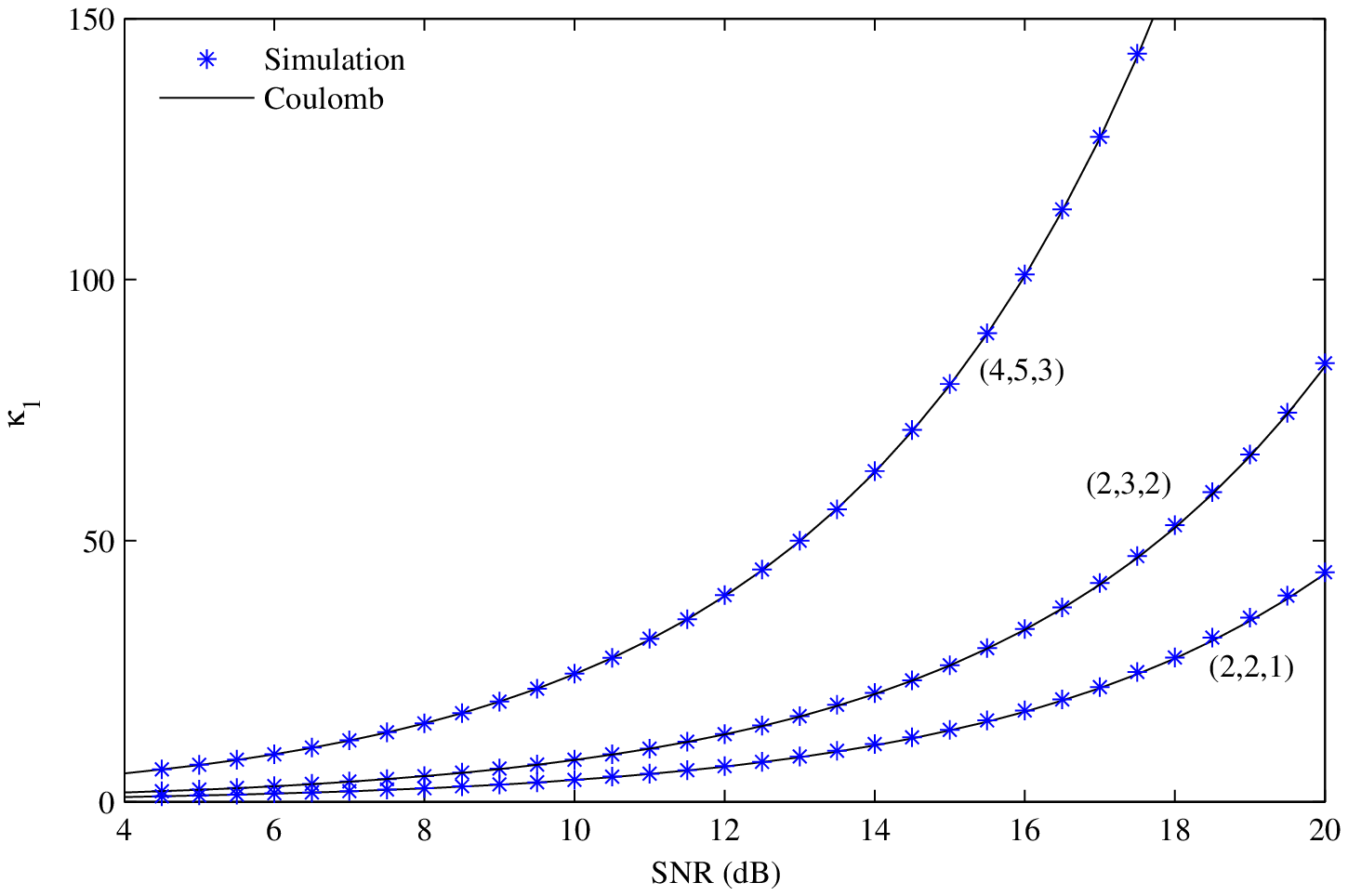}}
{\includegraphics[width=0.45\columnwidth]{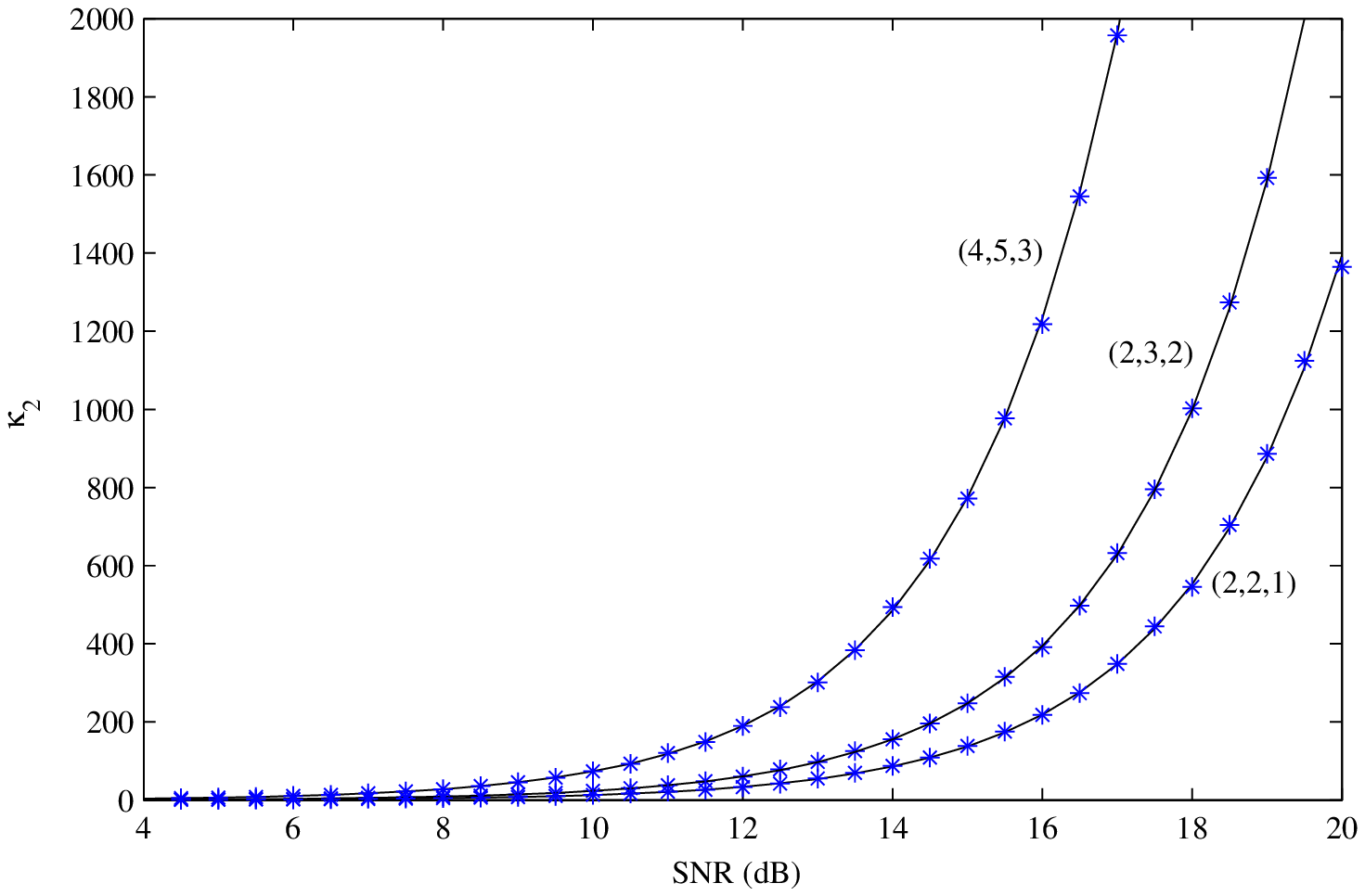}}
{\includegraphics[width=0.45\columnwidth]{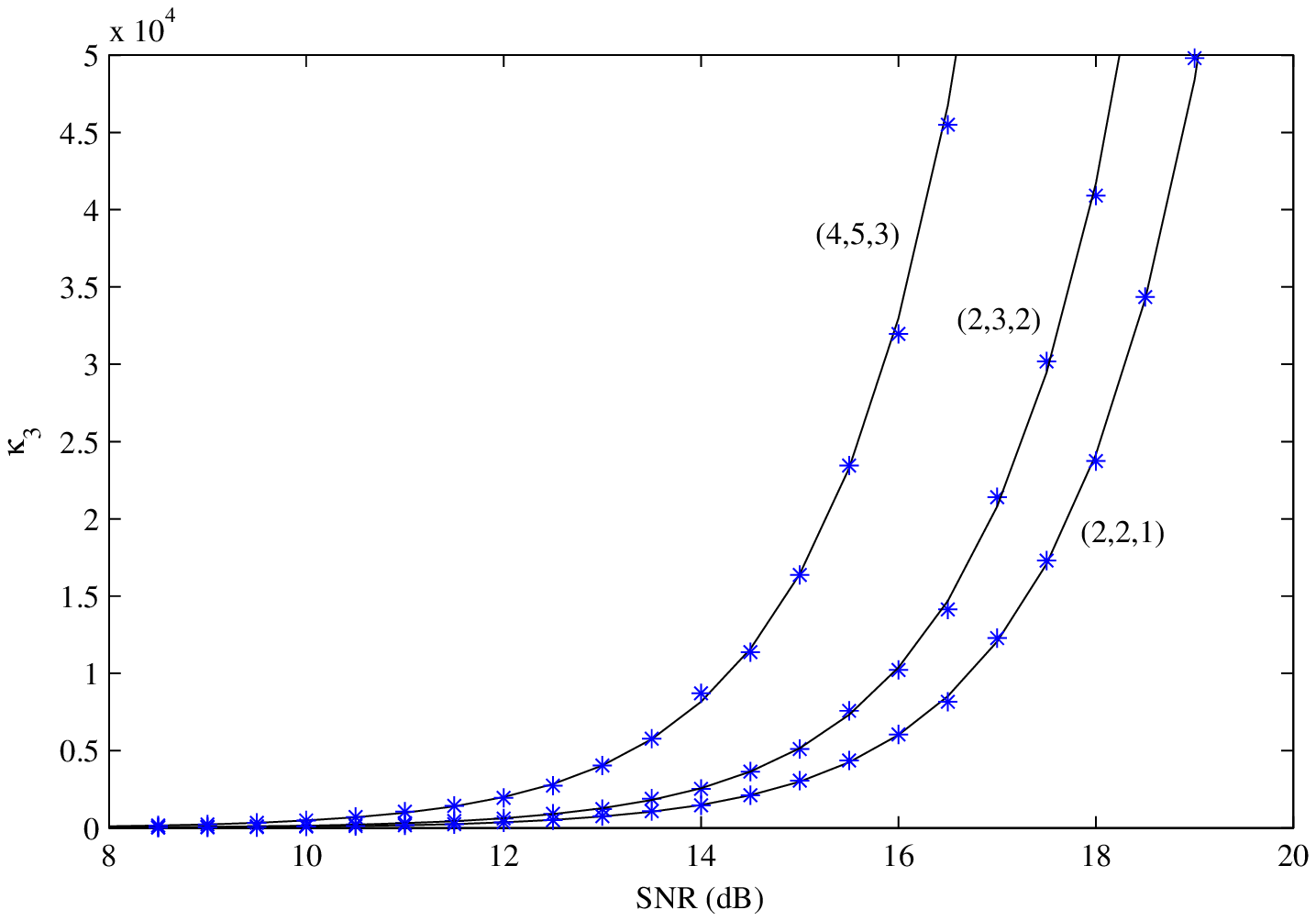}}
 \caption{Illustration of the mean $\kappa_1$, variance $\kappa_2$, and third cumulant $\kappa_3$
of the received SNR at the destination $\gamma$, as a function of
the average received SNR at the relay $\bar{\gamma}$. In each case,
the simulated cumulants are compared with those obtained based on
the Coulomb fluid approximation. As in Fig.\ \ref{fig:SERQPSK}, each
set of curves represents a specific antenna configuration of the
form $({\N}, N_R, N_D)$. For configurations with $\N = 2$, the
full-rate Alamouti OSTBC is used (i.e., $R = 1$), whilst for $\N
=4$, a rate $1/2$ orthogonal design is employed (i.e., $R = 1/2$). The
relay power $\tilde{b}$ is assumed to scale with $\bar{\gamma}$ by
setting $\tilde{b} = \bar{\gamma}$.} \label{Fig2}
\end{figure}

%
\section{Power Series Expansion For $H_n(s,{v})$}\label{Sec:Cumulant_Painleve_Analysis}

In this section we seek power series solutions for
(\ref{eq:PDE(s,v)}), to determine the cumulants. We begin by
removing the square roots of the PDE
and arrive at the following lemma.
%
%
%
%
%
%
%
\begin{lemma}
Let
\bea
\label{eq:k:defn}K&:=&2\Big(v(\partial_vH_n)-H_n+n(n+\al)\Big),
\eea
and
\begin{small}
\bea
L&:=&K\bigg[H_n-2\frac{(1+cs)^2}{vc}\Big(\partial_sH_n\Big)\Big(\partial_v H_n\Big)-\frac{2(1+cs)^3}{(vc)^2}\Big(\partial_sH_n\Big)^2\nonumber\\
&&-(2n+{\N}+\al+v)\partial_vH_n
-2{\N}\frac{(1+cs)}{vc}\partial_sH_n+(2n-{\N}+\al)\frac{s(1+cs)}{v}\partial_sH_n\bigg]\nonumber\\
&&-\frac{(1+cs)^2}{(vc)^2}\Big[(v\partial^2_{vs}-\partial_s)H_n\Big]\Big[(1+cs)\big(v\partial_{vs}^2-\partial_s\big)H_n
+v^2c\big(\partial^2_{vv}H_n\big)\Big],
\eea
\end{small} 
then the PDE (\ref{eq:PDE(s,v)}) may
be rewritten in an equivalent square-root-free form, as
\bea
\label{eq:PDE(Hn)LongIntro}
4K^2A_1^2\big(L+A_2^2\big)^2&=&\Big(L^2+K^2(A_1^2-A_2^2)-A_1^2A_2^2\Big)^2,
\eea
where $A_1$ and $A_2$ are given by (\ref{eq:A1*2}) and (\ref{eq:A2*2}) respectively.
\end{lemma}


Now, recall from (\ref{eq:LAFs1logMgf}) that the function $H_n(s,v)$
is related to the moment generating function
(\ref{eq:mgfMultipleIntegral}) through
\bea
\label{eq:LAFs1logMgfv2}H_n(s,v)&=&v\partial_{v}\log
\mathcal{M}_\ga(s,v),
\eea
and note that the logarithm of the moment
generating function has the expansion,
\bea \log
\mathcal{M}_\ga(s,v)&=&\sum\limits_{j=1}^\infty
(-1)^j\frac{\kappa_j^{}(v)}{j!}s^j,
 \eea
where $\kappa_j(v)$ is the $j$th cumulant.
Note that for the sake of brevity we write $\mathcal{M}_\ga(s,v)$ in place of $\mathcal{M}_\ga\left(\frac{v}{1+cs},v\right)$.
Consequently, from
(\ref{eq:LAFs1logMgfv2}), $H_n(s,v)$ has the following
expansion in $s,$
\bea\label{eq:PowHn(st0)}
H_n(s,v)=v\sum\limits_{j=1}^\infty
(-1)^jc^j{\N}\frac{a_j^\prime(v)}{j!}s^j,
\eea
where ${}^\prime$ denotes differentiation with respect to $v$, and $a_j(v)$ is related
to the $j$th cumulant $\kappa_j(v)$ by
\footnote{We introduce $a_j$
in order to make a clearer comparison to the Coulomb Fluid model and
\cite{DharMckayMallik2010} without having to exactly specify a value
for the constant $c$, since it falls out of the subsequent
expansion.} 
\bea a_j(v)c^j{\N}&=&\kappa_j(v). \eea

\subsection{Analysis of $\kappa_1$}

In this subsection, the computation of the mean $\kappa_1$ is presented. To this end,
upon substituting (\ref{eq:PowHn(st0)}) into the PDE
(\ref{eq:PDE(Hn)LongIntro}), keeping the lowest power of $s$ in the
resulting expansion gives a highly nonlinear ODE expressed in the factored form
\bea\label{eq:Ps1_1^1Psi_2^1=0} \Psi_1^{(1)}
\Psi_2^{(1)} &=& 0,
\eea where
\begin{small}
\bea\label{eq:Psi_1^1}
\Psi_1^{(1)}&:=&
\left(a_1''(v)\right)^2
\Bigg[\frac{{\N}^2v^2}{4}\left(a_1''(v)\right)^2+n(n+\al)\Bigg({\N}^2\left(a_1'(v)\right)^2-{\N}^2
\left(a_1'(v)\right)-n(n+\al)\Bigg)\Bigg],\;
\eea
\end{small}
and
\begin{small}
\bea\label{eq:ODEa1(v)}
\Psi_2^{(1)}&:=&
\bigg[v^2\left(a_1'''(v)\right)+v\left(a_1''(v)\right)+4n(n+\al)\bigg(a_1'(v)-\frac{1}{2}\bigg)\bigg]^2
\nonumber\\
&&-(2n+v+\al)^2\bigg[v^2\left(a_1''(v)\right)^2+4n(n+\al)\left(a_1'(v)\right)\bigg(a_1'(v)-1\bigg)\bigg].
\eea
\end{small}
Hence,  (\ref{eq:Ps1_1^1Psi_2^1=0}) is equivalent to either $\Psi_1^{(1)}=0$ or $\Psi_2^{(1)}=0$.

The ODE $\Psi_1^{(1)}=0$ has two types of solutions, the first of the form $a_1(v)= c_1v+c_2$ where $c_1$ and $c_2$ are constants to be determined.
The second is more complicated but can be succinctly written as
\begin{small}
\bea
a_1(v)&=&f(c_1,{\N},n,\al)\;v\sin\left(2\sqrt{n}\sqrt{n+\al}\log v\right)
\nonumber\\&&
+g(c_1,{\N},n,\al)\;v\cos\left(2\sqrt{n}\sqrt{n+\al}\log v\right)
+c_2+\frac{v}{2},
\eea
\end{small}
where $c_1$ and $c_2$ are integration constants, and $f$ and $g$ are algebraic functions (we omit these for brevity).

Whilst these are valid mathematical solutions, it is clear that by
taking $n$ large, they differ drastically from that predicted by
the Coulomb Fluid method and, for small $n$, they differ with the
solutions obtained in Ref.~\onlinecite{DharMckayMallik2010}. This suggests that
these are not the solutions of interest to our problem; thus, we set
them aside and henceforth focus on $\Psi_2^{(1)}=0$.

It becomes evident at this point that we require initial
conditions for the cumulants $\kappa_j(v)/(c^j{\N})$, since the PDE
(\ref{eq:PDE(Hn)LongIntro}) gives rise to a system of ODEs satisfied
by $\kappa_j(v)/(c^j{\N})$. However, (\ref{eq:PDE(Hn)LongIntro}) is
not supplemented by any initial conditions from which one deduces
the initial conditions at $t^\prime=0$ of the system of ODEs.

As we shall see, results from the Coulomb Fluid analysis
(\ref{eq:CF:kappa1_Gen_bt})-(\ref{eq:CF:kappa3_Gen_bt}) are crucial,
as these will specify initial conditions at $t^\prime=0$.
Consequently, the Coulomb Fluid results are in fact leading order
approximations to the exact results. Without the Coulomb Fluid
analysis, each cumulant $\kappa_j(v)/(c^j{\N})$ would carry unknown
constants of integration.

We now examine the equation $\Psi_2^{(1)} = 0$ in the large $n$
region, which leads to an asymptotic characterization of $a_1(v)$
and thus $\kappa_1$. To proceed, we scale the variable $v$ by
$v=n{t^\prime}$.  Also, note that $\al=n(m/n-1)\equiv n\beta$. We
find that $a_1(t^\prime)$ satisfies the following ODE:
\bea\label{eq:LargenODEa1(tp)}
0&=&\bigg[\frac{{t^\prime}^2}{n}\left(a_1'''({t^\prime})\right)+\frac{t^\prime}{n}
\left(a_1''({t^\prime})\right)+2n(1+\bt)\left(2a_1'({t^\prime})-n\right)\bigg]^2
\nonumber\\
&&-({t^\prime}+2+\bt)^2\bigg[{t^\prime}^2\left(a_1''({t^\prime})\right)^2+4n^2(1+\bt)
\left(a_1'({t^\prime})\right)\left(a_1'({t^\prime})-n\right)\bigg].\nonumber\\
\eea
For large $n$, keeping the leading order term, we arrive at
\bea\label{eq:ODEa1(t)Approx}
\left(\frac{da_1}{d{t^\prime}}-\frac{n}{2}\right)^2&=&\frac{n^2({t^\prime}+2+\bt)^2}{4({t^\prime}+\bt)^2+16{t^\prime}},
\eea whose solutions are \bea a_1(t^\prime)&=&
\frac{n}{2}\left(t^\prime\pm\sqrt{(t^\prime+\bt)^2+4t^\prime}\right)+C_1\nonumber
\eea
where $C_1$ is a constant of integration. We retain the
solution that is bounded as $t^\prime\to\infty$,
\bea
a_1(t^\prime)&=&
\frac{n}{2}\left(t^\prime-\sqrt{(t^\prime+\bt)^2+4t^\prime}\right)+C_1.
\eea
Comparing this with the corresponding result
(\ref{eq:CF:kappa1_Gen_bt}) obtained from the Coulomb fluid
approximation, we see that $\kappa_1^{CF}(0)=cN_sn,$ which
suggests $a_1(0)=n$ and hence $C_1=n(1+\bt/2)$. Therefore, the first
cumulant $\kappa_1$ to $O(n)$ becomes
\bea
\frac{\kappa_1^{\rm Large \; n}(t^\prime)}{c{\N}}&=&a_1(t^\prime) \nonumber\\
\label{eq:kappa1LN}&=&\frac{n}{2}\left(2+\bt+t^\prime-\sqrt{(t^\prime+\bt)^2+4t^\prime}\right).
\eea

\subsection{Analysis of $\kappa_2$}

Having characterized $\kappa_1$, we now turn to the variance
$\kappa_2$.  To this end, equating the coefficients of the next
lowest power of $s$ to zero in the expansion of the PDE
(\ref{eq:PDE(Hn)LongIntro}), results in an ODE
which can be factored into the following form
\begin{small}
\bea
\Psi_1^{(1)} \left(\Psi_1^{(2)} +\Psi_2^{(2)} + \Psi_3^{(2)}
\right) &=&0.
\eea
\end{small}
The ODE $\Psi_1^{(1)}=0$, makes a reappearance, which we set aside,
 leaving the ODE
\bea\label{eq:ODE(a1va2v)}
\Psi_1^{(2)}+\Psi_2^{(2)}+\Psi_3^{(2)}&=&0 , \eea where
$\Psi_1^{(2)}$, $\Psi_2^{(2)}$ and $\Psi_3^{(2)}$ are given by
\begin{footnotesize}
\bea
\frac{\Psi_1^{(2)}}{2{\N}}&=&
8{v}^{2} \left( {a_{{1}}}^{\prime}-\frac{1}{2}\right)  \left( {v}^{2}
 \left(2 n-v+\alpha \right)  \left( {a_{{1}}}^{{\prime\prime}}
 \right) ^{2}+ \left(2n+v+\alpha \right)n(n+\al) {a_{{1}}}^{\prime}
 \left( {a_{{1}}}^{\prime}-1 \right)\right) {a_{
{1}}}^{{\prime\prime\prime}}\nonumber\\
&&
-{v}^{4} \left(2n+v+\alpha \right)
 \left( {a_{{1}}}^{{\prime\prime}} \right) ^{4}
 +2{v}^{3} \left(2n-v+\alpha \right)  \left( {a_{{1}}}^{\prime}-\frac{1}{2} \right)  \left( {a_{{1}}}
^{{\prime\prime}} \right) ^{3}\nonumber\\
&&
-{v}^{2} \Bigg[\Big( (2n+v+\al)\big((\al+v)(\al-v)+20n(n+\al)\big)-16n(n+\al)(2n+\al)\Big)
\left( {a_{{1}}}^{\prime} \right)\left( {a_{{1}}}^{\prime} -1\right)\nonumber\\
&&
\hspace{10mm}-n\left( n+\alpha
 \right)  \left( 2n-3v+\alpha \right)  \Bigg]  \left( {a_{{1
}}}^{{\prime\prime}} \right) ^{2}
\nonumber\\
 &&
 +8 \left( 2n+v+\alpha
 \right)n(n+\al) {a_{{1}}}^{\prime} \left( {a_{{1}}}^{\prime}-1 \right)    \nonumber\\
  &&
 \hspace{10mm} \times\bigg[ v \left( {a_{{1}}}^{\prime}-\frac{1}{2}
 \right) {a_{{1}}}^{{\prime\prime}}+  \left( {(v+\al)
}^{2}+4vn\right) {a_{{1}}}^
{\prime}(1-{a_{{1}}}^
{\prime})+ n\left( n+\alpha \right)\bigg],
\eea
\bea
\Psi_2^{(2)}&=&
\Bigg[ 2{v}^{4} \left( 2n(n+\al)-1 \right)  \left( {a_{{1}}
}^{{\prime\prime}} \right) ^{2}-{v}^{3} \left( {v}^{}{a_{{2}}}^{
{\it '''}}+8n(n+\al) \left( {a_{{2}}}^{\prime} -1\right)   \right) {a_{{1}}}^{{\prime\prime}}\nonumber\\
&&
+4{v}^{2} n(n+\al)\bigg( v
 \left( {a_{{1}}}^{\prime}-\frac{1}{2}\right) {a_{{2}}}^{{\prime\prime}}+ \Big((v+\al)(2n+v+\al)+2n(4n+\al-1)\Big)  \left( {a_{{1}}}^{\prime} \right) \left( {a_{{1}}}^{\prime}-1 \right) \nonumber\\
&&
+n(n+\alpha)+ \left({a_{{2}}}^{\prime}-1\right) {a_
{{1}}}^{\prime}-\frac{1}{2}{a_{{2}}}^{\prime} \bigg)  \Bigg] {a_{{1}}}^{{\it \prime\prime\prime}}
+4\,{v}^{3} \left( {n}(n+\al)-\frac{1}{2} \right)  \left( {a_{{1}}}^{{\prime\prime}} \right) ^{3}-v^4a_2^{\prime\prime\prime}\left(a_1^{\prime\prime}\right)^2\nonumber\\
 &&
 +v^2 \Bigg[ v\left( 2n+v+\alpha \right) ^{2}{a_{{2}}}^{{\prime\prime}}- 2\bigg(\Big(4(2n+v+\al)^2-8n(n+\al)\Big)n(n+\al)-v(2n+\al)-\al^2 \bigg) {a_{{1}}}^{\prime}\nonumber\\
&&
- \Big( 12{n}(n+\al)+2vn
+ \left( v+\alpha \right)
\alpha \Big) {a_{{2}}}^{\prime}
+4n \left( n+\alpha \right)  \Big( (2n+v+\al)^2-2n(n+\al)+3 \Big)    \Bigg]  \left( {a_{{1}}}^{{\prime\prime}}
 \right) ^{2}\nonumber\\
 &&
 -4n(n+\al)v^2a_1^{\prime\prime}\left(a_1^{\prime}-\frac{1}{2}\right)\left(va_2^{\prime\prime\prime}-a_2^{\prime\prime}\right)
 -12vn^2(n+\al)^2a_1^{\prime\prime},
 \eea
 \end{footnotesize}
 and
\begin{footnotesize}
 \bea
 \frac{\Psi_3^{(2)}}{4n(n+\al)}&=&
 v\Bigg[  \Big((2n+v+\al)^2+4n(n+\al)-2 \Big)
 \left( {a_{{1}}}^{\prime} \right) ^{2}
 - \Big((2n+v+\al)^2-4n(n+\al)+\frac{1}{2}\Big) {a_{{2}}}^{\prime}\nonumber\\
 &&
 \hspace{2mm}+\bigg(  \Big( 2(2n+v+\al)^2+1-8n(n+\al) \Big) {a_{{2}}}^{\prime}-3(2n+v+\al)^2+1+4n(n+\al) \bigg) {a_{{1}}}^{\prime}
\Bigg] {a_{{1}}}^{{\prime\prime}}\nonumber\\
&&
+ v \bigg[
 \left( 4nv+ \left( v+\alpha \right) ^{2} \right)  \left( {a_{{1}}}
^{\prime} \right)(1-a_1^\prime)+ n\left( n+\alpha \right)  \bigg] {a_{{2}}}^{{\prime\prime}}
\nonumber\\&&
+\Big( 2nv+\left( v+\alpha
 \right) \alpha \Big) {a_{{2}}}^{\prime}(1-a_1^\prime)a_1^\prime
\nonumber\\
&&+
 \bigg[ 2(2n+v+\al)\Big(2n+\al-2(2n+v+\al)n(n+\al)\Big)
 \nonumber\\&&
 \qquad+8n(n+\al)\big(2n(n+\al)-1\big)\bigg]  \left( {a_{{1}}}^{\prime} \right)^{2}(a_1^\prime-1)\nonumber\\
 &&
+ 2\Big((2n+v+\al)^2-4n(n+\al) \Big)n(n+\al)  \left( {a_{{1}}}^{\prime} \right) ^{2}
-2 n\left( n+\alpha \right)  \left( -\frac{1}{2}{a_
{{2}}}^{\prime}+ n\left( n+\alpha \right)  \right)
\nonumber\\
&&+ 2\Big(6n(n+\al)-(2n+v+\al)^2-1 \Big) n \left( n+\alpha \right)
  {a_{{1}}}^{\prime}.
\eea
\end{footnotesize}
Clearly, $\Psi_1^{(2)}$ depends on $a_1$ only, whilst $\Psi_2^{(2)}$
and $\Psi_3^{(2)}$ depend on both $a_1$ and $a_2$.  To extract the
large $n$ behavior of $\kappa_2$, we replace $v$ by $n{t^\prime}$ in
(\ref{eq:ODE(a1va2v)}) and make use of the large $n$ formula for
$a_1(t')$ in (\ref{eq:kappa1LN}), resulting to a highly non-linear
ODE satisfied by $a_2(t^\prime)$. To proceed further, keeping the
highest powers of $n$, a first order equation is obtained for
$a_2(t^\prime)$,
\begin{small}
\bea
\frac{da_2({t^\prime})}{d{t^\prime}}&=&
n
-n\sqrt{(t^\prime+\bt)^2+4{t^\prime}} \Bigg[ {\frac {  {t^\prime}+\beta+2}{(t^\prime+\bt)^2+4t^\prime}
-{\frac {{t^\prime}(\bt+1) }{ \big((t^\prime+\bt)^2+4t^\prime\big)^2}}}
 \Bigg]
 \nonumber\\&&
 -{\frac {2{\N}{t^\prime}(\bt+1)({t^\prime}-\bt)({t^\prime}+\bt) }{ \big((t^\prime+\bt)^2+4t^\prime\big)^3}}.
\eea
\end{small}
Integrating this leads to an expression for $a_2(t^\prime)$ which,
again yields an unknown integration constant.  This constant can be determined
through a comparison with the Coulomb Fluid results for
$\kappa_2(t^\prime)$ in (\ref{eq:CF:kappa2_Gen_bt}) at $t'=0,$ giving $a_2(0)=n$.

The above analysis gives the leading order characterization of the variance,
\begin{small}
\bea\label{eq:kappa2LN}
\frac{\kappa_2^{\rm Large \; n}(t^\prime)}{c^2{\N}}&=&a_2(t^\prime)\nonumber\\
&=&
n\left({t^\prime}+1+\frac{\bt}{2}\right)-n\sqrt{(t^\prime+\bt)^2+4{t^\prime}}
\left[1-\frac{1}{2}\frac{\bt^2+(\bt+2)t^\prime}{\Big((t^\prime+\bt)^2+4{t^\prime}\Big)}\right]
+\frac{{\N}(1+\bt){t^\prime}^2}{\big((t^\prime+\bt)^2+4t^\prime\big)^2}.
\nonumber\\
\eea
\end{small}

\subsection{Beyond $\kappa_1$ and $\kappa_2$}

The same procedure for computing $\kappa_1$ and
$\kappa_2$ easily extends to higher cumulants.
Here we give the leading order formula for
$\kappa_3$, $\kappa_4$ and $\kappa_5$,
\begin{footnotesize}
\bea
\frac{\kappa_3^{\rm Large \; n}(t^\prime)}{c^3{\N}}&=&a_3(t^\prime) \nonumber\\
&=&
n\left(3{t^\prime}+\bt+2\right)
+\frac{6{\N}{t^\prime}^2(1+\bt)\big(\bt^2+(\bt+2)t^\prime\big)}{\big((t^\prime+\bt)^2+4t^\prime\big)^3},
\nonumber\\&&
-n\sqrt{(t^\prime+\bt)^2+4{t^\prime}}
\Bigg[3-\frac{2\big(1+\bt+\bt^2+(\bt+2){t^\prime}\big)}{(t^\prime+\bt)^2+4{t^\prime}}
+\frac{2\big(\bt^3+\bt^2+2(\bt+1)(\bt+2){t^\prime}\big)}{\big((t^\prime+\bt)^2+4{t^\prime}\big)^2}\Bigg]
\nonumber\\&&
\eea
\bea
\frac{\kappa_4^{\rm Large \; n}(t^\prime)}{c^4{\N}}&=&a_4(t^\prime) \nonumber\\
&=&
-3n\sqrt{(t^\prime+\bt)^2+4{t^\prime}}
\nonumber\\&&\times
\left[1+{\frac {3{t^\prime} \left( {t^\prime}+2+\beta \right) }{({t^\prime}+\bt)^2+4{t^\prime}}}
+{\frac {2{{t^\prime}}^{2} \left( {t^\prime}-3
-3\,\beta \right) }{ \big((t^\prime+\bt)^2+4{t^\prime}\big)^2}}
-{\frac {2{{t^\prime}}^
{3} \left({t^\prime}^2 +(3+\bt){t^\prime}-2-3\,\beta\right) }{\big((t^\prime+\bt)^2+4{t^\prime}\big)^3}}
\right]
\nonumber\\&&
+3n \left( 4t^\prime+\bt+2\right)
+\frac{36\N{{t^\prime}}^{2} \left( 1+\beta \right)  \Big(  \left( {\beta}^{2}+{\frac
{29}{6}}\,\beta+{\frac {29}{6}} \right) {{t^\prime}}^{2}+2\,{\beta}^{2} \left(
\beta+2\right) {t^\prime}+{\beta}^{4} \Big)}
{\big((t^\prime+\bt)^2+4{t^\prime}\big)^4}
\;,
\eea
\end{footnotesize}
and
\begin{footnotesize}
\bea
\frac{\kappa_5^{\rm Large \; n}(t^\prime)}{c^5{\N}}&=&a_5(t^\prime) \nonumber\\
&=&
-12n\sqrt{(t^\prime+\bt)^2+4{t^\prime}}
\Bigg[ 1+{
\frac {4{t^\prime} \left(11t^\prime-4\bt- 8\right) }{({t^\prime}+\bt)^2+4{t^\prime}}}
-{\frac {4{t^\prime}^{2} \left(5{t^\prime}^2+2(5\bt+49){t^\prime}-37(\bt+1)\right) }{ \big((t^\prime+\bt)^2+4{t^\prime}\big)^2}}
\nonumber\\&&
-{\frac {2{{t^\prime}}^{3} \left(10{t^\prime}^2+12(\bt+4){t^\prime}^2-3(125\bt+334){t^\prime}+348\bt+232\right) }{\big((t^\prime+\bt)^2+4{t^\prime}\big)^3}}
\nonumber\\&&
+{\frac {10{t^\prime}^{4}
 \left(28{t^\prime}^3+3(12\bt+23){t^\prime}^2-2(135\bt+242)t^\prime+158\bt+79\right) }{ \big((t^\prime+\bt)^2+4{t^\prime}\big)^4}}
 \nonumber\\&&
 -{\frac {280{u}^{5} \left(2{t^\prime}^3+(3\bt+4){t^\prime}({t^\prime}-4)+5\bt+2\right) }{ \big((t^\prime+\bt)^2+4{t^\prime}\big)^5}} \Bigg]
+12n\left(5t^\prime+\bt+ 2\right)
\nonumber\\&&
{+\frac { 240\N{t^\prime}^2(1+\bt)\Big(  \left( {
\beta}^{2}+\frac{13}{2}\beta+\frac{13}{2} \right) {t^\prime}^{2}+2\,{\beta}^{2} \left( 2+\beta
 \right) t^\prime+{\beta}^{4} \Big)  \left(  \left( 2+\beta \right) t^\prime+{
\beta}^{2} \right)}
{\big((t^\prime+\bt)^2+4{t^\prime}\big)^5}}.
\eea
\end{footnotesize}

\subsection{Comparison of Cumulants obtained from ODEs with those Obtained from Determinant Representation}
This subsection serves as a check for consistency of our equations.
For small values of $n$ we compute the Hankel determinant from the
moments formula (\ref{eq:Moments(Kummer)}), since
$
D_n(T,t)=\det\big(\mu_{i+j}(T,t)\big)_{i,j=0}^{n-1}.
$
The moment generating function in $s$ and $t$ reads

\bea\label{eq:MGF_det_rep}
\mathcal{M}_\ga(s,t)&=&\left(\frac{1}{1+cs}\right)^{n{\N}}
\frac{\det\Big(c_{i+j}(s,t)\Big)_{i,j=0}^{n-1}}
{\det\Big(c_{i+j}(0,t)\Big)_{i,j=0}^{n-1}}
\eea
where
\bea
 c_j(s,t):=t^{\al+j+1}\Gamma(\al+j+1)
\sum\limits_{k=0}^{\N}\binom{{\N}}{k}\left(cs\right)^kU\left(\al+j+1,\al+j+2-k,\frac{t}{1+cs}\right),\nonumber\\
\eea 
which was derived in Ref.~\onlinecite{DharMckayMallik2010}. For small
fixed integer values of $n$ and $\N$, e.g., $n=2, 3$ and $\N=10$,
and $n=4,5$ and $\N=1$, the above determinant can computed without
much difficulty, from which the cumulants follow. It can be seen
that $\frac{\kappa_1(t)}{c{\N}}=a_1(t)$ and
$\frac{\kappa_2(t)}{c^2{\N}}=a_2(t)$ obtained from
(\ref{eq:MGF_det_rep}), satisfied third order ODEs for $a_1(t)$
given by (\ref{eq:ODEa1(v)}) and $a_2(t)$ given by
(\ref{eq:ODE(a1va2v)}). Similar results hold for the higher
cumulants, which provide
a consistency check.

\section{Large $n$ Corrections of Cumulants obtained from Coulomb Fluid}\label{Sec:Large_n_Corr}
We have shown that the PDE
(\ref{eq:PDE(Hn)LongIntro}) satisfied by \bea
H_n(s,v)&=&v\partial_v\log\mathcal{M}_\ga(s,v)\nonumber \eea may be
used to generate a series of non-linear ODEs that are satisfied by the cumulants $\kappa_j(v)$. Under a
large $n$ assumption, where $v=nt^\prime$, the first few of these
ODEs are approximated as first order ODEs for
$\kappa_l^{{\rm Large\; n}}(t^\prime)$, whose solutions matched
exactly with that obtained from the Coulomb Fluid analysis for the
cumulants $\kappa_l^{\rm CF}(t^\prime)$.

We give a method in this section where the non-linear ODEs generated
from the PDE (\ref{eq:PDE(Hn)LongIntro}) are employed systematically
to obtain ``correction terms'' to the Coulomb Fluid results.


\subsection{Large $n$ expansion of $\kappa_1$}
We assume the first cumulant or $\kappa_1$ has the following large $n$ expansion
\bea
\label{eq:K1ExpansionLN}a_1(t^\prime)&=&\frac{\kappa_1({t^\prime})}{c{\N}}\nonumber\\
&=&
ne_{-1}(t^\prime)+e_0(t^\prime)+\sum\limits_{k=1}^\infty\frac{e_k(t^\prime)}{n^k}.
\eea Substituting the above into (\ref{eq:LargenODEa1(tp)}), and
setting the coefficients of $n^j$ to zero, a system of first order
ODEs for $e_k(t^\prime)$ is obtained. These are are solved
successively starting from $e_{-1}(t^\prime)$, followed by
$e_0(t^\prime)$ and so on.

The expression for $\kappa_1^{{\rm Large \; n}}(t^\prime)$ given by (\ref{eq:CF:kappa1_Gen_bt}), gives rise to the initial conditions
\bea
\label{eq:InCond_ei}
e_{-1}(0)=1 \qquad\text{and}\qquad e_k(0)=0 \qquad\text{for}\qquad k=0,1,2,3,\dots.
\eea

As a result
$e_{-1}(t^\prime)$ is found to
satisfy the ODE
\bea
\left(\frac{de_{-1}({t^\prime})}{d{t^\prime}}-\frac{n}{2}\right)^2&=&
\frac{n^2({t^\prime}+2+\bt)^2}{4({t^\prime}+\bt)^2+16{t^\prime}}\;,
\eea
and has two solutions
$$
e_{-1}(t^\prime)=\frac{1}{2}\left(2\mp\bt+t^\prime\pm\sqrt{(t^\prime+\bt)^2+4t^\prime}\right),
$$
both satisfying the initial condition $e_1(0)=1.$

We retain the solution that is finite at infinity
\bea\label{eq:LNC:em1prime}
e_{-1}(t^\prime)=\frac{1}{2}\left(2+\bt+t^\prime-\sqrt{(t^\prime+\bt)^2+4t^\prime}\right),
\eea to match with that obtained from the Coulomb Fluid computation,
namely, (\ref{eq:CF:kappa1_Gen_bt}).

Continuing, we set the coefficient of $n^3$ to zero, which implies,
\bea\label{eq:LNC:e0prime} \frac{de_0}{d{t^\prime}}&=&0. \eea
Setting the coefficient of $n^2$ to zero gives rise to an
ODE involving $e_{-1}^{'}(t^{\prime})$,
$e_{-1}^{''}(t')$, $e_{-1}^{'''}(t')$, $e_{0}^{'}(t')$  and $e_1^{'}(t^{\prime})$. Simplifying, with
$e_{-1}(t')$ given by (\ref{eq:LNC:em1prime}) and $e_{0}^{'}(t')$ given by (\ref{eq:LNC:e0prime}), we find that
$e_1(t^\prime)$ satisfies: \bea
\label{eq:LNC:e1prime}\frac{\big(({t^\prime}+\bt)^2+4t^\prime\big)^{7/2}}{t^\prime(1+\bt)}\;
\frac{de_1}{d{t^\prime}}&=&-3{t^\prime}^2-(\bt+2){t^\prime}+2\bt^2.
\eea Setting the coefficient of the next lowest power of $n$ to zero
gives rise to an ODE involving $e_{-1}^{'}(t^\prime)$, $e_{-1}^{''}(t^\prime)$, $e_{-}^{'''}(t^\prime)$, $e_{0}^{'}(t^\prime)$,
$e_0^{''}(t^\prime)$, $e_0^{'''}(t^\prime)$, $e_{1}^{'}(t^\prime)$, and $e_2^{'}(t^\prime)$. From the expression of
$e_{-1}(t')$ and $e_0'(t^\prime)$ given by (\ref{eq:LNC:em1prime}) and (\ref{eq:LNC:e0prime}) respectively, we find \bea
\frac{de_2}{d{t^\prime}}&=&0. \eea Continue with this process, the
next three terms $e_3(t^\prime)$,  $e_4(t^\prime)$ and
$e_5(t^\prime)$ are found to satisfy the following equations:
\begin{small}
\bea
\frac{\big(({t^\prime}+\bt)^2+4t^\prime\big)^{13/2}}{t^\prime(1+\bt)}\frac{de_3}{dt^\prime}&=&
l_1^{(3)}(t^\prime),\\
\frac{de_4}{d{t^\prime}}&=&0,\\
\label{eq:LNC:e5prime}\frac{\big(({t^\prime}+\bt)^2+4{t^\prime}\big)^{19/2}}{{t^\prime}(1+\bt)}
\frac{de_5}{dt^\prime}&=&
l_1^{(5)}(t^\prime),
\eea
\end{small}
where $l_1^{(3)}(t^\prime)$ and $l_1^{(5)}(t^\prime)$ are given by
\begin{small}
\bea
\label{eq:l_1^3}l_1^{(3)}(t^\prime)&=&-40{t^\prime}^6-16(\bt+2){t^\prime}^5+(149\bt^2-79\bt-79)
{t^\prime}^4+(\bt+2)(145\bt^2-27\bt-27){t^\prime}^3
\nonumber\\&&
-(19\bt^2-236\bt-236)\bt^2{t^\prime}^2+37(\bt+2)\bt {t^\prime}-2\bt^6,\\
\label{eq:l_1^5}l_1^{(5)}(t^\prime)&=&-1260{{t^\prime}}^{10}+412 \left(\bt+2\right) {{t^\prime}}^{9}
+ 12\left(325{\beta}^{2}-97\bt-97\right) {{t^\prime}}^{8}
\nonumber\\&&+8(\bt+2)(1999\bt^2-516\bt-516) {{t^\prime}}^{
7}
\nonumber\\&&
- \left( 7967\,{\beta}^{4}-70762\,{\beta}^{3}-61492\,{\beta}^{2}+
18540\,\beta+9270 \right) {{t^\prime}}^{6}
\nonumber\\&&-(\bt+2)(23233\bt^4-43750\bt^3-41500\bt^2+4500\bt+2250){{t^\prime}}^{5}\nonumber\\
 &&-2 \left( 4294\,{\beta}^{4}+33485\,{\beta}^{3}+10575{
\beta}^{2}-45820{\beta}-22910\right) \bt^2{{t^\prime}}^{4}
\nonumber\\&&
+34(\bt+2)(107\bt^2-695\bt-695)\bt^4{{t^\prime}}^{3}
+ \left( 1717\,{\beta}^{2}+10072\,{\beta
}+10072 \right) \bt^6{{t^\prime}}^{2}
\nonumber\\&&
- 217\left({\beta}^{}+2
\right)\bt^8 {t^\prime}+2\bt^{10}.
\eea
\end{small}
Solving ODEs
(\ref{eq:LNC:e0prime})--(\ref{eq:LNC:e5prime}) with initial
conditions $e_k(0)=0$, we obtain,
\begin{small}
\bea\label{eq:LNC:kappa_1_LN_Exp}
\frac{\kappa_1({t^\prime})}{c{\N}}&=& \frac{\kappa_1^{\rm
CF}({t^\prime})}{c{\N}}+(1+\bt){t^\prime}^2\sqrt{(t^\prime+\bt)^2+4{t^\prime}}
\sum\limits_{k=0}^{2}\frac{A_{2k+1}^{(1)}(t^\prime)}{n^{2k+1}}+\mathcal{O}\left(\frac{1}{n^7}\right),
\eea
\end{small}
where $A^{(1)}_{1}({t^\prime})$, $A^{(1)}_{3}({t^\prime})$ and $A^{(1)}_{5}({t^\prime})$ are given by
\begin{small}
\bea
A^{(1)}_{1}({t^\prime})&=&\frac{1}{\big((t^\prime+\bt)^2+4t^\prime\big)^{3}},\\
A^{(1)}_{3}({t^\prime})&=&
\frac{1}{\big((t^\prime+\bt)^2+4t^\prime\big)^{4}}+\frac{7{t^\prime}({t^\prime}-2\bt-4)}
{\big(({t^\prime}+\bt)^2+4t^\prime\big)^{5}}+\frac{105{t^\prime}^2(1+\bt)}{\big((t^\prime+\bt)^2+4t^\prime\big)^{6}},\\
A^{(1)}_{5}({t^\prime})&=&
\frac{1}{\big((t^\prime+\bt)^2+4t^\prime\big)^{5}}+\frac{2{t^\prime}(337{t^\prime}-38\bt-76)}{\big(({t^\prime}+\bt)^2+4t^\prime\big)^{6}}
-\frac{33{t^\prime}^2(15{t^\prime}^2+484{t^\prime}+60{t^\prime}\bt-106\bt-106)}{\big((t^\prime+\bt)^2+4t^\prime\big)^{7}}\nonumber\\
&&\frac{2002{t^\prime}^3(6{t^\prime}^2+52{t^\prime}+15{t^\prime}\bt-18\bt-12)}{\big((t^\prime+\bt)^2
+4t^\prime\big)^{8}}-\frac{50050{t^\prime}^4({t^\prime}^2+4{t^\prime}+2{t^\prime}\bt-2\bt-1)}
{\big(({t^\prime}+\bt)^2+4t^\prime\big)^{9}}.
\eea
\end{small}
%
%
%
%
\subsection{Large $n$ Expansion of $\kappa_2$}

In computing a series expansion for the variance, we proceed in a
similar way as was just done for the mean. First, we substitute the
expansion for $a_1(t^\prime)$
from  (\ref{eq:LNC:kappa_1_LN_Exp}) and
\bea
\label{eq:K2ExpansionLN}a_2(t^\prime)&=&\frac{\kappa_2({t^\prime})}{c^2{\N}} \nonumber\\
&=&
nf_{-1}(t^\prime)+f_0(t^\prime)+\sum\limits_{k=1}^\infty\frac{f_k(t^\prime)}{n^k},
\eea
into (\ref{eq:ODE(a1va2v)}).
Equating the coefficients of  $n^k$ to zero, a system of first order
ODEs for $f_k(t^\prime)$ are found. Comparing with the Coulomb Fluid
results for $\kappa_2(t^\prime)$ in (\ref{eq:CF:kappa2_Gen_bt})
gives rise to the initial conditions \bea\label{eq:InCond_fi}
f_{-1}(0)=1 \qquad\text{and}\qquad f_k(0)=0 \qquad\text{for}\qquad
k=0,1,2,3,\dots. \eea In this case, setting the coefficient of
$n^{12}$ to 0 results in a first order ODE for
$f_{-1}(t^\prime)$: 
\bea\label{eq:LNC:fm1prime}
\frac{df_{-1}}{d{t^\prime}}&=&
1-\left(\frac{{t^\prime}+\bt+2}{\big((t^\prime+\bt)^2+4t^\prime\big)^{1/2}}
-\frac{2(\bt+1){t^\prime}}{\big((t^\prime+\bt)^2+4t^\prime\big)^{3/2}}\right).
\eea The solution of this, with the initial condition
$f_{-1}(0)=1,$ reads \bea f_{-1}(t^\prime)&=&
\left({t^\prime}+1+\frac{\bt}{2}\right)-\sqrt{(t^\prime+\bt)^2+4{t^\prime}}\left(1-\frac{1}{2}
\frac{\bt^2+(\bt+2)t^\prime}{\Big((t^\prime+\bt)^2+4{t^\prime}\Big)}\right).
\eea
Continuing the process, we obtain an ODE for $f_0(t^\prime)$:
\bea\label{eq:LNC:f0prime}
\frac{df_0}{dt^\prime}&=&
-\frac{2{\N}(\bt+1){t^\prime}({t^\prime}-\bt)({t^\prime}+\bt)}{\big((t^\prime+\bt)^2+4t^\prime\big)^{3}}.
\eea
The solution with the condition $f_0(0)=0$ reads
\bea
f_0(t^\prime)&=&
\frac{{\N}(1+\bt){t^\prime}^2}{\big((t^\prime+\bt)^2+4t^\prime\big)^2},
\eea
from which it can be immediately seen that
$$
nf_{-1}(t^\prime)+f_0(t^\prime) =\frac{\kappa_2^{\rm
CF}(t^\prime)}{c^2{\N}}, $$ and we recover
 $\kappa_2^{\rm CF}(t^\prime)$ found previously.

The ODEs satisfied by $f_k(t^\prime)$, $k=1,2,3,4$
are reported in Appendix \ref{App:LNCorrDiffEqns}, equations
(\ref{Appdiffeqn:f1})--(\ref{Appdiffeqn:f4}). These will allow us to
compute $f_j(t')$ with appropriate initial conditions.


Briefly, the idea is that to determine the ODE satisfied by the
$k$-th correction term $f_k(t^\prime),$ for $k$  odd, the preceding
$\frac{1}{2}(k+1)$ ODEs satisfied by $f_j(t^\prime)$, $j=-1,1,3\dots,k$ are employed.
For even $k$, the previous $\frac{k}{2}$ ODEs
satisfied by $f_j(t^\prime)$, $j=0,2,4\dots,k$ are employed.

Going through the procedure described, we find that
$\kappa_2(t^\prime)$ has the following large $n$ expansion,
\begin{small}
\bea\label{eq:LNC:kappa_2_LN_Exp}
\frac{\kappa_2({t^\prime})}{c^2{\N}}&=& \frac{\kappa_2^{\rm
CF}({t^\prime})}{c^2{\N}}+(1+\bt){t^\prime}^2\sqrt{(t^\prime+\bt)^2+4{t^\prime}}
\sum\limits_{k=1}^{2}\frac{A_{2k-1}^{(2)}(t^\prime)}{n^{2k-1}}
\nonumber\\&&
+{\N}(1+\bt){t^\prime}^2\sum\limits_{k=1}^{2}\frac{B_{2k}^{(2)}(t^\prime)}{n^{2k}}+\mathcal{O}\left(\frac{1}{n^5}\right)\;,
\eea
\end{small}
where
\begin{small}
\bea
A^{(2)}_{1}({t^\prime})&:=&
\frac{3}{\big(({t^\prime}+\bt)^2+4t^\prime\big)^{3}}-\frac{5{t^\prime}({t^\prime}+2+\bt)}{\big((t^\prime+\bt)^2+4t^\prime\big)^{4}},\\
B^{(2)}_{2}({t^\prime})&:=&
\frac{1}{\big((t^\prime+\bt)^2+4t^\prime\big)^{3}}+\frac{7{t^\prime}({t^\prime}-2\bt-4)}{\big(({t^\prime}+\bt)^2+4t^\prime\big)^{4}}
+\frac{104{t^\prime}^2(1+\bt)}{\big((t^\prime+\bt)^2+4t^\prime\big)^{5}},\\
A^{(2)}_{3}({t^\prime})&:=&
\frac{3}{\big((t^\prime+\bt)^2+4t^\prime\big)^{4}}+\frac{7{t^\prime}(22{t^\prime}-9\bt-18)}
{\big(({t^\prime}+\bt)^2+4t^\prime\big)^{5}}
-\frac{21{t^\prime}^2(9{t^\prime}^2+73{t^\prime}+
9{t^\prime}\bt-49\bt-49)}{\big((t^\prime+\bt)^2+4t^\prime\big)^{6}}
\nonumber\\&&
+\frac{1155{t^\prime}^3({t^\prime}^2+3{t^\prime}+{t^\prime}\bt-3\bt-2)}{\big((t^\prime+\bt)^2+4t^\prime\big)^{7}},
\\
B^{(2)}_{4}({t^\prime})&:=&
\frac{1}{\big((t^\prime+\bt)^2+4t^\prime\big)^{4}}+\frac{2{t^\prime}(337{t^\prime}-38\bt-76)}
{\big(({t^\prime}+\bt)^2+4t^\prime\big)^{5}}
\nonumber\\&&
-\frac{{t^\prime}^2(495{t^\prime}^2+15944{t^\prime}+1980{t^\prime}\bt
-3496\bt-3496)}{\big((t^\prime+\bt)^2+4t^\prime\big)^{6}}\nonumber\\
&&+\frac{56{t^\prime}^3(214{t^\prime}^2+1853{t^\prime}+535{t^\prime}\bt-642\bt-428)}{\big((t^\prime+\bt)^2+4t^\prime\big)^{7}}
-\frac{49840{t^\prime}^4({t^\prime}^2+4{t^\prime}+2{t^\prime}\bt-2\bt-1)}{\big(({t^\prime}+\bt)^2+4t^\prime\big)^{8}}.
\nonumber\\
\eea
\end{small}
%

\subsection{Beyond $\kappa_1$ and $\kappa_2$}

The procedure adopted above for computing the large $n$ expansion
series for $\kappa_1(t^\prime)$ and $\kappa_2(t^\prime)$ easily
extends to the higher cumulants. By way of example, here we focus on
the third cumulant $\kappa_3$. In this case, an asymptotic expansion
for $a_3(t')$ is assumed, \bea
\label{eq:K3ExpansionLN}a_3(t^\prime)&=&\frac{\kappa_3({t^\prime})}{c^3{\N}} \nonumber\\
&=&
ng_{-1}(t^\prime)+g_0(t^\prime)+\sum\limits_{k=1}^\infty\frac{g_k(t^\prime)}{n^k},
\eea
along with the initial conditions,
\bea \label{eq:InCond_gi} g_{-1}(0)=2 \qquad\text{and}\qquad
g_k(0)=0 \qquad\text{for}\qquad k=0,1,2,3,\dots. \eea In this case,
we find that the large $n$ expansion of the third cumulant reads
\begin{small}
\bea \frac{\kappa_3({t^\prime})}{c^3{\N}}&=& \frac{\kappa_3^{\rm
CF}({t^\prime})}{c^3{\N}}+(1+\bt){t^\prime}^2\sqrt{(t^\prime+\bt)^2+4{t^\prime}}
\sum\limits_{k=1}^{2}\frac{A_{2k-1}^{(3)}(t^\prime)}{n^{2k-1}}
\nonumber\\&&
+{\N}(1+\bt){t^\prime}^2\sum\limits_{k=1}^{2}\frac{B_{2k}^{(3)}(t^\prime)}{n^{2k}}+\mathcal{O}\left(\frac{1}{n^5}\right),
\eea
\end{small}
where
\begin{footnotesize}
\bea
A^{(3)}_{1}({t^\prime})&=&\frac{12}{\big(({t^\prime}+\bt)^2+4t^\prime\big)^{3}}
-\frac{10{t^\prime}({t^\prime}+4\bt+8)}{\big((t^\prime+\bt)^2+4t^\prime\big)^{4}}
+\frac{140{t^\prime}^2(1+\bt)}{\big((t^\prime+\bt)^2+4t^\prime\big)^{5}}\nonumber\\
&&+{\N}^2\Bigg(\frac{2{t^\prime}({t^\prime}-2\bt-4)}{\big((t^\prime+\bt)^2+4t^\prime\big)^{4}}
+\frac{32{t^\prime}^2(1+\bt)}{\big((t^\prime+\bt)^2+4t^\prime\big)^{5}}\Bigg),\\
B^{(3)}_{2}({t^\prime})&=&
\frac{6}{\big((t^\prime+\bt)^2+4t^\prime\big)^{3}}
+\frac{6{t^\prime}(37{t^\prime}-19\bt-38)}{\big(({t^\prime}+\bt)^2+4t^\prime\big)^{4}}
-\frac{12{t^\prime}^2(21{t^\prime}^2+172{t^\prime}+21{t^\prime}\bt-134\bt-134)}{\big((t^\prime+\bt)^2+4t^\prime\big)^{5}}
\nonumber\\&&
+\frac{1560{t^\prime}^3({t^\prime}^2+3{t^\prime}+{t^\prime}\bt-3\bt-2)}{\big((t^\prime+\bt)^2+4t^\prime\big)^{5}},
\\
A^{(3)}_{3}({t^\prime})&=&
\frac{12}{\big((t^\prime+\bt)^2+4t^\prime\big)^{4}}
+\frac{42{t^\prime}(35{t^\prime}-8\bt-16)}{\big(({t^\prime}+\bt)^2+4t^\prime\big)^{5}}
-\frac{126{t^\prime}^2(11{t^\prime}^2+206{t^\prime}+26{t^\prime}\bt-67\bt-67)}{\big((t^\prime+\bt)^2+4t^\prime\big)^{6}}\nonumber\\
&&+\frac{924{t^\prime}^3(21{t^\prime}^2+152{t^\prime}+45{t^\prime}\bt-63\bt-42)}{\big((t^\prime+\bt)^2+4t^\prime\big)^{7}}
-\frac{60060{t^\prime}^4({t^\prime}^2+4{t^\prime}+2{t^\prime}\bt-2\bt-1)}{\big(({t^\prime}+\bt)^2+4t^\prime\big)^{8}}\nonumber\\
&&+{\N}^2\Bigg(\frac{20{t^\prime}(12{t^\prime}-\bt-2)}{\big(({t^\prime}+\bt)^2+4t^\prime\big)^{5}}
-\frac{36{t^\prime}^2(5{t^\prime}^2+166{t^\prime}+20{t^\prime}\bt-34\bt-34)}{\big((t^\prime+\bt)^2+4t^\prime\big)^{6}}\nonumber\\
&&\qquad+\frac{12{t^\prime}^3(378{t^\prime}^2+3317{t^\prime}+945{t^\prime}\bt-1134\bt-756)}{\big((t^\prime+\bt)^2+4t^\prime\big)^{7}}
-\frac{19392{t^\prime}^4({t^\prime}^2+4{t^\prime}+2{t^\prime}\bt-2\bt-1)}{\big(({t^\prime}+\bt)^2+4t^\prime\big)^{8}}\Bigg),\nonumber\\ \\
\frac{B^{(3)}_{4}({t^\prime})}{6}&=&
\frac{1}{\big((t^\prime+\bt)^2+4t^\prime\big)^{4}}
+\frac{{t^\prime}(1199{t^\prime}-97\bt-194)}{\big(({t^\prime}+\bt)^2+4t^\prime\big)^{5}}
+\frac{{t^\prime}^2(3074{t^\prime}^2-42140{t^\prime}-5340{t^\prime}\bt+6004\bt+6004)}{\big((t^\prime+\bt)^2+4t^\prime\big)^{6}}\nonumber\\
&&-\frac{{t^\prime}^3(4455{t^\prime}^3+40366{t^\prime}^2-460496{t^\prime}+4455{t^\prime}^2\bt-130120{t^\prime}\bt+94380\bt+62920)}{\big((t^\prime+\bt)^2+4t^\prime\big)^{7}}
\nonumber\\&&+\frac{28{t^\prime}^4(2247{t^\prime}^3+4309{t^\prime}^2-71182{t^\prime}+2247{t^\prime}^2\bt-32965{t^\prime}\bt+19104\bt+9552)}{\big(({t^\prime}+\bt)^2+4t^\prime\big)^{8}}\nonumber\\
&&-\frac{199360{t^\prime}^5({t^\prime}^3-15{t^\prime}+{t^\prime}^2\bt-10{t^\prime}\bt+5\bt+2)}{\big(({t^\prime}+\bt)^2+4t^\prime\big)^{9}}.
\eea
\end{footnotesize}
For ease of reference, the ODEs satisfied by
$g_k(t^\prime)$, $k=1,2,3,4,$ are placed in Appendix
\ref{App:LNCorrDiffEqns}, equations
(\ref{Appdiffeqn:g1})--(\ref{Appdiffeqn:g4})

In summary, this entire section has shown that by expanding the
cumulants $\kappa_1(t^\prime)$, $\kappa_2(t^\prime)$ and
$\kappa_3(t^\prime)$ into asymptotic series in $n$, the Coulomb
Fluid results are recovered as the leading order contributions in
the large $n$ scenario. It is also seen that, by examining the
finite-$n$ correction terms for each cumulant, no terms of
$\mathcal{O}(n^2)$ or higher are present within the expansions.

\section{Asymptotic Performance Analysis Based on Coulomb Fluid}\label{Sec:CF_Large_s}

In this section, we return to the analysis of the moment generating
function, and consider the high SNR scenario (i.e., as
$\bg\rightarrow\infty$). To this end, we will study the Coulomb
Fluid based approximation derived in Section \ref{Sec:CF} (to be
complemented in Section \ref{Sec:PV_Large_s} through analysis based on
Painlev\'e equations), where the variables $T'$ and $t'$ are taken
to be dependent on $\bg,$ namely, \bea
T^\prime(s,\bg)=\frac{t^\prime(\bg)}{1+\frac{\bg s}{R\N}},\qquad
t^\prime(\bg)=\frac{1}{n}\frac{(1+\bg)N_R}{\ti{b}},\;\;\; {\rm where
\;\;\;} \ti{b}:=\nu\bg \; . \eea Note that as
$\bg\rightarrow\infty$,
$$
t'(\bg)\longrightarrow \frac{N_R}{n\nu}.
$$
To obtain the desired high SNR expansion, we compute the moment
generating function $\mathcal{M}_\ga$ as $s\to\infty$ and $\bg\to
\infty$.  The Coulomb Fluid based representation
(\ref{eq:CF:Mgf(Tptp)}), when expressed in terms of $T'(s,\bg)$ and
$\frac{N_R}{n\nu}$ reads
\begin{small}
\bea\label{eq:CF:HSNR:Mgf_bt_eq_0} \mathcal{M}_\ga (s) &\approx&
\exp\left(-S_2^{{\rm
AF}}\Big(T^\prime(s,\bg),\frac{N_R}{n\nu}\Big)-n\left[S_1^{{\rm AF}}
\Big(T^\prime(s,\bg),\frac{N_R}{n\nu}\Big)-{\N}\log\left(\frac{n\nu
T^\prime(s,\bg)}{N_R}\right)\right]\right), 
\nonumber\\
\eea
\end{small}
where $S_1^{{\rm AF}}$ and $S_2^{{\rm AF}}$ are given by
(\ref{eq:CF:S1AF}) and (\ref{eq:CF:S2AF}) respectively.  Our goal is
to compute the expansion of $\mathcal{M}_\ga (s)$ as $s\to\infty$
and $\bg\to\infty$, which turns out to be an expansion in
$(\bg s)^{-1}.$  It turns out that the cases $\beta = 0$ and $\beta
\neq 0$ behave fundamentally differently, and as such, these are
treated separately.

\subsection{The Case of $\beta = 0$}
With $\bt=0$, we have $N_R=n$. An easy computation shows that
$\mathcal{M}_\ga$ admits the following expansion:
\begin{align} \label{eq:CF:HSNR:MGFExpand_bt_0}
\mathcal{M}_\ga(s) = \sum\limits_{\ell=0}^\infty \frac{A_\ell}{(\bg
s)^{d+\ell/2}},
\end{align}
where $A_0$, $A_1$, $A_2$, $A_3,\dots$ are constants independent of $s$ and
$\bg$.  The leading exponent $d$ is given by
\begin{align}\label{eq:CF:HSNR:d_bt_0}
d = {\N}\left(n-\frac{{\N}}{4}\right)
\end{align}
whilst the first few $A_\ell$ are
\begin{footnotesize}
\bea
A_0 &=&\left(\frac{R{\N}}{\nu}\right)^{{\N}(n-{\N}/4)}\frac{\left(1+\sqrt{1+4\nu}\right)^{2n{\N}}}
{4^{{\N}(n+{\N}/4)}(1+4\nu)^{\frac{{\N}^2}{4}}}\exp\left(-\frac{n{\N}}{2\nu}\left(1-\sqrt{1+4\nu}\right)\right),\\
\sqrt{\frac{\nu}{R{\N}}}A_1 &=&
A_0 \frac{{\N}}{2}\left({\N}\sqrt{1+4\nu}-4n\right),\\
%
\sqrt{\frac{\nu}{R{\N}}}A_2 &=& \frac{A_1
}{8}\bigg[2{\N}\left({\N}\sqrt{1+4\nu}-4n\right)-
\frac{{\N}(1+4\nu)+8n(2\nu-1)}{\left({\N}\sqrt{1+4\nu}-4n\right)}\bigg],\\
%
%
\sqrt{\frac{\nu}{R{\N}}}A_3 &=& \frac{A_2
}{6}\Bigg[{\N}\left({\N}\sqrt{1+4\nu}-4n\right)-\sqrt{1+4\nu}
\nonumber\\
&&\qquad\qquad+\frac{\Big[32{\N}n^2-8n(2{\N}^2+1)(2\nu-1)-3{\N}(1+4\nu)\Big]\sqrt{1+4\nu}}
{2{\N}\left({\N}\sqrt{1+4\nu}-4n\right)^2-{\N}(1+4\nu)-8n(2\nu-1)}\nonumber\\
&&\qquad\qquad-\frac{32n\Big[\frac{{\N}^2}{4}(1+4\nu)-2n{\N}(2\nu-1)-3\nu+\frac{1}{4}\Big]}
{2{\N}\left({\N}\sqrt{1+4\nu}-4n\right)^2-{\N}(1+4\nu)-8n(2\nu-1)}\Bigg].
\eea
\end{footnotesize}

We have refrained from presenting $A_k,k\geq 4$, as these are
rather long.

\begin{remark}
For the special case $\ti{b}=\bg$ or $\nu=1$, implying equal
relay and source power, $A_0$ reduces to the remarkably simple formula:
\bea\label{eq:CF:HSNR:A_Bt0_(t=1)} A_0 &=&
\frac{(R{\N})^{{\N}\left(n-\frac{{\N}}{4}\right)}\varphi^{2{\N}n}}
{20^{\frac{{\N}^2}{4}}}\exp\left(\frac{{\N}n}{\varphi}\right), \eea
where $\varphi=(1+\sqrt{5})/2$ is the Golden ratio.
\end{remark}

%
%

\subsubsection{High SNR Analysis of the Symbol Error Rate (SER)}

Based on (\ref{eq:SERExact}), the SER of MPSK modulation can be
expanded at high SNR using (\ref{eq:CF:HSNR:MGFExpand_bt_0}),
resulting in
\begin{align} \label{eq:HSNR:PMPSK_Beta0_0}
P_{\rm MPSK} = \frac{1}{\pi} \sum\limits_{\ell=0}^\infty \frac{
A_\ell}{ (\bg g_{\rm MPSK})^{d+\ell/2} } \mathcal{I}_{d,
\ell}(\Theta)
\end{align}
where
\begin{align} \label{eq:HSNR:IFirstInt}
\mathcal{I}_{d, \ell}(\Theta) = \int\limits_0^\Theta \sin^{2 d + \ell}
\theta d \theta \; .
\end{align}
Considering the first order expansion, following Ref.~\onlinecite{WangGiannakis2003}, we
may write
\begin{align} \label{eq:HSNR:PMPSK_Beta0}
P_{\rm MPSK} = \Big( G_a \bar{\gamma} \Big)^{-G_d} + o \Big(
\bar{\gamma}^{-G_d} \Big),
\end{align}
where we identify
\begin{align}
G_d = N_s \left( n - \frac{N_s}{4} \right) \;
\end{align}
as the so-called \emph{diversity order}, and identify the factor
\begin{align}\label{eq:HSNR:Array_Gain_bt_0}
G_a = g_{\rm MPSK} \left( \frac{ A_0 \mathcal{I}_{G_d,0}(\Theta)}{
\pi} \right)^{-\frac{1}{G_d}}
\end{align}
as the so-called \emph{array gain} (or \emph{coding gain}).
We note
that the result for $G_d$ above is consistent with a previous result
obtained via a different method in Ref.~\onlinecite{SongShin2009}, whilst the
expression for $G_a$ appears new.

Whilst it appears that a closed-form solution for the integral (\ref{eq:HSNR:IFirstInt})
is not forthcoming in general (though it can be easily evaluated
numerically), such a solution does exist for the important special
case of BPSK modulation, for which $M = 2$. In this case we have the
particularization, $\Theta = \pi/2$, for which Ref.~\onlinecite{GradRyzhJeff2007}
gives
\begin{align}
\mathcal{I}_{d, \ell}(\pi/2) = \frac{\sqrt{\pi}}{2} \frac{\Gamma(d +
\ell/2 + 1/2) }{ \Gamma(d + \ell/2 + 1) } \; .
\end{align}
%
%
Hence, using (\ref{eq:HSNR:Array_Gain_bt_0}), $G_a$ admits the simplified
form
\begin{align}
G_a =  \left( \frac{A_0}{2 \sqrt{\pi}} \frac{\Gamma(G_d + 1/2) }{
\Gamma(G_d + 1) } \right)^{-\frac{1}{G_d}} \; .
\end{align}

The high SNR results above are illustrated in Fig.\
\ref{fig:SERQPSK_HighSNR_bt_0}. The ``Simulation'' curves are based on
numerically evaluating the exact SER relation (\ref{eq:SERExact});
the ``Coulomb Fluid (Exact)'' curves are based on substituting
(\ref{eq:CF:HSNR:Mgf_bt_eq_0}) into (\ref{eq:SERExact}) and numerically
evaluating the resulting integral; the ``Coulomb Fluid (Leading term
only)'' curves are based on (\ref{eq:HSNR:PMPSK_Beta0}); whilst
``Coulomb Fluid (Leading 4 terms)'' curves are based on the first
four terms of (\ref{eq:HSNR:PMPSK_Beta0_0}).  The leading-order
approximation is shown to give a reasonably good approximation at
high SNR, whilst the additional accuracy obtained by including a few
correction terms is also clearly evident.

\begin{figure}[!ht]
\includegraphics[width=.9\textwidth]{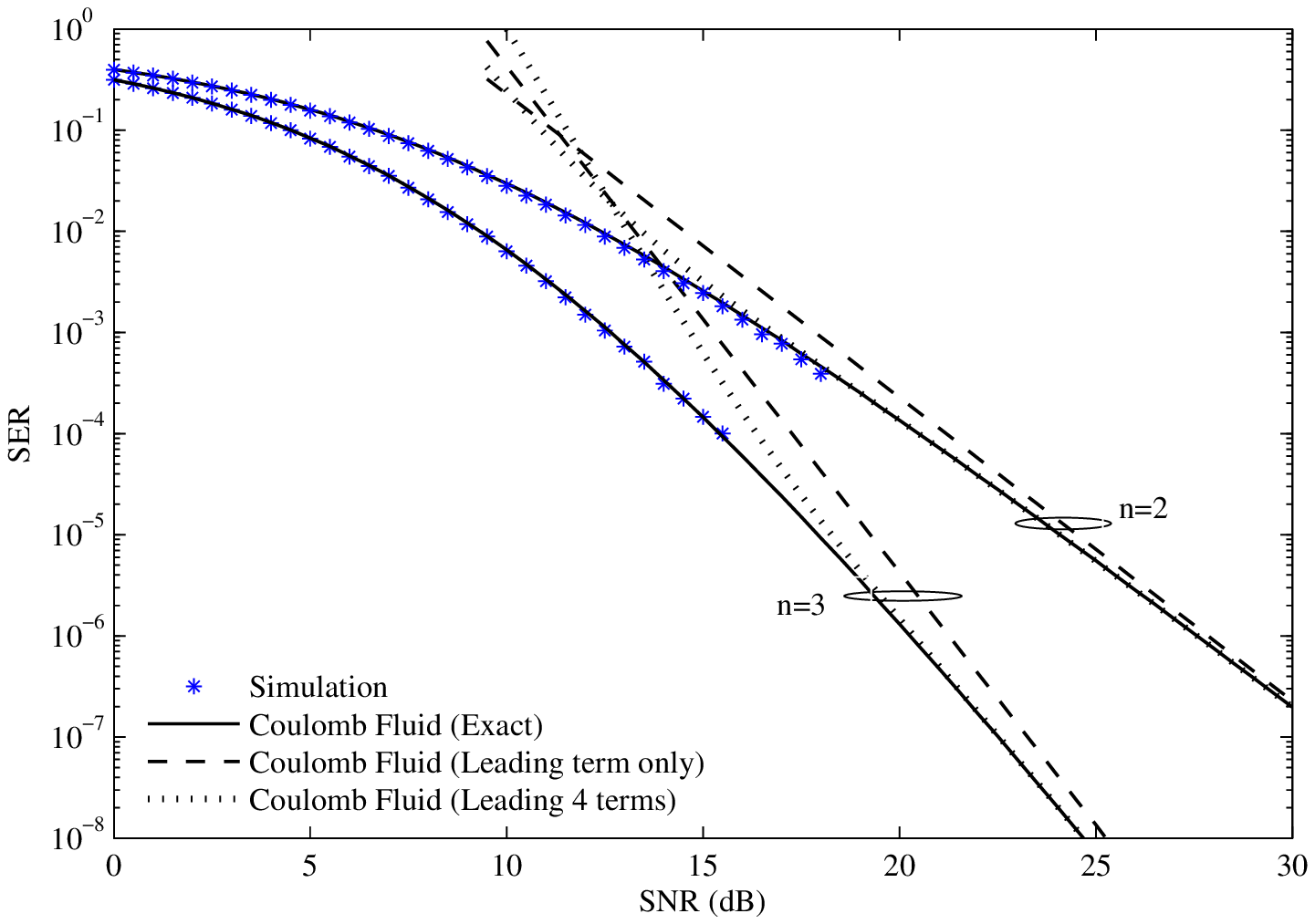}
\caption{Illustration of the SER versus average received SNR (at
relay) $\bar{\gamma}$; comparison of analysis and simulations.
Results are shown for $N_R = N_D = n$ and ${\N} = 2$, with the
full-rate Alamouti OSTBC code (i.e., $R = 1$). QPSK digital
modulation is assumed, such that $M = 4$. The relay power
$\tilde{b}$ is assumed to scale with $\bar{\gamma}$ by setting
$\tilde{b} = \frac{3}{2}\bg$.} \label{fig:SERQPSK_HighSNR_bt_0}
\end{figure}

\subsubsection{High SNR Analysis of the Probability Density Function of $\gamma$}

With the moment generating function expansion given above in
(\ref{eq:CF:HSNR:MGFExpand_bt_0}), we may also readily obtain an
approximation for the probability density function (PDF) of $\gamma$, denoted $f_\ga(x)$, by direct
Laplace Transform inversion. In particular, we obtain
\begin{align}
f_\ga(x) = \sum_{\ell = 0}^\infty \frac{A_\ell}{\Gamma(d + \ell/2)}
\frac{x^{d+\ell/2-1} }{ {\bar{\gamma}}^{d+\ell/2} } \; .
%
\end{align}
For the case of very large $\bar{\gamma}$, with
\begin{align} \label{eq:HSNR:FirstOrderPDFBeta0}
f_\ga(x) =
\frac{A_0}{\Gamma(d)}\frac{x^{d-1}}{\bg^d}+\mathcal{O}\left(\frac{1}{\bg^{d+1/2}}\right)
\; ,
\end{align}
the leading term gives an approximation for the PDF deep in the
left-hand tail. Of course, with the inclusion of more terms, a more
refined approximation is obtained.
\subsection{The Case of $\bt \neq 0$ (with $N_R < N_D$) }\label{sec:HSNR_bt_neq_0}

For the situation where $\bt\neq 0$, two sub-cases arise. This first
is $N_R < N_D$, for which $N_R=n$; the second is $N_R > N_D$, for
which $N_R = (1+\bt)n.$

In the following, we will focus on the first sub-case, $N_R<N_D$. It
turns out however, that our results also apply for $N_R>N_D$ upon
transforming the quantity $\nu$ to $\nu^\ast$ by \bea \nu =
\frac{\nu^\ast}{1+\bt}, \eea where $\nu^\ast(>0)$ is interpreted as
the (fixed) scaling factor between $\ti{b}$ and $\bg$, i.e., \bea
\nu^\ast=\frac{\ti{b}}{\bg} \; . \eea

The moment generating function given by
(\ref{eq:CF:HSNR:Mgf_bt_eq_0}) admits an expansion distinct from the
$\bt=0$ case which does \emph{not} have fractional powers of $1/(
\bg s )$, reading \bea \label{eq:CF:HSNR:MGFBetaNonZero}
\mathcal{M}_\ga(s)&=&\sum\limits_{l=0}^\infty \frac{A_l}{(\bg
s)^{d+l}} \eea where \bea d&=&n{\N}, \eea and
\begin{small}
\bea
\label{eq:HSNR:bt:A0nu}A_0&=& \frac{\Big(2+\bt+\nu\bt^2+\bt\sqrt{\big(1+\nu\bt\big)^2+4\nu}\Big)^{\frac{{\N}}{2}({\N}-n\bt)}
\Big(1+2\nu+\nu\bt+\sqrt{\big(1+\nu\bt\big)^2+4\nu}\Big)^{\frac{n{\N}}{2}(2+\bt)}}
{\nu^{n{\N}}\Big(\big(1+\nu\bt\big)^2+4\nu\Big)^{\frac{{\N}^2}{4}}\Big(1+\bt\Big)^{\frac{n{\N}}{2}(2+\bt)}
\bt^{\frac{{\N}}{2}({\N}-2n\bt)}}\nonumber\\
&&\times\frac{\big(R{\N}\big)^{n{\N}}}{2^{\frac{{\N}}{2}(2n+{\N})}}\exp\left(-\frac{n{\N}}{2\nu}
\Big(1+\nu\bt-\sqrt{\big(1+\nu\bt\big)^2+4\nu}\Big)\right),\\
A_1&=& \frac{{A_0{\N}^2R}}{2\nu\bt^2}\left[{\N}\bt\sqrt{(1+\nu\bt)^2+4\nu}-(2n+{\N})\bt(1+\nu\bt)-2{\N}\right].
\eea
\end{small}

In this situation, the sub-leading terms are very complicated,
however, the $j$th term in the expansion can be written in the
following form: \bea A_j&=&\frac{A_0R^j{\N}^{j+1}}{2
(j!)\left(\nu\beta^2\right)^j}\left[E_j{\N}\bt\sqrt{(1+\nu\bt)^2+4\nu}+F_j\right],
\eea where $E_j$ and $F_j$ also depend upon $\nu$, ${\N}$, $n$ and
$\bt$.

In $A_2$, $E_2$, and $F_2$ are given by
\begin{footnotesize}
\bea
E_2&=&-({\N}^2+2n{\N}+1)\bt(1+\nu\bt)-2{\N}^2-2,\\
F_2&=&
\bigg[  \left( {\N}+n \right)  \left( {{\N}}^{2}+n{\N}+1 \right) +n(n{\N}+1)
 \bigg] \bt^2(1+\nu\bt)^2+4{\N} \left( {{\N}}^{2}+n{\N}+1 \right)\bt(1+\nu\bt)\nonumber\\
 &&+ 2\left( n-{{\N}}^{3
}+2{\N} \right) \beta+2{\N} \left( {{\N}}^{2}+4 \right) \; .
\eea
\end{footnotesize}
In $A_3$, $E_3$, and $F_3$ are given by
\begin{footnotesize}
\bea
E_3&=&
\bigg[  \left( {{\N}}^{2}+3n{\N}+3\,{n}^{2}+1 \right)  \left( {{\N}}^{2}+2
\right) -6\,{n}^{2} \bigg] \bt^2(1+\nu\bt)^2
\nonumber\\&&
+ \left( 4{{\N}}^{2}+8+6n{\N}
\right)  \left( {{\N}}^{2}+1 \right) \bt(1+\nu\bt)+{{\N}}^{4} \left( 3-\beta \right) +
6\left(\beta+3 \right) {{\N}}^{2}
\nonumber\\&&
+12+3n{\N}\beta+4\,\beta \\
%
%
F_3&=&
- \bigg[  \left(  \left( {\N}+n \right) ^{3}+{\N}+{n}^{3} \right)  \left( {{\N}
}^{2}+2 \right) +4\,n \left( 1-{n}^{2} \right)  \bigg] \bt^3(1+\nu\bt)^3\nonumber\\
&&
-6{\N}
 \Big( {{\N}}^{2}+n{\N}+2 \Big)  \Big( {{\N}}^{2}+n{\N}+1 \Big) \bt^2(1+\nu\bt)^2\nonumber\\
 &&
 +
\bigg[  3\left(\bt-3\right) {{\N}}^{5}+ 6\left( \beta-1
\right) n{{\N}}^{4}-36{{\N}}^{3}- \left( 24+15\,\beta \right) n{{\N}}^{2}
\nonumber\\&&
\;\;-6\left(\beta\,{n}^{2}+4\,\beta+8 \right) {\N}-12\,\beta\,n
\bigg] \bt(1+\nu\bt)\nonumber\\
&&-2(3\bt+4) \left( \frac{4}{3}n+n{{\N}}^{2}-{{\N}}^{5}+8{\N} \right) -24{{\N}}^{3}-10{{\N}}^{5}+{\frac {32}{3}}n+8n{{\N}}^{2
}
.
\eea
\end{footnotesize}
In $A_4$, $E_4$, and $F_4$ are given by
\begin{scriptsize}
\bea
E_4&=&
- \bigg[  \bigg\{ 2+{{\N}}^{3}({\N}+4n)+3{{\N}}^{2} \left(1+ 2{n}^{2}
 \right) +2n{\N} \left( 3+2\,{n}^{2} \right) \bigg\}  \left( {{\N}}^{2}+3
 \right) + 4n{\N}\left(1-3{n}^{2} \right) \bigg] \bt^3(1+\nu\bt)^3\nonumber\\
 &&
 - \bigg[
36+2{\N}^3(3{\N}+8n)+ 6{\N}^2\left(5+ 2{n}^{2} \right)+44
n{\N} \bigg]  \left( {{\N}}^{2}+1 \right) \bt^2(1+\nu\bt)^2\nonumber\\
&&
+ \bigg[ 2{\N}^6 \left(\bt-5\right) +4n{\N}^5 \left( \beta-3 \right)-6{\N}^4 \left( \,
\beta+13 \right)-6n{\N}^3 \left( 5\,\beta+12 \right)-4{\N}^2
 \left( 50+17\bt+3\bt n^2\right)\nonumber\\
 &&\qquad-2n{\N}\left(
24+23\,\beta \right)-36(3+\bt)\bigg] \bt(1+\nu\bt)
+ 4{\N}^6\left(\bt-1
 \right)-12{\N}^4 \left( \bt+5 \right)-12n{{\N}}^{3}
\beta\nonumber\\
&&- 8{\N}^2\left( 17\,\beta+28 \right)-28n{\N}\beta-36(3+2\bt)
\\
%
%
%
F_4&=&
 \bigg[  2{n}^{4}{\N}^3+4n^3{\N}^2 \left( {{\N}}^{2}+3 \right) +2n^2{\N}
  \left( 3{{\N}}^{4}+9{{\N}}^{2} +11\right) +2n \left( {{\N}}^
{2}+3 \right)  \left( 2{{\N}}^{4}+3{{\N}}^{2}+2 \right) 
\nonumber\\*&&
 +{{\N}}^{
7}+6{{\N}}^{5}+11{{\N}}^{3}+6{\N} \bigg] \bt^4(1+\nu\bt)^4
\nonumber\\*&&
+8{\N} \left( {{\N}}^{2}+n{\N}+
1 \right)  \left( {{\N}}^{2}+n{\N}+2 \right)  \left( n{\N}+3+{{\N}}^{2} \right) \bt^3(1+\nu\bt)^{3}
\nonumber\\*&&
+ \Bigg[ 12{{\N}}^{2}\beta\,{n}^{3}+ \bigg\{ 12{\N}^3 \left( 4+3\,\beta
 \right)+ 12{\N}^5\left(1-\beta \right)+60{\N}\beta
 \bigg\} {n}^{2}
 \nonumber\\*&&\qquad
 + \bigg\{  12{\N}^6\left( 3-\beta \right)+6{\N}^4 \left(
26-\beta \right)+ 30{\N}^2\left( 6+5\,\beta \right)+72
\,\beta \bigg\} n+4\left(5 -\beta \right) {{\N}}^{7}+6 \left(23 -
\beta\right) {{\N}}^{5}
\nonumber\\*&&
\qquad\qquad+ \left( 310+46\,\beta \right) {{\N}}^{3}
+ \left( 144\,
\beta+288 \right) {\N} \Bigg] \bt^2(1+\nu\bt)^{2}
\nonumber\\*&&
+ \bigg[  8\left( 4+
3{M}^{2} \right) {n}^{2}{\N}\bt+ \Big(  8\left( 1-3\beta \right) {
{\N}}^{6}+ 24\left(\beta+3 \right) {{\N}}^{4}+ \left( 280\,\beta+256
 \right) {{\N}}^{2}+96\,\beta \Big) n\nonumber\\
 &&\qquad  +16\left( 1-\beta \right) {{\N}}
^{7}+ \left( 536+184\,\beta \right) {{\N}}^{3}+ \left( 168-24\,\beta \right) {{\N}}
^{5}+ \left( 576\,\beta+768 \right) {\N} \bigg] \bt(1+\nu\bt)
\nonumber\\*&&
+ 2\left( 1+{\beta}^{2}-6\beta \right) {{\N}}^{7}+ \left( -72\,\beta-24\,{\beta}^{2
}+48 \right) {{\N}}^{5}-12\,\beta\,n \left( \beta-1 \right) {{\N}}^{4}+
 \left( -8\,{\beta}^{2}+168\,\beta+352 \right) {{\N}}^{3}
 \nonumber\\*&&
 +24\, \left( \frac{10}{3}+\beta \right) n{\N}^2\beta+ \Big( 768+ 6\left( {n}^{2}+16
 \right) {\beta}^{2}+768\,\beta \Big) {\N}+24\, \left( \beta+\frac{5}{2}
 \right) n\beta\; .
\eea
\end{scriptsize}
\subsubsection{High SNR Analysis of the Symbol Error Rate (SER)}

Based on (\ref{eq:SERExact}), the SER of MPSK modulation can be
expanded at high SNR using (\ref{eq:CF:HSNR:MGFBetaNonZero}) into
\begin{align}\label{eq:HSNR:SER_MPSK_N0_Expand}
P_{\rm MPSK} = \frac{1}{\pi} \sum\limits_{\ell=0}^\infty \frac{
A_\ell}{ (\bg g_{\rm MPSK})^{d+\ell} }\mathcal{I}_{d +
\ell}(\Theta)
\end{align}
where
\begin{align} \label{eq:HSNR:ISecInt}
\mathcal{I}_{r}(\Theta) = \int_0^\Theta \sin^{2 r} \theta d \theta
\; .
\end{align}
Note that here, in contrast to (\ref{eq:HSNR:IFirstInt}), the exponent
$r$ is a \emph{positive integer}.  As such, (\ref{eq:HSNR:ISecInt})
admits the following closed-form solution: (Ref.~\onlinecite[2.513.1]{GradRyzhJeff2007})
\begin{align}
\mathcal{I}_{r}(\Theta) = \frac{\Theta}{2^{2 r}} \binom{2 r}{r} +
\frac{(-1)^r}{2^{2 r -1}} \sum_{j=0}^{r-1} (-1)^j \binom{2 r}{j}
\frac{\sin \Big( 2( r- j)\Theta \Big) }{2(r- j) } \; .
\end{align}

As before, a first-order approximation is of key interest, giving
\begin{align}\label{eq:HSNR:PMPSK_BetaN0}
P_{\rm MPSK} = \Big( G_a \bar{\gamma} \Big)^{-G_d} + o \Big(
\bar{\gamma}^{-G_d} \Big),
\end{align}
where we identify the {\em diversity order}
\begin{align}
G_d = n N_s   \; ,
\end{align}
and the {\em array gain}
\begin{align}
G_a = g_{\rm MPSK} \left( \frac{A_0 \mathcal{I}_{G_d}(\Theta)}{\pi}
\right)^{-\frac{1}{G_d}} \; .
\end{align}
The result for $G_d$ above is consistent with a result obtained via
a different method in Ref.~\onlinecite{SongShin2009}, whilst the expression for
$G_a$ appears new. The high SNR results above, for the case $\beta
\neq 0$, are illustrated in Fig.\ \ref{fig:SERQPSK_HighSNR_BNotZ}. As
before, the ``Simulation'' curves are based on numerically
evaluating the exact SER relation (\ref{eq:SERExact}), and the
``Coulomb Fluid (Exact)'' curves are based on substituting
(\ref{eq:CF:HSNR:Mgf_bt_eq_0}) into (\ref{eq:SERExact}) and
numerically evaluating the resulting integral.  Moreover, the
``Coulomb Fluid (Leading term only)'' curves are based on
(\ref{eq:HSNR:PMPSK_BetaN0}), whilst the ``Coulomb Fluid (Leading 5
terms)'' curves are based on the first five terms of
(\ref{eq:HSNR:SER_MPSK_N0_Expand}).  Again, the leading-order
approximation is shown to give a reasonably good approximation at
high SNR, whilst the additional accuracy obtained by including a few
correction terms is also evident.

\begin{figure}[!ht]
\includegraphics[width=.9\textwidth]{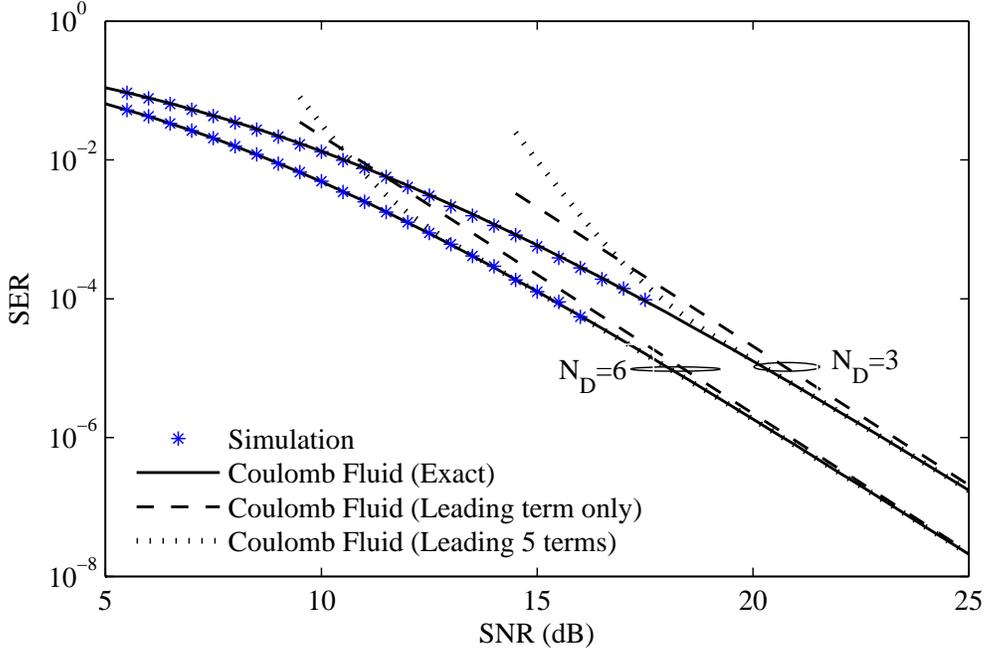}
\caption{Illustration of the SER versus average received SNR (at
relay) $\bar{\gamma}$; comparison of analysis and simulations.
Results are shown for $N_R = 2$ and ${\N} = 2$, with the full-rate
Alamouti OSTBC code (i.e., $R = 1$). QPSK digital modulation is
assumed, such that $M = 4$. The relay power $\tilde{b}$ is assumed
to scale with $\bar{\gamma}$ by setting $\tilde{b} = \frac{3}{2}
\bar{\gamma}$.} \label{fig:SERQPSK_HighSNR_BNotZ}
\end{figure}

\subsubsection{High SNR Analysis of the Probability Density Function of $\gamma$}

As before, based on the moment generating function expansion (\ref{eq:CF:HSNR:MGFBetaNonZero}),
applying for $\beta \neq 0$, we can immediately take a Laplace
inversion to obtain the following high SNR representation for the
PDF of $\gamma$,
\begin{align}
f_\ga(x) = \sum_{\ell = 0}^\infty \frac{A_\ell}{\Gamma(d + \ell)}
\frac{x^{d+\ell-1} }{ {\bar{\gamma}}^{d+\ell} } \; .
\end{align}
We are mainly interested in the leading order term, and so we write the above as
\begin{align} \label{eq:HSNR:FirstOrderPDFBeta}
f_\ga(x) =
\frac{A_0}{\Gamma(d)}\frac{x^{d-1}}{\bg^d}+\mathcal{O}\left(\frac{1}{\bg^{d+1}}\right)
\; \; .
\end{align}
Note that despite the similarity with (\ref{eq:HSNR:FirstOrderPDFBeta0}),
interestingly, these results do not coincide upon taking $\beta \to
0$ in (\ref{eq:HSNR:FirstOrderPDFBeta}), due to the differences in $d$
and $A_0$. This seems to indicate that the double asymptotics
$\bar{\gamma} \to \infty$ and $\beta \to 0$ are non-commutative.

\section{Characterizing $A_0$ Through Painlev\'e $V$}\label{Sec:PV_Large_s}

In this section, we obtain the leading term of the large $s$ and large $\bg$
expansion (\ref{eq:CF:HSNR:MGFBetaNonZero}) from a Painlev\'e $V$
differential equation, thus demonstrating the accuracy of the
Coulomb Fluid approximation. We will focus in this section on the
case $\beta \neq 0$.

Recall that the moment generating function of SNR $\ga$, regarded a function of $s$ and $t$,
given by (\ref{eq:mgfMultipleIntegral}), reads,
\begin{align}
\mathcal{M}_\gamma(s,{t})=\frac{1}{\mathcal{K}_{n,\al}}\left(\frac{1}{1+cs}\right)^{n{\N}}
\frac{1}{n!}\int\limits_{[0,\infty)^n}\prod\limits_{1\leq
i<j\leq n}(x_j-x_i)^2\prod\limits_{k=1}^n x_k^\al
e^{-x_k}\left(\frac{{t}+x_k}{\frac{t}{1+cs}+x_k}\right)^{\N} dx_k,
\end{align}
where $\mathcal{K}_{n,\al}$ is a normalization constant in (\ref{eq:NormConst_n_al}). We also recall that $c$ and $t$ are given by  
\bea
c=\frac{\bg}{R\N},\qquad t=\frac{(1+\bg)N_R}{\ti{b}}, \quad\text{where}\quad \ti{b}:=\nu\bg,
\eea
respectively. So as
$\bg\rightarrow\infty$, we note that
$$
t\longrightarrow \frac{N_R}{\nu}.
$$
A simple computation shows that $\mathcal{M}_\ga(s,t)$ admits the following expansion for large $\bg s$:
\begin{footnotesize}
\bea
\mathcal{M}_\gamma(s,{t})&=&\frac{(R\N)^{n\N}}{\mathcal{K}_{n,\al}(\bg s)^{n\N}}\left(1-\frac{n\N}{\bg s}+\dots\right)
\nonumber\\&&
\times
\frac{1}{n!}\int\limits_{[0,\infty)^n}\prod\limits_{1\leq
i<j\leq n}(x_j-x_i)^2\prod\limits_{k=1}^{n} x_k^{\al-\N}
e^{-x_k}\left({t}+x_k\right)^{\N}\left(1-\frac{tR\N^2}{\bg sx_k}+\dots\right) dx_k \nonumber \\
\label{eq:PV:MGFExp_Ls}
&=&\frac{(R\N)^{n\N}}{D_n[w_{{\rm
dLag}}(\cdot,0,\al-\N,\N)]}\frac{1}{(\bg s)^{n\N}}D_n\big[w_{{\rm dLag}}(\cdot,t,\al-\N,\N)\big]
\nonumber\\&&
-\frac{R^{n\N+1}\N^{n\N+2}}{D_n[w_{{\rm
dLag}}(\cdot,0,\al-\N,\N)]}\frac{1}{(\bg s)^{n\N+1}}\Bigg[nD_n\big[w_{{\rm dLag}}(\cdot,t,\al-\N,\N)\big]
\nonumber\\&&
\quad+
\frac{t}{n!}\int\limits_{[0,\infty)^n}\prod\limits_{1\leq
i<j\leq n}(x_j-x_i)^2\left(\sum\limits_{l=1}^{n}x_l^{-1}\right)\prod\limits_{k=1}^{n} x_k^{\al-\N}
e^{-x_k}\left({t}+x_k\right)^{\N}dx_k\Bigg]+\mathcal{O}\left(\frac{1}{(\bg s)^{n\N+2}}\right),
\nonumber\\
\eea
\end{footnotesize}
where $w_{{\rm dLag}}(x,t,\al-\N,\N)$
is the deformation of the classical Laguerre weight, i.e.,
\bea\label{eq:HSNR:PV:DLag_def}
w_{{\rm dLag}}(x,t,\al-\N,\N)&=&
x^{\al-\N} e^{-x}(t+x)^{\N},\quad t>0, \al-\N>-1, \eea 
and $\mathcal{K}_{n,\al}=D_n[w_{{\rm
dLag}}(\cdot,0,\al-\N,\N)]$ is \emph{independent} of $t$.

The condition
$\al-\N>-1$ is required for the validity of the orthogonality
relation with respect to the generalized Laguerre weight $w_{\rm
Lag}^{(\al-\N)}(x)$ (see Ref.~\onlinecite{Szego1939}). This, in turn, will ensure
the validity of the first two terms in (\ref{eq:PV:MGFExp_Ls}),
while for the multiple integral in the third term to converge, the
condition $\al-\N>0$ is required.

For the problem at hand, $\al$ and $\N$ are integers, satisfying $\al\geq0$ and $\N>0$.
Therefore $\al\geq\N$ 
implies that the result for $A_0$ presented below is valid for $\al>0$.

Comparing the expansion (\ref{eq:PV:MGFExp_Ls}) with
(\ref{eq:CF:HSNR:MGFBetaNonZero}), we see that the diversity order
is \bea d&=&n\N, \eea and \bea
A_0&=&(R\N)^{n\N}
\frac{D_n[w_{{\rm
dLag}}(\cdot,t,\al-\N,\N)]}
{D_n[w_{{\rm
dLag}}(\cdot,0,\al-\N,\N)]}. 
\eea 
We see that $A_0$ is up to a
constant the Hankel determinant which generates a particular
Painlev\'e V and shows up in the single-user MIMO problem studied in
Ref.~\onlinecite{ChenMckay2010}.

We obtain $A_0,$ through a large $n$ expansion of $$\frac{D_n[w_{{\rm
dLag}}(\cdot,t,\al-\N,\N)]}
{D_n[w_{{\rm
dLag}}(\cdot,0,\al-\N,\N)]},$$ for $\alpha>0.$  We will see that $A_0$
precisely matches that obtained in (\ref{eq:HSNR:bt:A0nu}).

From Ref.~\onlinecite{ChenMckay2010} we learned that the logarithmic derivative
of $D_n[w_{\rm dLag}(\cdot,t,\alpha-N_s,N_s)]$ with respect to $t,$
\bea H_n(t)&:=& 
t\frac{d}{dt}\log 
D_n[w_{{\rm dLag}}(\cdot,t,\al-\N,\N)]
\nonumber\\
&=&
t\frac{d}{dt}\log 
\left(\frac{D_n[w_{{\rm dLag}}(\cdot,t,\al-\N,\N)]}
{D_n[w_{{\rm dLag}}(\cdot,0,\al-\N,\N)]}\right)
\eea 
satisfies the Painlev\'e V:
\bea\label{eq:PV:PVODE_t}
\Big(t H_{n}^{\prime\prime}\Big)^2&=&\Big[\big(t+2n+\al\big)H_{n}^\prime-H_{n}+n\N\Big]^2\nonumber\\
&&-4\Big(tH_{n}^\prime-H_{n}+n(n+\al)\Big)\Big(H_{n}^\prime\Big)\Big(H_{n}^\prime+\N\Big),
\eea where $'$ denote derivative w.r.t. $t$.  

We restrict to the case where $N_R<N_D$, for which $N_R=n$. The case where $N_R>N_D$ can be considered in a similar fashion as outlined in Section \ref{sec:HSNR_bt_neq_0}. Setting $\al=n\bt$ in
the above equation, where $\bt=\frac{m}{n}-1$ is a fixed positive
number, and with the change of variable \bea\label{eq:PV:t(nu)}
t=\frac{n}{\nu}, \eea an easy computation shows that
$$Y_n(\nu):=H_n(n/\nu)$$ satisfies
\bea\label{eq:PV:PVODE_nu}
\frac{\nu^4}{n^2}\Big(2Y_n^\prime +\nu Y_{n}^{\prime\prime}\Big)^2&=&
\bigg[\nu\Big((\bt+2)\nu+1\Big)Y_n^\prime+Y_n-n\N\bigg]^2\nonumber\\
&&+\frac{4\nu^2}{n^2}\Big(-\nu Y_n^\prime-Y_n+(1+\bt)n^2\Big)\Big(Y_{n}^\prime\Big)\Big(-\nu^2Y_{n}^\prime+n\N\Big),\qquad
\eea
where $'$ denotes derivative with respect to $\nu$.

We seek a solution for $Y_n(\nu)$ in the form
\bea\label{eq:HSNR:PV:G_n_Seek}
Y_n(\nu)&=&
np_{-1}(\nu)+p_0(\nu)
+\sum\limits_{j=1}^{\infty}\frac{p_{j}(\nu)}{n^j},
\eea
from which  $A_0(\nu)$ is found to be
\begin{small}
\bea\label{eq:PV:A_0_Formula}
A_0(\nu)&=&(R\N)^{n\N}\exp\left(-\int\limits_{\infty}^{\nu}
\frac{np_{-1}(\nu^\prime)+p_0(\nu^\prime)+\sum_{j=1}^{\infty}n^{-j}p_{j}(\nu^\prime)}{\nu^\prime}d\nu^\prime\right), \\
&\approx& (R\N)^{n\N}\exp\left(-\int\limits_{\infty}^{\nu}
\frac{np_{-1}(\nu^\prime)+p_0(\nu^\prime)}{\nu^\prime}d\nu^\prime\right)
\Bigg[1-\frac{1}{n}\int\limits_{\infty}^{\nu}\frac{p_{1}(\nu^\prime)}{\nu^\prime}d\nu^\prime+\dots\Bigg].
\nonumber\\
\eea
\end{small}


Substituting (\ref{eq:HSNR:PV:G_n_Seek}) into (\ref{eq:PV:PVODE_nu})
leads to (\ref{eq:PV:PVODE_nu}) taking the form \bea
c_{-2}n^2+c_{-1}n+c_0+\sum\limits_{j=1}^{\infty}c_{j}n^{-j}=0, \eea where
$c_{-2}$ depends on $p_{-1}(\nu)$ and its derivatives, and $c_i$,
$i=-1,0,1,2,\dots$ depend on $p_{-1}(\nu)$, $p_0(\nu)$ up to
$p_{i+1}(\nu)$ and their derivatives. Of course, each $c_i$ also
depends upon
$\nu$, 
$\N$ and $\bt$.
Assuming that the coefficient of $n^k$ is zero, we find that the equation $c_{-2}=0$ gives us
\bea
\bigg[\nu\Big((\bt+2)\nu+1\Big)p_{-1}^\prime(\nu)+p_{-1}(\nu)-\N\bigg]^2
&=&
4\nu^2(1+\bt)p_{-1}^\prime(\nu)\Big(\nu^2p_{-1}^\prime(\nu)+\N\Big).
\nonumber\\
\eea
With MAPLE, the solutions of the above ODE for $p_{-1}(\nu)$ are found to be
\bea
\label{eq:HSNR:PV:p1(nu)_1}p_{-1}(\nu)&=&\frac{\Big(-\nu\bt-1+\sqrt{\big(1+\nu\bt\big)^2+4\nu}\Big)\N}{2\nu},\\
\label{eq:HSNR:PV:p1(nu)_2}p_{-1}(\nu)&=&\frac{\Big(-\nu\bt-1-\sqrt{\big(1+\nu\bt\big)^2+4\nu}\Big)\N}{2\nu},\\
p_{-1}(\nu)&=&\left(2+\bt+\frac{1}{\nu}\right)C_2+\N+2\sqrt{(1+\bt)C_2(\N+C_2)},
\eea where $C_2$ is a constant of integration.\footnote{Note that the
constants of integration in (\ref{eq:HSNR:PV:p1(nu)_1}) and
(\ref{eq:HSNR:PV:p1(nu)_2}) are zero.}.

The equation $c_{-1}=0$ is a coupled differential equation involving both $p_{-1}(\nu)$ and $p_0(\nu)$ given by
\begin{small}
\bea\label{eq:HSNR:PV:c1=0} 
\left[ 2\nu^3\Big(p_{-1}^\prime(\nu)\Big)^2 +2{\nu}^{2}p_{-1}^\prime(\nu)p_{-1}(\nu)+\nu \left( 1-{\nu}\beta \right)p_0^\prime(\nu)+p_0(\nu)\right]  \Big(\N-\nu p_{-1}(\nu)\Big) \nonumber\\
 -\nu\, \bigg[  \Big( {\beta}
^{2}{\nu}^{3}+ 2\left( \beta+2 \right) {\nu}^{2}+\nu \Big) p_0^\prime(\nu) -2\,{\nu}^{2}p_{-1}(\nu)^2 + \Big( 1+ \left( \beta+2 \right)
\nu \Big) p_{{0}}(\nu)\bigg]p_{-1}^\prime(\nu)\nonumber\\ 
-2\,{\nu}^{5}\Big(p_{-1}^\prime(\nu)\Big)^3= \Big( \nu\, \left( \nu+1 \right)  \left(
\beta\,\nu+1 \right) p_0^\prime(\nu) -  \left( \nu-1 \right)p_0(\nu)  \Big) p_{{-1}}(\nu).
\nonumber\\
\eea
\end{small}
With $p_{-1}(\nu)$ given by (\ref{eq:HSNR:PV:p1(nu)_1}), chosen to
match the result from the Coulomb Fluid (\ref{eq:HSNR:bt:A0nu}), we
find that the first order ODE in $p_0(\nu)$ has the solution,
\bea\label{eq:HSNR:PV:p0(nu)} p_0(\nu)&=&
\frac{\Big(2+\bt+\nu\bt^2-\bt\sqrt{\big(1+\nu\bt\big)^2+4\nu}\Big)\nu\N^2}{2\left(\big(1+\nu\bt\big)^2+4\nu\right)}.
\eea We disregard the second and third solutions for $p_{-1}(\nu);$ as
these would lead to $p_0(\nu)$ which do not generate the $A_0$ in
agreement with that obtained from the Coulomb Fluid method.


Substituting $p_{-1}(\nu)$ from (\ref{eq:HSNR:PV:p1(nu)_1}) and $p_0(\nu)$ from (\ref{eq:HSNR:PV:p0(nu)}) into  $c_0=0$ gives,
\begin{small}
\bea\label{eq:HSNR:PV:pm1(nu)}
p_{1}(\nu)&=&
{\frac {{\N}{\nu}^{2} \Big[  \Big( {\beta}^{2} \left( 2+\beta \right) {
\nu}^{2}+ 2\left( \bt^2+2\beta+2 \right) \nu+2+\beta
\Big) {\N}^{2}- \left(1+ \beta \right) \nu \Big] }
{ \Big(\big(1+\nu\bt\big)^2+4\nu\Big)^{5/2}}}
\nonumber\\&&
-{\frac {{\N}^{3}{\nu}^{2} \left( 1+ \left( 2+\beta \right) \nu \right)
\beta}{ \Big(\big(1+\nu\bt\big)^2+4\nu\Big)^{2}}}.
\eea
\end{small}

Hence, 
$A_0(\nu)$ has a large $n$ expansion, \bea
A_0(\nu)&=&\;q_0(\nu)\Bigg[1+\frac{q_1(\nu)}{n}+\mathcal{O}\left(\frac{1}{n^2}\right)\Bigg],
\label{eq:A0_expand} \eea where
\begin{footnotesize}
\bea
\label{eq:PV:A_0_Explicit}q_0(\nu)&=&
\frac{\Big(2+\bt+\nu\bt^2+\bt\sqrt{\big(1+\nu\bt\big)^2+4\nu}\Big)^{\frac{{\N}}{2}({\N}-n\bt)}
\Big(1+2\nu+\nu\bt+\sqrt{\big(1+\nu\bt\big)^2+4\nu}\Big)^{\frac{n{\N}}{2}(2+\bt)}}
{\nu^{n{\N}}\Big(\big(1+\nu\bt\big)^2+4\nu\Big)^{\frac{{\N}^2}{4}}\Big(1+\bt\Big)^{\frac{n{\N}}{2}(2+\bt)}
\bt^{\frac{{\N}}{2}({\N}-2n\bt)}}\nonumber\\
&&\times\frac{\big(R{\N}\big)^{n{\N}}}{2^{\frac{{\N}}{2}(2n+{\N})}}\exp\left(-\frac{n{\N}}{2\nu}\Big(1+\nu\bt-
\sqrt{\big(1+\nu\bt\big)^2+4\nu}\Big)\right).
\eea
\end{footnotesize}

The leading term of $A_0$ in
(\ref{eq:A0_expand}) agrees \emph{precisely} with the $A_0$ computed
via the Coulomb Fluid method in (\ref{eq:HSNR:bt:A0nu}).


Using a method similar to the cumulant analysis of Section
\ref{Sec:Large_n_Corr}, we can also compute the first correction
term to $A_0$ (i.e., the quantity $q_1(\nu)$) using
(\ref{eq:HSNR:PV:pm1(nu)}), which reads
\begin{small}
\bea
q_1(\nu)&=&
\frac{{\N} \Big(  \left( 2{\N}^2-1 \right) {\beta}^{2}+
2\left(8{\N}^2-1\right) (1+\beta)\Big)}{24(1+\bt)}
\Bigg[\frac{1}{\bt}-\frac{\nu^2(\bt^2\nu+3\bt+6)}{\Big(\big(1+\nu\bt\big)^2+4\nu\Big)^{3/2}}\Bigg]
\nonumber\\&&
-{\frac { \left( 2\,\beta\,\nu+4\,\nu+1 \right) {\N}^{3}}{2\beta
\Big(\big(1+\nu\bt\big)^2+4\nu\Big) }}
-\frac { \left( 2+\beta
 \right) \left( 2{\N}^2-1\right) {\N}}{ 8\Big(\big(1+\nu\bt\big)^2+4\nu\Big)^{3/2} \left(1+\beta\right) }
\Bigg[(2+\bt)\nu+\frac{1}{3}\Bigg]
\nonumber\\&&
+\frac{(2+\bt)\N^3\nu^2}{\Big(\big(1+\nu\bt\big)^2+4\nu\Big)^{3/2}}.
\eea
\end{small}
Higher order corrections could also be obtained in a similar way.


\section{Conclusion}
In this paper, we have introduced two methods for characterizing the
received SNR distribution in a certain MIMO communication system
adopting AF relaying. We showed that the mathematical problem of
interest pertains to computing a certain Hankel determinant
generated by a particular two-time deformation of the classical
Laguerre weight.  By employing the ladder operator approach,
together with Toda-type evolution equations in the time variables,
we established an exact representation of the Hankel determinant in
terms of a double-time PDE, which reduces
to a Painlev\'e V in various limits. This result yields an exact and
fundamental characterization of the SNR distribution, through its
moment generating function.  Complementary to the exact
representation, we also introduced the linear statistics Coulomb
Fluid approach as an efficient way to compute very quickly the
asymptotic properties of the moment generating function for
sufficiently large dimensions (i.e., for sufficiently large numbers
of antennas). These results--which have only started to be used
recently in problems related to wireless communications and information theory--produced simple closed-form approximations for the
moment generating function. These were employed to yield simple
closed-form approximations for the error probability (for a class of
$M$-PSK digital modulation), which were shown via simulations to be
remarkably accurate, even for very small dimensions.

To further demonstrate the utility of our methodology, we employed
our asymptotic Coulomb Fluid characterization in conjunction with
the PDE representation to provide a
rigorous study of the cumulants of the SNR distribution. Starting
with a large-$n$ framework, we computed in closed-form the
finite-$n$ corrections to the first few cumulants. It was seen that
the Coulomb Fluid approach supplies the crucial initial conditions
which are instrumental in obtaining asymptotic expansions from the
PDE.  We also derived asymptotic
properties of the moment generating function when the average SNR
was sufficiently high, and in such regime extracted key performance
quantities of engineering interest, namely, the array gain and
diversity order.

\begin{acknowledgments}
N. S. Haq is supported by an EPSRC grant.
M. R. McKay is
supported by the Hong Kong Research Grants Council under Grant No.
616911.
\end{acknowledgments}

%
%
%
%
%
\newpage
\appendix


\section{Characterization of Hankel Determinant Using the Theory of Orthogonal Polynomials}
In this section, we describe the process by which we characterize
the Hankel determinant by using the theory of orthogonal polynomials
and their associated ladder operators. See
Refs.~\onlinecite{BasorChenEhrhardt,ChenIsmail2005,ChenIts2009,ChenMckay2010,Szego1939}
for the background to this theory.

Consider a sequence of polynomials $\{P_n(x)\}$ orthogonal with
respect to the weight function $w_{\rm AF}(x,T,t)$, given by
(\ref{eq:w(x,T,t)}) \bea w_{{\rm AF}}(x,T,{t})&=& x^\al
e^{-x}\left(\frac{{t}+x}{T+x}\right)^{\N},\qquad 0\leq x<\infty,
\eea i.e.,
\begin{equation}\label{defn:OrthRel}
\int\limits_{0}^\infty P_n(x)P_m(x)w_{\rm AF}(x,T,t)dx =
h_n\delta_{n,m} ,
\end{equation}
where $h_n$ is the square of the $L^2$ norm of $P_n(x)$.
The Hankel determinant (\ref{def:HankelDn}) is reduced to the following product:
\bea\label{Defn:D_n}
D_n[w_{\rm AF}]=\prod\limits_{j=0}^{n-1}\!h_j.
\eea
Our convention is to write $P_n(x)$ as
\begin{equation}\label{defn:Pn(x)}
P_n(x):= x^n +\P_1(n)x^{n-1}+\P_2(n)x^{n-2}+\dots+P_n(0).
\end{equation}
Hence, this implies that properties of the Hankel determinant may be
obtained by characterizing the class of polynomials which are
orthogonal with respect to $w_{{\rm AF}}(x,T,{t})$, over
$[0,\infty)$.
It is clear that the coefficients of the polynomial $P_n(x)$,
$\P_i(n)$, will depend on $T$, $t$, $\al$ and $\N$; for brevity, we
do not display this dependence.

From the orthogonality relation, the three term
recurrence relation follows:
\begin{equation}\label{def:ThreeTermRel}
xP_n(x)=P_{n+1}(x)+\alpha_nP_n(x)+\beta_nP_{n-1}(x), \qquad
n=0,1,2,\dots
\end{equation}
with initial conditions $$P_0(x)\equiv1,
\qquad\text{and}\qquad\bt_0P_{-1}(x)\equiv0.$$

The main aim is to determine these unknown recurrence coefficients
$\alpha_n$ and $\beta_n$ from the given weight. Substituting
(\ref{defn:Pn(x)}) into the three term recurrence relation results
in
\begin{equation}\label{Defn:alpha_n(p_n-p_n+1)}
\alpha_n=\P_1(n)-\P_1(n+1)
\end{equation}
where $\P_1(0):=0$. Taking a telescopic sum gives
\begin{equation}\label{Defn:SumAlpha}
\sum\limits_{j=0}^{n-1}\!\alpha_j=-\P_1(n) .
\end{equation}
Moreover, combining the orthogonality relationship
(\ref{defn:OrthRel}) with the three term recurrence relation leads
to \bea\label{defn:beta_n=h_n/n-1} \bt_n=\frac{h_n}{h_{n-1}} , \eea
which can also be expressed in terms of the Hankel determinant $D_n$
in (\ref{Defn:D_n}) through
\begin{equation}\label{defn:beta_n(D_i)}
\beta_n=\frac{D_{n+1}D_{n-1}}{D_n^2},
\end{equation}
since
$$ h_n=\frac{D_{n+1}}{D_n}.$$

\subsection*{Ladder Operators, Compatibility Conditions, and
Difference Equations}\label{sec:Ladder} In the theory of Hermitian
random matrices, orthogonal polynomials plays an important role,
since the fundamental object, namely, Hankel determinants or
partition functions, are expressed in terms of the associated $L^2$
norm, as indicated for example in (\ref{Defn:D_n}).  Moreover, as
indicated above, the Hankel determinants are intimately related to
the recurrence coefficients $\al_n$ and $\bt_n$ of the orthogonal
polynomials (for other recent examples, see Refs.~\onlinecite{BasorChen2009,BasorChenEhrhardt,ChenZhang2010,ForresterOrmerod2010}).

As we now show, there is a recursive algorithm that facilitates the
determination of the recurrence coefficients $\alpha_n$ and
$\beta_n$. This is implemented through the use of so-called ``ladder
operators'' as well as their associated compatibility conditions.
This approach can be traced back to Laguerre and Ref.~\onlinecite{Shohat1939}.
Recently, Magnus \cite{Magnus1995} applied ladder operators to
non-classical orthogonal polynomials associated with random matrix
theory and the derivation of Painlev\'e equations, while
Ref.~\onlinecite{TracyWidom1999} used the associated compatibility conditions
in the study of finite $n$ matrix models. See
Refs.~\onlinecite{BasorChenEhrhardt,ChenIsmail2005,ChenIts2009} for other
examples of the application of this approach.

From the weight function $w_{\rm AF}(x,T,t),$ one constructs the
associated potential $\textsf{v}(x)$ through \bea
\label{eq:v(x)}\textsf{v}(x)&=&-\log w_{\rm AF}(x)=x-\al\log x -
{\N}\log\left(\frac{{t}+x}{T+x}\right), \eea and therefore, \bea
\label{eq:vp(x)}\textsf{v}^\prime(x)&=&
1-\frac{\al}{x}-\frac{{\N}}{x+{t}}+\frac{{\N}}{x+T}. 
\eea 
As shown in Ref.~\onlinecite{ChenIsmail1997}, a pair of ladder operators, to be
satisfied by our orthogonal polynomials of interest, are expressed
in terms of $\textsf{v}(x)$ and are given by
\begin{equation}\label{Defn:LadderOp}\begin{aligned}
\left[\frac{d}{dx}+B_n(x)\right]P_n(x)=&\bt_nA_n(x)P_{n-1}(x),\\
\left[\frac{d}{dx}-B_n(x)-\textsf{v}^\prime(x)\right]P_{n-1}(x)=&-A_{n-1}(x)P_n(x),
\end{aligned}\end{equation}
where
\begin{equation}\label{Defn:A_nB_n}\begin{aligned}
A_n(x)=&
\frac{1}{h_n}\int\limits^{\infty}_{0}\frac{\textsf{v}^\prime(x)-\textsf{v}^\prime(y)}{x-y}P_n^2(y)w_{\rm AF}(y)dy,\\
B_n(x)=&
\frac{1}{h_{n-1}}\int\limits^{\infty}_{0}\frac{\textsf{v}^\prime(x)-\textsf{v}^\prime(y)}{x-y}P_n(y)P_{n-1}(y)w_{\rm
AF}(y)dy,
\end{aligned}\end{equation}
and where, for the sake of brevity, we have dropped the $t,T$
dependence in $w_{\rm AF}.$

Moreover, there are associated fundamental compatibility conditions
to be satisfied by $A_n(x)$ and $B_n(x)$, which are given by
\cite{ChenIsmail2005}
\begin{equation}\tag{$S_1$}
B_{n+1}(x) + B_n (x) = (x-\alpha_n)A_n(x) - \textsf{v}^\prime(x),
\end{equation}
\begin{equation}\tag{$S_2$}
1+(x-\alpha_n)[B_{n+1}(x)-B_n(x)]=\beta_{n+1}A_{n+1}(x) - \beta_n
A_{n-1}(x).
\end{equation}
These were initially derived for any polynomial $\textsf{v}(x)$ (see
Refs.~\onlinecite{Bauldry1990,BonanClark1990,Mhaskar1990}), and then were shown
to hold for all $x\in\mathbb{C}\cup\{\infty\}$ in greater generality.
\cite{ChenIsmail2005}

We now combine $(S_1)$ and $(S_2)$ as follows.  First, multiplying
$(S_2)$ by $A_n(x)$, it can be seen that the RHS is a first order
difference, while $(x-\alpha_n)A_n(x)$ on the LHS can be replaced by
$B_{n+1}(x)+B_n(x)+\textsf{v}^\prime(x)$ from $(S_1)$. Then, taking
a telescopic sum with initial conditions
\begin{equation*}
B_0(x)=A_{-1}(x)=0
\end{equation*}
leads to the useful identity
\begin{equation}\tag{$S_2^{\ \prime}$}
\sum\limits_{j=0}^{n-1}\!A_j(x)\ + \ B_n^{\ 2}(x) +\textsf{v}^\prime
(x)B_{n}(x) =  \beta_n A_n(x)A_{n-1}(x).
\end{equation}
The condition $(S_2')$ is of considerable interest, since it is
intimately related to the logarithm of the Hankel determinant. In
order to gain further information about the determinant, we need to
find a way to reduce the sum to fixed number of quantities; for
which, $(S_2^\prime)$ ultimately provides a way of going forward.

\begin{remark}
Since our $\textsf{v}^\prime(x)$ is a rational function of $x$, we
see that \bea
\label{eq:vpx-vpy}\frac{\textsf{v}^\prime(x)-\textsf{v}^\prime(y)}{x-y}&=&
\frac{\al}{xy}+\frac{{\N}}{(x+{t})(y+{t})}-\frac{{\N}}{(x+T)(y+T)},
\eea is also a rational function of $x$, which in turn implies that
$A_n(x)$ and $B_n(x)$ are rational functions of $x$. Consequently,
equating the residues of the simple and double pole at  $x=0$,
$x=-T$,  $x=-t$ on both sides of the compatibility conditions
$(S_1)$, $(S_2)$ and $(S_2^\prime)$, we obtain equations containing
numerous $n$, $T$ and ${t}$ dependant quantities; which we call the
``auxiliary variables'' (to be introduced below). The resulting
non-linear discrete equations are likely very complicated, but the
main idea is to express the recurrence coefficients $\al_n$ and
$\bt_n$ in terms of these auxiliary variables, and eventually take
advantage of the product representation (\ref{Defn:D_n}) to obtain
an equation satisfied by the logarithmic derivative of the Hankel
determinant.
\end{remark}

Now substituting (\ref{eq:vpx-vpy}) into (\ref{Defn:A_nB_n}) which define $A_n(x)$
and $B_n(x)$, and followed by integration by parts, we find
\bea
A_n(x)&=&\frac{R_n(T,{t})+1-R_n^\ast(T,{t})}{x}+\frac{R_n^\ast(T,{t})}{x+{t}}-\frac{R_n(T,{t})}{x+T},\\
B_n(x)&=&\frac{r_n(T,{t})-n-r_n^\ast(T,{t})}{x}+\frac{r_n^\ast(T,{t})}{x+{t}}-\frac{r_n(T,{t})}{x+T},
\eea
where
\begin{small}
\begin{equation}\label{eq:AuxVarMC}\begin{aligned}
R_n^\ast(T,{t})\equiv&\frac{{\N}}{h_n}\int\limits^{\infty}_{0}\frac{w_{\rm AF}(y)P_n^2(y)}{y+{t}}dy,&&&&&&&&&
r_n^\ast(T,{t})\equiv&\frac{{\N}}{h_{n-1}}\int\limits^{\infty}_{0}\frac{w_{\rm AF}(y)P_n(y)P_{n-1}(y)}{y+{t}}dy,\\
R_n(T,{t})\equiv&\frac{{\N}}{h_n}\int\limits^{\infty}_{0}\frac{w_{\rm AF}(y)P_n^2(y)}{y+T}dy,&&&&&&&&&
r_n(T,{t})\equiv&\frac{{\N}}{h_{n-1}}\int\limits^{\infty}_{0}\frac{w_{\rm AF}(y)P_n(y)P_{n-1}(y)}{y+T}dy,
\end{aligned}\end{equation}
\end{small}
are the auxiliary variables.

\subsection*{Difference Equations from Compatibility Conditions}
Inserting $A_n(x)$ and $B_n(x)$ into the compatibility conditions
$(S_1)$, $(S_2)$ and $(S_2^\prime)$,
and equating the residues as described above yields a system of 12
equations. Note that there are actually $14$ equations, the
compatibility conditions $(S_1)$ and $(S_2)$ also have
$\mathcal{O}(x^0)$ terms which yield equations that lead to $0=0$.

The compatibility conditions $(S_1)$ give the following set of
equations,
\bea
\label{sys:MCS1.1}r_{n+1}+r_n-r_{n+1}^\ast-r_n^\ast-2n-1&=&\alpha-\alpha_n(R_n+1-R_n^\ast),\\
\label{sys:MCS1.2}r_{n+1}^\ast+r_n^\ast&=&{\N}-(\alpha_n+{t})R_n^\ast,\\
\label{sys:MCS1.3}-r_{n+1}-r_n&=&-{\N}+(\alpha_n+T)R_n. \eea Similarly,
the compatibility condition $(S_2)$ gives the following set of
equations,
\bea
\label{sys:MCS2.1}-\al_n(r_{n+1}-r_n-1-r_{n+1}^\ast+r_n^\ast)&=&\bt_{n+1}-\bt_n
\nonumber\\&&
+\bt_{n+1}(R_{n+1}-R_{n+1}^\ast)-\bt_n(R_{n-1}-R_{n-1}^\ast),\qquad\\
\label{sys:MCS2.2}({t}+\al_n)(r_n^\ast-r_{n+1}^\ast)&=&\bt_{n+1}R_{n+1}^\ast-\bt_nR_{n-1}^\ast,\\
\label{sys:MCS2.3}(T+\alpha_n)(r_n-r_{n+1})&=&\beta_{n+1}R_{n+1}-\beta_nR_{n-1}.
\eea
To obtain a system of difference equations from the
compatibility condition $(S_2^\prime)$ is somewhat more complicated, but
carrying out the process, equating all respective residues in $(S_2^\prime)$ yields six equations.
The first three are obtained by equating the residues of the double
pole at $x=0$, $x=-t$ and $x=-T$ respectively:
\bea
\label{sys:MCS2p1.1}-2r_nr_n^\ast-(2n-{\N}+\al)r_n\nonumber\\
+(2n+{\N}+\al)r_n^\ast+n(n+\al)&=&\bt_n-\bt_n(R_nR_{n-1}^\ast+R_n^\ast R_{n-1})\nonumber\\
&&+\bt_n(R_n+R_{n-1})
-\bt_n(R_n^\ast+R_{n-1}^\ast),\\
\label{sys:MCS2p1.2} r_{n}^\ast(r_n^\ast-\N)&=&\beta_nR_n^\ast R_{n-1}^\ast, \\
\label{sys:MCS2p1.3}r_{n}(r_n-\N)&=&\beta_nR_nR_{n-1}, \eea while
the last three are given by equating the residues of the simple pole at $x=0$,
$x=-t$ and $x=-T$ respectively:
\begin{footnotesize}
\bea\label{sys:MCS2p2.1} \sum\limits_{j=0}^{n-1}
(R_j-R_j^\ast)+\frac{(2r_n^\ast-{\N})(r_n-n-r_n^\ast)-\al
r_n^\ast}{{t}}+\frac{({\N}-2r_n)(r_n-n-r_n^\ast)+\al r_n}{T}
+r_n-r_n^\ast\nonumber\\
=\left(\frac{1}{T}+\frac{1}{{t}}\right)\beta_n(R_n^\ast
R_{n-1}+R_nR_{n-1}^\ast)
+\frac{\beta_n(R_n^\ast+R_{n-1}^\ast-2R_n^\ast R_{n-1}^\ast)}{{t}}
-\frac{\beta_n(R_n+R_{n-1}+2R_nR_{n-1})}{T}, 
\nonumber\\
\eea
\end{footnotesize}
\bea\label{sys:MCS2p2.2} \sum\limits_{j=0}^{n-1}
R_j^\ast-\frac{(2r_n^\ast-{\N})(r_n-n-r_n^\ast)-\al r_n^\ast}{{t}}
+\frac{({\N}-2r_n)r_n^\ast+{\N}r_n}{T-{t}} +r_n^\ast\nonumber\\=
-\frac{\beta_n(R_n^\ast+R_{n-1}^\ast-2R_n^\ast
R_{n-1}^\ast)}{{t}}-\left(\frac{1}{{t}}+\frac{1}{T-{t}}\right)\beta_n(R_n^\ast
R_{n-1}+R_nR_{n-1}^\ast), \eea \bea\label{sys:MCS2p2.3}
\sum\limits_{j=0}^{n-1} R_j +\frac{({\N}-2r_n)(r_n-n-r_n^\ast)+\al
r_n}{T} +\frac{({\N}-2r_n)r_n^\ast+{\N}r_n}{T-{t}} +r_n\nonumber\\=
-\frac{\beta_n(R_n+R_{n-1}+2R_nR_{n-1})}{T}+\left(\frac{1}{T}-\frac{1}{T-{t}}\right)\beta_n(R_n^\ast
R_{n-1}+R_nR_{n-1}^\ast). \eea
\begin{remark}
It can be seen that (\ref{sys:MCS2p2.1}) is a combination of
(\ref{sys:MCS2p2.2}) and (\ref{sys:MCS2p2.3}), and serves as a
consistency check. We also see that (\ref{sys:MCS2p2.2}) and
(\ref{sys:MCS2p2.3}) respectively gives us the value of
$\sum_jR_j^\ast$ and $\sum_jR_j$ automatically in closed form. Hence
we can also obtain any linear combination of these sums in closed
form, which is a crucial step in obtaining a link between a linear
combination of the logarithmic partial derivatives of the Hankel
determinant, and the quantities $\bt_n$, $r_n$ and $r_n^\ast$. This
step would not be possible without $(S_2^\prime)$, also known as the
\emph{bilinear identity}.
\end{remark}
\subsection*{Analysis of Non-linear System}
Whilst some of the above difference equations
(\ref{sys:MCS1.1})-(\ref{sys:MCS2p2.3}) look rather complicated, our
aim is to manipulate these equations in such a way as to give us
insight into the recurrence coefficients $\al_n$ and $\bt_n.$  Our
dominant strategy is to always try to describe the recurrence
coefficients $\al_n$ and $\bt_n$ in terms of the auxiliary variables
$R_n$, $r_n$, $R_n^\ast$ and $r_n^\ast$.

The first thing we can do is to sum (\ref{sys:MCS1.1}) through
(\ref{sys:MCS1.3}) to get a simple expression for the recurrence
coefficient $\al_n$ in terms of the auxiliary variables $R_n$ and
$R_n^\ast$:
\bea
\label{eq:aln(RnRnast)v1}\al_n&=&TR_n-{t}R_n^\ast+\al+2n+1.
\eea
We can immediately see that by taking a telescopic sum of the above
from $j=0$ to $j=n-1$, and recalling equation
(\ref{Defn:SumAlpha}), gives rise to
\bea
\label{eq:p1n(SumRjRjast)}T\sum\limits_{j=0}^{n-1} R_j-{t}\sum\limits_{j=0}^{n-1}R_j^\ast+n(n+\al)&=&\sum\limits_{j=0}^{n-1}\al_j,
\nonumber\\
&=&-\textsf{p}_1(n).
\eea
Looking at $(S_2)$ next, we sum
equations (\ref{sys:MCS2.1}) through (\ref{sys:MCS2.3}) to get
\bea
\label{eq:S2Combined}\al_n=\bt_{n+1}-\bt_n+T(r_n-r_{n+1})+{t}(r_{n+1}^\ast-r_n^\ast).
\eea
Taking a telescopic sum of the above from $j=0$ to $j=n-1$ and
rearranging yields \bea \label{eq:btn(rnrnastp1n)}\bt_n&=&T
r_n-{t}r_n^\ast-\textsf{p}_1(n). \eea
We are now in a position to derive the following important lemma,
which describes the recurrence coefficients $\al_n$ and $\bt_n$ in
terms of the set of auxiliary variables:

\begin{lemma}
The quantities $\al_n$ and $\bt_n$ are given in terms of the
auxiliary variables $R_n$, $r_n$, $R_n^\ast$ and $r_n^\ast$ as \bea
\label{eq:aln(RnRnast)}\al_n&=&TR_n-{t}R_n^\ast+\al+2n+1,\\
\label{eq:btn(rnRnrnastRnast)}\bt_n&=&\frac{1}{1+R_n-R_n^\ast}
\bigg[(1+R_n)\frac{r_n^\ast(r_n^\ast-{\N})}{R_n^\ast}+(R_n^\ast-1)\frac{r_n(r_n-{\N})}{R_n}-2r_nr_n^\ast\nonumber\\
&&\hspace{34mm}-(2n-{\N}+\al)r_n+(2n+{\N}+\al)r_n^\ast+n(n+\al)\bigg].
\nonumber\\
\eea
\end{lemma}
\begin{proof}
Equation (\ref{eq:aln(RnRnast)}) is a restatement of equation
(\ref{eq:aln(RnRnast)v1}).

To obtain (\ref{eq:btn(rnRnrnastRnast)}), we use the equations
obtained from$(S_2^\prime)$. We eliminate $R_{n-1}^\ast$ and
$R_{n-1}$ from (\ref{sys:MCS2p1.1}) using (\ref{sys:MCS2p1.2}) and
(\ref{sys:MCS2p1.3}) respectively, and then rearrange to express
$\bt_n$ in terms of the auxiliary quantities.
\end{proof}

\subsection*{Toda Evolution}\label{sec:Toda}
In this section, $n$ is kept fixed while we vary two parameters in
the weight function (\ref{eq:w(x,T,t)}), namely $T$ and ${t}$. The
other parameters, $\al$ and ${\N}$ are kept fixed.

Differentiating (\ref{defn:OrthRel}), w.r.t. $T$ and ${t},$ for $m=n,$ gives
\bea
\label{eq:dTloghn}\partial_T(\log h_n)&=&-R_n,\\
\label{eq:dtloghn}\partial_{{t}}(\log h_n)&=&R_n^\ast
\eea
respectively. Then, from equation (\ref{defn:beta_n=h_n/n-1}), i.e.
$\bt_n=h_n/h_{n-1}$, it follows that
\bea
\label{eq:dTbtn}\partial_T\bt_n&=&\bt_n(R_{n-1}-R_n),\\
\label{eq:dtbtn}\partial_{{t}}\bt_n&=&\bt_n(R_n^\ast-R_{n-1}^\ast).
\eea
Applying $\partial_T$ and $\partial_{{t}}$ to the orthogonality
relation
\bea \int\limits^{\infty}_{0}
w_{\rm AF}(x,T,{t})P_n(x)P_{n-1}(x)dx=0
\nonumber
\eea
results in the following two relations:
\bea
\label{eq:dTp1n}\partial_T\textsf{p}_1(n)&=&r_n,\\
\label{eq:dtp1n}\partial_{{t}}\textsf{p}_1(n)&=&-r_n^\ast.
\eea

Note that in the above computations with $\partial_t$ and $\partial_T$ we must keep in mind that the coefficients
of $P_n(x)$ depend on $t$ and $T.$

Using the above two equations, we get that
\bea
(T\partial_T+{t}\partial_{{t}})\P_1(n)&=&Tr_n-{t}r_n^\ast,\nonumber\\
\label{eq:LAFp1(n)}&=&\P_1(n)+\bt_n,
\eea
where the second equality follows from (\ref{eq:btn(rnrnastp1n)}).

Using equation (\ref{Defn:alpha_n(p_n-p_n+1)}), i.e.
$\al_n=\textsf{p}_1(n)-\P_1(n+1)$, we get that
\bea
\label{eq:alT}\partial_T\al_n&=&r_n-r_{n+1},\\
\label{eq:alt}\partial_{{t}}\al_n&=&r_{n+1}^\ast-r_n^\ast.
\eea

Now, combining (\ref{eq:alT}), (\ref{eq:alt}),
(\ref{eq:btn(rnrnastp1n)}), and (\ref{Defn:alpha_n(p_n-p_n+1)}), and
similarly using (\ref{eq:aln(RnRnast)}), (\ref{eq:dTbtn}) and (\ref{eq:dtbtn}), we arrive at
the following lemma:

\begin{lemma}\label{Lem:TodaEqns}
The recurrence coefficients $\al_n$ and $\bt_n$ satisfy the
following partial differential relations: \bea
(T\partial_T+{t}\partial_{{t}}-1)\al_n&=&\bt_n-\bt_{n+1},\nonumber\\
(T\partial_T+{t}\partial_{{t}}-2)\bt_n&=&\bt_n(\al_{n-1}-\al_n).\nonumber
\eea
In the second of the above two equations, we eliminate $\al_n$ using the
first equation to derive a second order differential-difference
relation for $\bt_n$:
\bea
(T^2\partial_{TT}^2+2T{t}\partial_{T{t}}^2+{t}^2\partial_{{t}{t}}^2)\log\bt_n&=&\bt_{n-1}-2\bt_n+\bt_{n+1}-2,\nonumber
\eea
which is a two-variable generalization of the Toda molecule equations.~\cite{Sogo1993}
\end{lemma}

Now, before proceeding to examine the time-evolution behaviour of the Hankel determinant, we state the following lemmas regarding the auxiliary variables $R_n$, $R_n^\ast$, $r_n$ and $r_n^\ast$.
\begin{lemma}\label{Lem:AuxVarPDESys}
The auxiliary variables $R_n$, $R_n^\ast$, $r_n$ and $r_n^\ast$ satisfy the following first order PDE system:
\bea
\label{eq:AnF:AV_PDE_Sys_1}
T\partial_T R_n-t\partial_T R_n^\ast&=&\Big[TR_n-tR_n^\ast+2n+\al+T\Big]R_n+2r_n-\N,\\
\label{eq:AnF:AV_PDE_Sys_2}T\partial_t R_n-t\partial_t R_n^\ast&=&\Big[-TR_n+tR_n^\ast-2n-\al-t\Big]R_n^\ast-2r_n^\ast+\N,
\eea 
\begin{footnotesize}
\bea
\label{eq:AnF:AV_PDE_Sys_3}T\partial_Tr_n-t\partial_T r_n^\ast &=& r_n(r_n-\N)
\Bigg[
\frac{1}{R_n}-\frac{R_n^\ast-1}{1+R_n-R_n^\ast}
\Bigg]
-\frac{R_n(1+R_n)}{R_n^\ast(1+R_n-R_n^\ast)}r_n^\ast(r_n^\ast-\N)
\nonumber\\&&
+\frac{R_n}{1+R_n-R_n^\ast}\Big[2r_nr_n^\ast+(2n-\N+\al)r_n-(2n+\N+\al)r_n^\ast-n(n+\al)\Big],\\
\label{eq:AnF:AV_PDE_Sys_4}T\partial_tr_n-t\partial_t r_n^\ast &=& r_n^\ast(r_n^\ast-\N)
\Bigg[
-\frac{1}{R_n^\ast}+\frac{1+R_n}{1+R_n-R_n^\ast}
\Bigg]
+\frac{R_n^\ast(R_n^\ast-1)}{R_n(1+R_n-R_n^\ast)}r_n(r_n-\N)
\nonumber\\&&
-\frac{R_n^\ast}{1+R_n-R_n^\ast}\Big[2r_nr_n^\ast+(2n-\N+\al)r_n-(2n+\N+\al)r_n^\ast-n(n+\al)\Big].
\eea
\end{footnotesize}
\end{lemma}
\begin{proof}
To obtain (\ref{eq:AnF:AV_PDE_Sys_1}), we substitute in for $\al_n$ in (\ref{eq:alT}) using (\ref{eq:aln(RnRnast)}) to yield \bea
r_n-r_{n+1}&=&\partial_T\al_n,\nonumber\\
&=&R_n+T\partial_TR_n-t\partial_T R_n^\ast.
\eea
Eliminating $r_{n+1}$ in the above formula using (\ref{sys:MCS1.3}) gives
\bea
2r_n&=&\Big[1-\al_n-T\Big]R_n +T\partial_T R_n-t\partial_T R_n^\ast +\N.
\eea
Finally, we eliminate $\al_n$ again using (\ref{eq:aln(RnRnast)}), and then rearrange to obtain (\ref{eq:AnF:AV_PDE_Sys_1}).

We obtain (\ref{eq:AnF:AV_PDE_Sys_2}) in a similar method to (\ref{eq:AnF:AV_PDE_Sys_1}). This time, we first substitute in for $\al_n$ in (\ref{eq:alt}) using (\ref{eq:aln(RnRnast)}). We then proceed to eliminate $r_{n+1}^\ast$ using (\ref{sys:MCS1.2}), and finally eliminate $\al_n$ again using (\ref{eq:aln(RnRnast)}) to obtain (\ref{eq:AnF:AV_PDE_Sys_2}). 

To obtain (\ref{eq:AnF:AV_PDE_Sys_3}), we differentiate equation (\ref{eq:btn(rnrnastp1n)}) with respect to $T$. Following this, we substitute in for $\partial_T\bt_n$ and $\partial_T\P_1(n)$ using (\ref{eq:dTbtn}) and (\ref{eq:dTp1n}) respectively to yield
\bea
T\partial_Tr_n-t\partial_T r_n^\ast&=&\bt_n(R_{n-1}-R_n).
\eea
We then proceed to replace $\bt_nR_{n-1}$ by $r_n(r_n-\N)/R_n$ using (\ref{sys:MCS2p1.3}), and finally we eliminate $\bt_n$ in favor of $R_n$, $R_n^\ast$, $r_n$ and $r_n^\ast$ using (\ref{eq:btn(rnRnrnastRnast)}) to obtain equation (\ref{eq:AnF:AV_PDE_Sys_3}).

Similarly, we obtain (\ref{eq:AnF:AV_PDE_Sys_4}) by first differentiating (\ref{eq:btn(rnrnastp1n)}) with respect to $t$. Following this, we substitute in for $\partial_t\bt_n$ and $\partial_t\P_1(n)$ using (\ref{eq:dtbtn}) and (\ref{eq:dtp1n}) respectively to yield
\bea
T\partial_tr_n-t\partial_t r_n^\ast&=&\bt_n(R_{n}^\ast-R_{n-1}^\ast).
\eea
We then proceed to replace $\bt_nR_{n-1}^\ast$ by $r_n^\ast(r_n^\ast-\N)/R_n^\ast$ using (\ref{sys:MCS2p1.2}), and finally we eliminate $\bt_n$ in favor of $R_n$, $R_n^\ast$, $r_n$ and $r_n^\ast$ using (\ref{eq:btn(rnRnrnastRnast)}) to obtain equation (\ref{eq:AnF:AV_PDE_Sys_4}), completing our proof.
\end{proof}
From Lemma \ref{Lem:AuxVarPDESys}, the auxiliary variables $r_n$ and $r_n^\ast$ appear linearly in equations (\ref{eq:AnF:AV_PDE_Sys_1}) and (\ref{eq:AnF:AV_PDE_Sys_2}) respectively. By eliminating $r_n$ and $r_n^\ast$ from (\ref{eq:AnF:AV_PDE_Sys_3}) and (\ref{eq:AnF:AV_PDE_Sys_4}) using (\ref{eq:AnF:AV_PDE_Sys_1}) and (\ref{eq:AnF:AV_PDE_Sys_2}), we obtain  a second order PDE system for the auxiliary variables $R_n$ and $R_n^\ast$. This is stated in the following lemma:
\begin{lemma}
The auxiliary variables $R_n$ and $R_n^\ast$ satisfy the following second-order PDE system:
\begin{footnotesize}
\bea
0&=&
2 {T}\Big(T\partial_{TT}+t\partial_{Tt}\Big){R_n}-2t\Big( T\partial_{TT}+t\partial_{tt}\Big){R_{n}^{\ast}}
\nonumber\\&&
-{\frac{\Big((1+R_{n})(1-R_n^\ast)+R_n\Big)}{R_{{n}}(1+R_{{n}}-R_{n}^{\ast})}}
\bigg[
T\Big({\partial_{T}{R_n}}\Big)-t\Big(\partial_TR_n^\ast\Big)
\bigg]^2
+\frac{R_n(1+R_n)}{R_n^\ast(1+R_n-R_n^\ast)}
\bigg[
T\Big({\partial_{t}{R_n}}\Big)-t\Big(\partial_tR_n^\ast\Big)
\bigg]^2
\nonumber\\&&
+\frac{2R_{{n}}}{1+R_{{n}}-R_{n}^{\ast}}
\bigg[T\Big(\partial_TR_n\Big)-t\Big(\partial_TR_n^\ast\Big)\bigg]
\bigg[T\Big(\partial_{t}{R_n}\Big)-t\Big(\partial_t R_n^\ast\Big)\bigg]
\nonumber\\&&
+2T(
tR_{n}^{\ast}+1)\Big(\partial_{T}{R_n}\Big)+2t(TR_n+1)\Big(\partial_tR_n\Big)
-2t^2\bigg[R_n^\ast\Big(\partial_T R_n^\ast\Big)+R_n\Big(\partial_tR_n^\ast\Big)\bigg]
-2t(T-t)\Big(\partial_TR_n^\ast\Big)
\nonumber\\&&
-2R_n(TR_n-tR_n^\ast+T)\bigg(TR_n-tR_n^\ast+2n+\al+1+\frac{t}{2}\bigg)
-T(T-t)R_n(R_n+1)
\nonumber\\&&
-{\frac {\Big({R_{{n}}}(R_n -R_{n}^{\ast})-R_{{n}}^{\ast}\Big) {N_{{s}}}^2 }{R_{n}^{\ast}R_{{n}}}}-{\frac {\al^2 R_{{n}}}{1+R_{{n}}-R_{n}^{\ast}}},
\eea
\bea
0&=&
2T\Big(T\partial_{Tt}+t\partial_{tt}\Big)R_n-2t(T+t)\Big(\partial_{tt}R_n^\ast\Big)
\nonumber\\&&
+\frac{R_n^\ast(1-R_n^\ast)}{R_n(1+R_n-R_n^\ast)}
\bigg[
T\Big({\partial_{T}{R_n}}\Big)-t\Big(\partial_TR_n^\ast\Big)
\bigg]^2
+\frac{\Big((1+R_n)(1-R_n^\ast)-R_n^\ast\Big)}{R_n^\ast(1+R_n-R_n^\ast)}
\bigg[
T\Big({\partial_{t}{R_n}}\Big)-t\Big(\partial_tR_n^\ast\Big)
\bigg]^2
\nonumber\\&&
-\frac{2R_{{n}}^\ast}{1+R_{{n}}-R_{n}^{\ast}}
\bigg[T\Big(\partial_TR_n\Big)-t\Big(\partial_TR_n^\ast\Big)\bigg]
\bigg[T\Big(\partial_{t}{R_n}\Big)-t\Big(\partial_t R_n^\ast\Big)\bigg]
\nonumber\\&&
+2T(tR_n^\ast+1)\Big(\partial_TR_n^\ast)+2t(TR_n+1)\Big(\partial_tR_n^\ast\Big)
-2T^2\bigg[R_n^\ast\Big(\partial_T R_n\Big)+R_n\Big(\partial_tR_n\Big)\bigg]
-2T(T-t)\Big(\partial_t R_n\Big)
\nonumber\\&&
+2R_n^\ast(TR_n-tR_n^\ast+t)\bigg(TR_n-tR_n^\ast+2n+\al+1+\frac{T}{2}\bigg)
-t(T-t)R_n^\ast(1-R_n^\ast)
\nonumber\\&&
+{\frac {\Big({R_{{n}}^\ast}(R_n -R_{n}^{\ast})-R_{{n}}\Big) {N_{{s}}}^2 }{R_{n}^{\ast}R_{{n}}}}+{\frac {\al^2 R_{{n}}^\ast}{1+R_{{n}}-R_{n}^{\ast}}}.
\eea
\end{footnotesize}
\end{lemma}

\subsection*{Toda Evolution of Hankel Determinant}
Recall that the moment generating function is related to the Hankel
determinant through equation (\ref{eq:Mgf(Dn)}):
\begin{equation*}
\mathcal{M}_\ga(T,{t})=\frac{1}{D_n[w_{\rm Lag}^{(\al)}(\cdot)]}\left(\frac{T}{{t}}\right)^{n{\N}}D_n(T,{t}),
\end{equation*}
while 
the Hankel determinant is related to $h_n(T,t)$ by the relation $D_n(T,t)=\prod_{j=0}^{n-1}h_j(T,t)$.

We compute the partial derivatives of the Hankel determinant by taking a telescopic sum from $j=0$ to $j=n-1$ of the
partial derivatives of $\log h_n$, i.e. equations (\ref{eq:dTloghn})
and (\ref{eq:dtloghn}). Using (\ref{Defn:D_n}), we then obtain \bea
\label{eq:dTlogDn}\partial_T(\log D_n)&=&-\sum\limits_{j=0}^{n-1}R_j,\\
\label{eq:dtlogDn}\partial_{{t}}(\log
D_n)&=&\sum\limits_{j=0}^{n-1}R_j^\ast.
\eea
Hence we now have expressions of the partial derivatives of $D_n(T,{t})$ in terms of the
auxiliary variables $R_n$ and $R_n^\ast$.

At this point, recalling equation (\ref{def:HnIntro}) that
\bea
H_n(T,{t}) =(T\partial_T+{t}\partial_{{t}})\log D_n(T,{t}).
\eea
This quantity related to the logarithmic derivative of the Hankel
determinant is of special interest as our goal is to find the
PDE that $H_n(T,{t})$ satisfies.
From the definition of $H_n(T,{t})$, along with (\ref{eq:dTlogDn})
and (\ref{eq:dtlogDn}), we find that
\bea
\label{eq:LlogDn(SumRnSumRnast)}H_n(T,{t})&=&-T\sum\limits_{j=0}^{n-1}R_j+{t}\sum\limits_{j=0}^{n-1}R_j^\ast,\\
\label{eq:LlogDn(p1n)}&=&\textsf{p}_1(n)+n(n+\al),
\eea where the
second equality is a result of (\ref{eq:p1n(SumRjRjast)}).

At this point, we may express the fundamental quantity $H_n(T,{t})$
in terms of the auxiliary variables and $\beta_n$.  This is given in
the following key lemma.

\begin{lemma}
 \bea
\label{eq:Hn(rnrnastRnRnast)}H_n(T,{t})&=&2r_nr_n^\ast+(2n-{\N}+\al+T)r_n-(2n+{\N}+\al+{t})r_n^\ast\nonumber\\
&&-(1+R_n)\frac{r_n^\ast(r_n^\ast-{\N})}{R_n^\ast}+(1-R_n^\ast)\frac{r_n(r_n-{\N})}{R_n}+\bt_nR_n-\bt_nR_n^\ast.\nonumber\\
\eea
\end{lemma}
\begin{proof}
Using (\ref{sys:MCS2p2.2}) and (\ref{sys:MCS2p2.3}), we obtain
the sum
$$
{t}\sum\limits_{j=0}^{n-1}R_j^\ast-T\sum\limits_{j=0}^{n-1}R_j
$$
in closed form. All that remains is to eliminate $R_{n-1}^\ast$ and
$R_{n-1}$ from this equation. We do this using (\ref{sys:MCS2p1.2})
and (\ref{sys:MCS2p1.3}).
\end{proof}

To continue, we apply $t\partial_T+{t}\partial_{{t}}$ to equation
(\ref{eq:LlogDn(p1n)}), and find that \bea
\label{eq:LAFHn(p1(n)btn)}(T\partial_T+{t}\partial_{{t}})H_n&=&\P_1(n)+\bt_n,
\eea where we have used (\ref{eq:LAFp1(n)}) to replace
$(T\partial_T+{t}\partial_{{t}})\P_1(n)$ by $\P_1(n) + \beta_n.$

We see that equations (\ref{eq:LlogDn(p1n)}) and
(\ref{eq:LAFHn(p1(n)btn)}) can be regarded as a set of simultaneous
equations for $\P_1(n)$ and $\bt_n$. Solving for $\P_1(n)$ and
$\bt_n$, we find that
$$
\P_1(n)=H_n-n(n+\al),
$$
and \bea \label{eq:btn(Hn)} \bt_n =
(T\partial_T+{t}\partial_{{t}})H_n-H_n+n(n+\al). \eea

Before completing the proof of Theorem \ref{Thm:Hn(Tt0)}, we show
here that our Hankel determinant $D_n(T,t)$ is related to the
$\tau$-function of a Toda PDE.

Substituting the definition (\ref{def:HnIntro}) into
(\ref{eq:btn(Hn)}) and together with (\ref{defn:beta_n(D_i)}), we
find

%
$$
(T\partial_T+{t}\partial_{{t}})^2\log
D_n-(T\partial_T+{t}\partial_{{t}})\log
D_n+n(n+\al)=\frac{D_{n+1}D_{n-1}}{D_n^2}.
$$
Defining
$$
\ti{D}_n(T,{t}) := (T{t})^{-\frac{n(n+\al)}{2}}D_n(T,{t}),
$$
after some computations, $\ti{D}_n(T,{t})$ satisfies the following
equation: \bea\label{eq:TodaMolEqnHigher}
\left(\frac{T}{{t}}\partial_{TT}^2+2\partial_{T{t}}^2+\frac{{t}}{T}\partial_{{t}{t}}^2\right)\log
\ti{D}_n(T,{t})&=& \frac{\ti{D}_{n+1}\ti{D}_{n-1}}{\ti{D}_n^2} \; .
 \eea
The above equation is a two parameter generalization of the Toda
molecule equation,~\cite{Sogo1993} and hence we identify
$\ti{D_n}(T,{t})$ as the corresponding $\tau$-function of the
two-parameter Toda equations.

A reduction may be obtained by writing $\ti{D_n}(T,{t})$ as
$$
\ti{D}_n(T,{t})=\left(\frac{T}{{t}}\right)^{\frac{n(n+\al)}{2}}\Phi_n(T),
$$
and equation (\ref{eq:TodaMolEqnHigher}) is reduced to a
$1$-parameter Toda equation,
$$
\partial_{TT}^2\log\Phi_n(T)=\frac{\Phi_{n+1}(T)\Phi_{n-1}(T)}{\Phi_n(T)^2}.
$$

\subsection*{Partial Differential Equation for $H_n(T,{t})$}
We differentiate (\ref{eq:LlogDn(p1n)}) with respect to $T$ and
${t}$, and make use of the expressions for $\partial_T\P_1(n)$ and
$\partial_{{t}}\P_1(n)$, i.e., (\ref{eq:dTp1n}) and (\ref{eq:dtp1n})
respectively, to give \bea
\label{eq:dTHn}\partial_TH_n&=&r_n,\\
\label{eq:dtHn}\partial_{{t}}H_n&=&-r_n^\ast \; . \eea Thus, we now
have expressed $r_n$ and $r_n^\ast$ in terms of the partial
derivatives of $H_n$.

Next we derive representations of $R_n$ and $R_n^\ast$ in terms of
$H_n$ and its partial derivatives.

The idea, following from previous work,
\cite{BasorChen2009,BasorChenEhrhardt,ChenZhang2010,ChenDai2010} is to express  $R_n$
and $1/R_n$ in terms $H_n$ and its derivatives w.r.t. $t$ and $T.$ Similarly, for $R_n^*$ and $1/R_n^*.$

We start by re-writing (\ref{sys:MCS2p1.3}) as
$$
\beta_n\;R_{n-1}=\frac{r_n(r_n-N_s)}{R_n},
$$
and substitute into (\ref{eq:dTbtn}) to arrive at
$$
\partial_T\bt_n=\frac{r_n(r_n-N_s)}{R_n} - \bt_n R_n,
$$
which is a linear equation in $R_n$ and $1/R_n$.

Substituting  $r_n$ given by (\ref{eq:dTHn}) and $\bt_n$ given by
(\ref{eq:btn(Hn)}) into the above equation produces \bea
(T\partial_{TT}^2+{t}\partial_{T{t}}^2)H_n&=&\frac{(\partial_TH_n)\big(\partial_TH_n-{\N}\big)}
{R_n}
-\Big(T(\partial_TH_n)+{t}(\partial_{{t}}H_n)-H_n+n(n+\al)\Big)R_n.
\nonumber\\
\eea Going through a similar process we find, using
(\ref{sys:MCS2p1.2}), (\ref{eq:btn(Hn)}) and (\ref{eq:dtHn}),
%
\bea
(T\partial_{T{t}}^2+{t}\partial_{{t}{t}}^2)H_n&=&-\frac{(\partial_{{t}}H_n)\big(\partial_{{t}}H_n+{\N}\big)}{R_n^\ast}
+\Big(T(\partial_TH_n)+{t}(\partial_{{t}}H_n)-H_n+n(n+\al)\Big)R_n^\ast.
\nonumber\\
\eea
The above two equations  are quadratic equations in
$R_n$ and $R_n^\ast$. Solving for them leads to
\bea
\label{eq:Rn(Hn)}2\Big(T(\partial_TH_n)+{t}(\partial_{{t}}H_n)-H_n+n(n+\al)\Big)R_n&=&-(T\partial_{TT}^2+{t}\partial_{T{t}}^2)H_n\pm A_1(H_n),\\
\label{eq:Rnast(Hn)}2\Big(T(\partial_TH_n)+{t}(\partial_{{t}}H_n)-H_n+n(n+\al)\Big)R_n^\ast&=&
(T\partial_{T{t}}^2+{t}\partial_{{t}{t}}^2)H_n\pm A_2(H_n),
 \eea
where $A_1(H_n)$ and $A_2(H_n)$ are defined by (\ref{defn:A1(Hn)}) and (\ref{defn:A2(Hn)}) respectively, where we have left out the notation that indicates that $A_1$ and $A_2$ also depends upon the partial derivatives of $H_n$.

In the last step we substitute (\ref{eq:dTHn}), (\ref{eq:dtHn}),
(\ref{eq:btn(Hn)}), $R_n$ given by (\ref{eq:Rn(Hn)}) and $R_n^\ast$
given by (\ref{eq:Rnast(Hn)}) into (\ref{eq:Hn(rnrnastRnRnast)}).
After some simplification, we obtain the PDE (\ref{eq:PDE(Hn)Intro}), completing the proof of Theorem
\ref{Thm:Hn(Tt0)}.

\section{Some Relevant Integral Identities}\label{App:Int}
The following integral identities, taken from Refs.~\onlinecite{ChenMckay2010,GradRyzhJeff2007},
are referenced for use in the Coulomb fluid derivations.

Let $W(t)=\sqrt{(t+a)(t+b)}$:
\begin{eqnarray}
\label{AInt:1/sqrt}
\int\limits^{b}_{a}\frac{dx}{\sqrt{(b-x)(x-a)}}&=&\pi,\\
\label{AInt:1/x+tsqrt}
\int\limits^{b}_{a}\frac{dx}{(x+t)\sqrt{(b-x)(x-a)}}&=&\frac{\pi}{W(t)},\\
\label{AInt:x/sqrt}
\frac{1}{2\pi}\int\limits^{b}_{a}\frac{x dx}{\sqrt{(b-x)(x-a)}}&=&\frac{(a+b)}{4},\\
\label{AInt:Psqrt/(x+t)}
\textsf{P}\int\limits_a^b\frac{\sqrt{(b-y)(y-a)}}{(y-x)(y+t)}dy&=&\pi\left(\frac{W(t)}{x+t}-1\right),\\
\label{AInt:log(x+t)/sqrt}
\frac{1}{2\pi}\int\limits_a^b\frac{\log(x+t)}{ \sqrt{(b-x)(x-a)}}dx&=&\log\left(\frac{\sqrt{t+a}+\sqrt{t+b}}{2}\right),\\
\label{AInt:log(x+t)/xsqrt}
\frac{1}{2\pi}\int\limits_a^b\frac{\log(x+t)}{x \sqrt{(b-x)(x-a)}}dx&=&\frac{1}{2\sqrt{ab}}\log\left(\frac{\Big(\sqrt{ab}+W(t)\Big)^2-t^2}{\Big(\sqrt{a}+\sqrt{b}\Big)^2}\right),\\
\label{AInt:xlog(x+t)/sqrt}
\frac{1}{2\pi}\int\limits_a^b\frac{x\log(x+t)}{\sqrt{(b-x)(x-a)}}dx&=&\frac{\Big(\sqrt{t+a}-\sqrt{t+b}\Big)^2}{4},\nonumber\\
&&+\frac{(a+b)}{4}\log\left(\frac{\Big(W(t)+t\Big)^2-ab}{4t}\right),\\
\label{AInt:log(x+t)/(x+t)sqrt}\frac{1}{2\pi}\int\limits_a^b\frac{\log(x+t)}{(x+t)\sqrt{(b-x)(x-a)}}dx&=&-\frac{1}{W(t)}\log\left(\frac{1}{2\sqrt{t+a}}+\frac{1}{2\sqrt{t+b}}\right),\;
\end{eqnarray}
\begin{small}
\begin{eqnarray}
\label{AInt:1/xsqrt(c>0)}
\int\frac{dx}{x\sqrt{a^\prime x^2+b^\prime x+c^\prime} }&=&
-\frac{1}{\sqrt{c^\prime}}\log\left(\frac{2c^\prime+b^\prime x+2\sqrt{c^\prime}\sqrt{a^\prime x^2+b^\prime x+c^\prime}}{x}\right),
\\&&
\text{where}\;(c^\prime>0),
\nonumber\\
&=&\frac{1}{\sqrt{c^\prime}}\log\left(\frac{x}{2c^\prime+b^\prime x}\right), \;\text{where}\; (c^\prime>0,b^{\prime2}=4a^\prime c^\prime).\qquad
\end{eqnarray}
\end{small}
Finally, note the well-known Schwinger parameterization
\bea
\label{AInt:Schwinger}
\log(A+B)&=&\log A +\int\limits_0^1\frac{Bd\eta}{A+\eta B}.
\eea

\begin{lemma}
From the above integrals, we have
\begin{small}
\begin{eqnarray}
\label{AInt:log(x+t2)/(x+t1)sqrt}\frac{1}{2\pi}\int\limits_a^b\frac{\log(x+t_2)}{(x+t_1)\sqrt{(b-x)(x-a)}}dx&=&\frac{1}{2W(t_1)}\log\left(\frac{\Big(W(t_1)+W(t_2)\Big)^2-(t_1-t_2)^2}{\Big(\sqrt{t_1+a}+\sqrt{t_1+b}\Big)^2}\right).\qquad\quad
\end{eqnarray}
\end{small}
\end{lemma}

\begin{proof}
We rewrite $\log(x+t_2)$ using (\ref{AInt:Schwinger}) and then use (\ref{AInt:1/x+tsqrt}) to obtain
\begin{eqnarray*}
LHS_{(\ref{AInt:log(x+t2)/(x+t1)sqrt})}
&=&
\frac{\log t_2}{2\sqrt{(t_1+a)(t_1+b)}}+\frac{1}{2\pi}\int\limits_0^1\int\limits_a^b\!\frac{x dx d\eta}{(t_2+x\eta)(x+t_1)\sqrt{(b-x)(x-a)}}.
\end{eqnarray*}
Now we use the partial fraction decomposition
\begin{equation*}
\frac{x}{(t_2+x\eta)(x+t_1)}=\frac{t_1}{(t_1\eta-t_2)(x+t_1)}-\frac{t_2}{(t_1\eta-t_2)(t_2+x\eta)},
\end{equation*}
and then integrate over the $x$ variable using (\ref{AInt:1/x+tsqrt}) to get
\begin{eqnarray}
LHS_{(\ref{AInt:log(x+t2)/(x+t1)sqrt})}&=&
\frac{\log t_2}{2\sqrt{(t_1+a)(t_1+b)}}
\nonumber\\&&
+\frac{1}{2}\int\limits_0^1\frac{d\eta}{t_1\eta-t_2}\left(\frac{t_1}{\sqrt{(t_1+a)(t_1+b)}}-\frac{t_2}{\sqrt{(t_2+a\eta)(t_2+b\eta)}}\right),\nonumber\\
\label{AIntp:=integrand1}&=&
\frac{\log(t_2-t_1)}{2\sqrt{(t_1+a)(t_1+b)}}
-\frac{1}{2}\lim\limits_{\epsilon\to0}\int\limits^{1\pm\epsilon}_0\!\frac{t_2d\eta}{(t_1\eta-t_2)\sqrt{(t_2+a\eta)(t_2+b\eta)}}.
\nonumber\\
\end{eqnarray}

In our problem $t_2={t}^\prime$, $t_1=T^\prime$ or $t_2={T}^\prime$,
$t_1=t^\prime$, hence
$t_2/t_1=({t}^\prime/T^\prime)^{\pm1}=({t}/T)^{\pm1}=(1+cs)^{\pm1}$.
For the case $s=0$, there exists another pole within the integrand,
and so we replace the $\int\limits^1_0\dots$ with
$\int\limits^{1\pm\epsilon}_0\dots$ so that we may invoke
(\ref{AInt:1/xsqrt(c>0)}).

To evaluate the remaining integral in (\ref{AIntp:=integrand1}), we
first make the change of variable $y=t_1\eta-t_2$, giving
\begin{footnotesize}
\begin{eqnarray}\label{AInt:ProofIntChangeVar}
\frac{1}{2}\lim\limits_{\epsilon\to0}\int\limits^{1\pm\epsilon}_0\!\frac{t_2d\eta}{(t_1\eta-t_2)\sqrt{(t_2+a\eta)(t_2+b\eta)}}&=&
\frac{1}{2}\lim\limits_{\epsilon\to0}\int\limits^{t_1-t_2\pm\epsilon t_1}_{-t_2}\frac{(t_2/t_1)dy}{y\sqrt{(t_2(t_1+a)+ay)(t_2(t_1+b)+by)}}.
\end{eqnarray}
\end{footnotesize}
Within the square root term, we have a quadratic function in $y$, given by
$$
ab y^2+ \Big((a+b)t_1t_2+2abt_2\Big)y+t_2^2(t_1+a)(t_1+b).
$$
For our problem, where $a$ and $b$ are given by (\ref{eq:CF:(a)(b)[bt]}) and $\beta\geq0$, we see that the discriminant of the quadratic form is given by
$$t_1^2(b-a)^2>0,$$ while the constant term is
$$t_2^2\Big(t_1^2+2t_1(2+\beta)+\beta^2\Big)>0.$$
Hence we may invoke (\ref{AInt:1/xsqrt(c>0)}) and then take the limit $\epsilon\to0$ to get
\begin{footnotesize}
\begin{eqnarray}
RHS_{(\ref{AInt:ProofIntChangeVar})}&=&
-\frac{1}{\sqrt{(t_1+a)(t_1+b)}}\Bigg(-\log(t_2-t_1)\nonumber\\
&&+\log\left(\frac{2\sqrt{(t_1+a)(t_1+b)}\sqrt{(t_2+a)(t_2+b)}+2ab+(t_2+t_1)(a+b)+2t_2t_1}{2\sqrt{(t_1+a)(t_1+b)}+2t_1+a+b}\right)\Bigg).\qquad\qquad
\end{eqnarray}
\end{footnotesize}
Substituting this back into (\ref{AIntp:=integrand1}) and after some algebra, we get (\ref{AInt:log(x+t2)/(x+t1)sqrt}).
\end{proof}
\section{Differential Equations For Large Scale Corrections To Cumulants}\label{App:LNCorrDiffEqns}
\subsection{$\kappa_2({t^\prime})$}
The correction terms $f_i({t^\prime})$ for $i=1,2,3,4$ satisfy the following differential equations:
\begin{small}
\bea
\label{Appdiffeqn:f1}
\frac{\big(({t^\prime}+\bt)^2+4{t^\prime}\big)^{9/2}}{2{t^\prime}(1+\bt)}\frac{df_1(t^\prime)}{dt^\prime}&=&
l_2^{(1)}(t^\prime),\\
\label{Appdiffeqn:f2}
\frac{\big(({t^\prime}+\bt)^2+4{t^\prime}\big)^{6}}{2\N{t^\prime}(1+\bt)}\frac{df_2({t^\prime})}{dt^\prime}&=&
l_2^{(2)}(t^\prime),\\
\label{Appdiffeqn:f3}
\frac{\big(({t^\prime}+\bt)^2+4{t^\prime}\big)^{15/2}}{2{t^\prime}(1+\bt)}\frac{df_3({t^\prime})}{dt^\prime}&=&
l_2^{(3)}(t^\prime),\\
\label{Appdiffeqn:f4}
\frac{\big(({t^\prime}+\bt)^2+4{t^\prime}\big)^{9}}{2\N{t^\prime}(1+\bt)}\frac{df_4({t^\prime})}{dt^\prime}&=&
l_2^{(4)}(t^\prime).
\eea
\end{small}
The functions $l_2^{(i)}(t^\prime)$, $i=1,2,3,4$ are given by 
\begin{small}
\bea
\label{Appdiffeqn:f:l_2^1}
l_2^{(1)}(t^\prime)
&=&
3{{t^\prime}}^{4}-3\left(\beta+2 \right) {{t^\prime}}^{3}-2\left( 6{
\beta}^{2}+\beta+1 \right) {{t^\prime}}^{2}-3 \left( {\beta}+2 \right) \bt^2{{t^\prime}}+3\bt^4.\\
\label{Appdiffeqn:f:l_2^2}
l_2^{(2)}(t^\prime)&=&
-16{{t^\prime}}^{6}-2\left( \beta+2 \right) {{t^\prime}}^{5}+3\left( 23\,{
\beta}^{2}-8\,\beta -8\right) {{t^\prime}}^{4}
+4(\bt+2) \left( 15\bt^2-2\bt-2\right) {{t^\prime}}^{3}
\nonumber\\&&
-2{\beta}^{2} \left(7\bt^2 -48\,\beta
-48\right) {{t^\prime}}^{2}
-18{\beta}^{4} \left( \beta+2
 \right) {t^\prime}+{\beta}^{6}.\\
\label{Appdiffeqn:f:l_2^3}
l_2^{(3)}(t^\prime)
&=&
80{{t^\prime}}^{8}-60\left( \beta+2 \right) {{t^\prime}}^{7}-3\left(217\bt^2 -47\,
\beta-47 \right) {{t^\prime}}^{6}
\nonumber\\&&
-(\bt+2) \left(497\bt^2-55\bt-55\right) {{t^\prime}}^{5}
+3\left(175\,{
\beta}^{4}-499\,{\beta}^{3}-481\,{\beta}^{2}+36\,\beta+18 \right) {{t^\prime}}^{4}
\nonumber\\&&
+15(\bt+2) \left(42\bt^2-29\bt-29
\right) {{t^\prime}}^{3}\bt^2
+5\,{\beta}^{4} \left( 7\,{\beta}^{2}+192\bt+192
\right) {{t^\prime}}^{2}
\nonumber\\&&
-81\left( \beta+2 \right) {\beta}^{6}{t^\prime}+3{\beta}^
{8}.\\
\label{Appdiffeqn:f:l_2^4}
l_2^{(4)}(t^\prime)
&=&
-540{{t^\prime}}^{10}+304\left( \beta+2 \right) {{t^\prime}}^{9}+4\left(1437\bt^2 -496\,
\beta-496\right) {{t^\prime}}^{8}
\nonumber\\&&
+36(\bt+2) \left( 185{\beta}^{2}-44\bt-44\right) {{t^\prime}}^{7}
\nonumber\\&&
-\left( 4593\bt^4-
30152\,{\beta}^{3}-26584\,{\beta}^{2}+7136\,\beta+3568
 \right) {{t^\prime}}^{6}
\nonumber\\&&
-4(\bt+2) \left(2751\,{
\beta}^{4}-4658\bt^3-4442\bt^2+432\bt+216\right) {{t^\prime}}^{5}
\nonumber\\&&
-3\,{\beta}^{2} \left(1227\,{\beta}^{4}+10456\bt^3+3944\,{\beta}^{2}-13024\beta-6512 \right) {{t^\prime}}^{4}
\nonumber\\&&
+36{\beta}^{4}(\bt+2) \left( 53\,{\beta}^{2}-302\bt-302 \right) {{t^\prime}}^{3}
+3{\beta}^{6} \left(275\bt^2+1632\bt+ 1632\right) {{t^\prime}}^{2}
\nonumber\\&&
-108{\beta}^{8} \left( \beta+2 \right) {t^\prime}+{\beta}^{10}.
\eea
\end{small}

\subsection{$\kappa_3({t^\prime})$}
The correction terms $g_i({t^\prime})$ for $i=1,2,3,4$ satisfy the differential equations
\begin{small}
\bea
\label{Appdiffeqn:g1}
\frac{\big(({t^\prime}+\bt)^2+4{t^\prime}\big)^{11/2}}{6{t^\prime}(1+\bt)}\frac{dg_1(t^\prime)}{dt^\prime}&=&
l_3^{(1)}(t^\prime),\\
\label{Appdiffeqn:g2}
\frac{\big(({t^\prime}+\bt)^2+4{t^\prime}\big)^{7}}{6\N{t^\prime}(1+\bt)}\frac{dg_2(t^\prime)}{dt^\prime}&=&
l_3^{(2)}(t^\prime),\\
\label{Appdiffeqn:g3}
\frac{\big(({t^\prime}+\bt)^2+4{t^\prime}\big)^{17/2}}{6{t^\prime}(1+\bt)}\frac{dg_3(t^\prime)}{dt^\prime}&=&
l_3^{(3)}(t^\prime),\\
\label{Appdiffeqn:g4}
\frac{\big(({t^\prime}+\bt)^2+4{t^\prime}\big)^{10}}{6\N{t^\prime}(1+\bt)}\frac{dg_4(t^\prime)}{dt^\prime}&=&
l_3^{(4)}(t^\prime).
\eea
\end{small}
The terms $l_3^{(i)}(t^\prime)$, $i=1,2,3,4$ are given by
\begin{small}
\bea
\label{Appdiffeqn:g:l_3^1}
l_3^{(1)}(t^\prime)
&=&
-{{t^\prime}}^{6}+9\left( \beta+2 \right) {{t^\prime}}^{5}+3\left( 5\,{\beta}
^{2}-6\,\beta-6 \right) {{t^\prime}}^{4}-\left( 15\,{\beta}^{2}+2\,\beta+2
\right)  \left( \beta+2 \right) {{t^\prime}}^{3}
\nonumber\\&&
-2\left( 15\,{\beta}^{2}+4\beta+4 \right) {{t^\prime}}^{2}{\beta}^{2}-6\left( \beta+2 \right) {t^\prime}{
\beta}^{4}+4{\beta}^{6}+ t^\prime{\N}^2\Big[ -{{t^\prime}}^{5}+\left( \beta+2
\right) {{t^\prime}}^{4}
\nonumber\\&&
+7{{t^\prime}}^{3}{\beta}^{2}
+5\left( \beta+2 \right) {{t^\prime}}^{2}{\beta}^{2}-2\left( {\beta}^{2}-4\,\beta-4 \right) {{t^\prime}}{\beta}^{2}-2\left( \beta+2 \right){\beta}^{4}\Big].\\
\label{Appdiffeqn:g:l_3^2}
l_3^{(2)}(t^\prime)
&=&
16\,{{t^\prime}}^{8}-61\left( \beta+2 \right) {{t^\prime}}^{7}-8\left( 28\bt^2-3\beta-3\right) {{t^\prime}}^{6}-7\left( \beta+2 \right)
\left( 3\,{\beta}^{2}+4\,\beta+4 \right) {{t^\prime}}^{5}
\nonumber\\&&
+2\left( 189\,{\beta}^{4}-76\bt^3-92\,{\beta}^{2}-32\,\beta-16 \right) {{t^\prime}}^{4}+7{\beta}^{2} \left( \beta+2 \right)  \left( 39\,{\beta}^{2}-4\,
\beta-4 \right) {{t^\prime}}^{3}
\nonumber\\&&
-4{\beta}^{4} \left( 7\,{\beta}^{2}-90\,
\beta-90 \right) {{t^\prime}}^{2}-27{\beta}^{6} \left( \beta+2 \right){t^\prime}+2{\beta}^{8}.\\
\label{Appdiffeqn:g:l_3^3}
l_3^{(3)}(t^\prime)
&=&
 -80{{t^\prime}}^{10}
+320\left( \beta+2 \right) {{t^\prime}}^{9}+5\left( 253\,{\beta}^{2}-
159\bt-159\right) {{t^\prime}}^{8}
\nonumber\\&&
-\left( \beta+2 \right)  \left( 311
\,{\beta}^{2}-125\,\beta-125 \right) {{t^\prime}}^{7}
-2\left( 1946\,{\beta}
^{4}-1015\,{\beta}^{3}-1062\,{\beta}^{2}-94\,\beta -47\right) {{t^\prime}}^{6}
\nonumber\\&&
-2\left( \beta+2 \right)  \left( 1498\,{\beta}^{4}-413\,{\beta}^{3}
-422\,{\beta}^{2}-18\,\beta-9 \right) {{t^\prime}}^{5}
\nonumber\\&&
+{\beta}^{2} \left(1085\,{\beta}^{4}-7567\,{\beta}^{3}-6843\,{\beta}^{2}+1448\,\beta+724
 \right) {{t^\prime}}^{4}
\nonumber\\&&
+5{\beta}^{4} \left( \beta+2 \right)  \left( 341\,{
\beta}^{2}-395\bt-395 \right) {{t^\prime}}^{3}
+2{\beta}^{6} \left( 97\,{
\beta}^{2}+1230\,\beta+1230 \right) {{t^\prime}}^{2}
\nonumber\\&&
-142\,{\beta}^{8} \left(
\beta+2 \right) T+4{\beta}^{10}
+{t^\prime}{\N}^2\Big[ -50{{t^\prime}}^{9}+70\left( \beta+2 \right) {{t^\prime}}^{8}
+2\left(325\bt^2 -97\,\beta-97\right) {{t^\prime}}^{7}
\nonumber\\&&
+2\left(\beta+2 \right)  \left( 305\,{\beta}^{2}-53\,\beta-53 \right) {{t^\prime}}^{6}
-2\left(380\bt^4 -1535\,{\beta}^{3}-1407\,{\beta}^{2}+256\,\beta+128\right) {{t^\prime}}^{5}
\nonumber\\&&
-2\left( \beta+2 \right)  \left( 700\,{
\beta}^{4}-981\,{\beta}^{3}-949\,{\beta}^{2}+64\,\beta+32 \right) {{t^\prime}}^
{4}
\nonumber\\&&
-2{\beta}^{2} \left( 215\,{\beta}^{4}+1982\,{\beta}^{3}+918\,{
\beta}^{2}-2128\,\beta-1064 \right) {{t^\prime}}^{3}
\nonumber\\&&
+2{\beta}^{4} \left(
\beta+2 \right)  \left( 125\,{\beta}^{2}-684\,\beta-684 \right) {{t^\prime}}^{2}
+22{\beta}^{6} \left( 5\,{\beta}^{2}+28\,\beta+28 \right) {{t^\prime}}
-10{\beta}^{8} \left( \beta+2 \right) \Big].\nonumber\\\\
\label{Appdiffeqn:g:l_3^4}
l_3^{(4)}(t^\prime)
&=&
1080\,{{t^\prime}}^{12}-3460\left( \beta+2 \right) {{t^\prime}}^{11}-32 \left(619\bt^2 -338\beta-338 \right) {{t^\prime}}^{10}
\nonumber\\&&
-2\left( \beta+2
\right)  \left( 2045\,{\beta}^{2}+204\,\beta+204 \right) {{t^\prime}}^{9}
+20\left(3309\,{\beta}^{4}-4120\bt^3-4040\bt^2+160\,\beta+80\right) {{t^\prime}}^{8}
\nonumber\\&&
+\left( \beta+2 \right)  \left( 81555\,{
\beta}^{4}-62404\,{\beta}^{3}-62084\,{\beta}^{2}+640\,\beta+320
\right) {{t^\prime}}^{7}
\nonumber\\&&
-80{\beta}^{2} \left(57\,{\beta}^
{4}-3845\,{\beta}^{3}-2401\,{\beta}^{2}+2888\bt+1444 \right) {{t^\prime}}^{6}
\nonumber\\&&
-3{\beta}^{2} \left( \beta+2 \right)  \left( 18777\,{\beta}^{4}-58336\,{\beta}^{3}
-50528\,{\beta}^{2}+15616\,\beta+7808 \right) {{t^\prime}}^{5}
\nonumber\\&&
-6{\beta}^{4}\left(3915\,{\beta}^{4}+24728\bt^3-2792\,{\beta}^{2}-
55040\,\beta-27520 \right) {{t^\prime}}^{4}
\nonumber\\&&
+45{\beta}^{6} \left( \beta+2 \right)
\left( 109\,{\beta}^{2}-1172\,\beta-1172 \right) {{t^\prime}}^{3}
+4{\beta}^{8} \left( 725\,{\beta}^{2}+4104\bt+4104 \right) {{t^\prime}}^{2}
\nonumber\\&&
-275{\beta}^{10} \left( \beta+2 \right){t^\prime}+2{\beta}^{12}.
\eea
\end{small}

\bibliography{AF_Paper_Ref}

\end{document}